\documentclass[aps,pra,11pt,onecolumn,nofootinbib,tightenlines,superscriptaddress]{revtex4-2}

\usepackage{mathrsfs}

\usepackage{graphicx}
\usepackage{comment}
\usepackage{amsmath}
\usepackage{amsfonts}
\usepackage{amssymb}
\usepackage{braket}
\usepackage{amsthm}

\usepackage{enumerate}

\usepackage{natbib}
\usepackage{appendix}
\usepackage{multirow}
\usepackage{array}
\usepackage{physics}
\usepackage{mathtools}
\usepackage{xcolor}
\usepackage{wasysym}
\usepackage{needspace}
\usepackage{etoolbox}

\makeatletter
\def\label#1{\@bsphack
  \begingroup
  \UseHookWithArguments{label}{1}{#1}%
  \protected@write\@auxout{}%
    {\string\newlabel{#1}{{\@currentlabel}{\thepage}%
    {\@currentlabelname}{\@currentHref}{\@kernel@reserved@label@data}}}%
  \endgroup
  \@esphack}
\let\ltx@label\label
\makeatother
\usepackage[colorlinks=true,urlcolor=blue,citecolor=blue,linkcolor=blue]{hyperref}

\usepackage[paperwidth=225mm,paperheight=307mm,centering,hmargin=2.45cm,vmargin=2.5cm]{geometry}

\makeatletter
\patchcmd{\@bibdataout@aps}{author="08"}{author="48"}{}{}
\patchcmd{\@bibdataout@aps}{author="08"}{author="48"}{}{}
\makeatother

\DeclareMathOperator*{\argmax}{arg \,max}

\DeclareMathOperator{\KM}{KM}

\newtheorem{theorem}{Theorem}
\newtheorem{lemma}[theorem]{Lemma}
\newtheorem{proposition}[theorem]{Proposition}
\newtheorem{corollary}[theorem]{Corollary}
\newtheorem{definition}[theorem]{Definition}
\newtheorem{remark}{Remark}

\DeclareMathOperator{\supp}{supp}

\def \d {\mathrm{d}}

\newcommand{\widebar}[1]{\overline{\mkern-4mu#1\mkern-1mu}}

\allowdisplaybreaks

\newdimen\savedTOChfuzz

\begin{document}

\title{\Large Log-Euclidean Rényi Conditional Mutual Information
}

\begin{abstract}
Quantum conditional mutual information (QCMI) is a fundamental measure of conditional
independence, with central applications throughout quantum information theory
and many-body physics. Despite its importance, finding a fully quantum
R\'enyi generalization that satisfies all desirable structural properties has
remained an open question. In this work, we introduce a quantum R\'enyi
conditional mutual information possessing these desirable properties: it is
monotone under local quantum channels, additive under tensor products,
monotone in the R\'enyi parameter, and converges to the QCMI in the limit $\alpha \to 1$. We establish bounds on this quantity in terms of relative entropies to the Petz-recovered state, showing that it provides a distinct measure of approximate conditional independence. As a consequence, we derive recovery-based upper and lower bounds on the QCMI in terms of quantum generalizations of the Kullback–Leibler divergence, improving upon the best known results in certain regimes. Finally,  we establish bounds in terms of relative entropies
to the sets of Markov and separable states and derive R\'enyi chain rules, which
appear to be novel even in the classical setting.
\end{abstract}

\author{Roberto Rubboli}
\email{ror@math.ku.dk}
\affiliation{Department of Mathematical Sciences, University of Copenhagen, Universitetsparken 5, 2100 Denmark}
\author{Amir Arqand}
\affiliation{Institute for Quantum Computing and Department of Physics and Astronomy,
University of Waterloo, Waterloo, Ontario, Canada, N2L 3G1}
\author{Mark M. Wilde}
\affiliation{School of Electrical and Computer Engineering,
Cornell University, Ithaca, New York 14850, USA}

\maketitle

\savedTOChfuzz=\hfuzz
\hfuzz=30pt
\tableofcontents
\hfuzz=\savedTOChfuzz

\section{Introduction}
\label{sec: intro}
\subsection{Background}

The quantum conditional mutual information (QCMI) is a fundamental measure of conditional independence in quantum information theory \cite{lieb1973fundamental,LiebRuskai1973Proof}, with various operational meanings \cite{DevetakYard2008ExactCost,YardDevetak2009OptimalSourceCoding,QiSharmaWilde2018,Sharma2020,Berta2018,Berta2018PRL}, and it has found applications across several areas of physics \cite{KatoBrandao2019Quantum,Kim2023,anshu2023one}. For a tripartite state $\rho_{ABC}$, it is defined as
\begin{equation}
I(A: B|C)_{\rho}\coloneqq H(AC)_{\rho}+H(BC)_{\rho}-H(C)_{\rho}-H(ABC)_{\rho} \,,
\end{equation}
where $H(AC)\coloneqq -\Tr[\rho_{AC}\log \rho_{AC}]$, with similar definitions for the other marginal entropies.
The QCMI quantifies the correlations between systems $A$ and $B$ that remain after conditioning on a third system $C$.

The QCMI is closely related to the structure of quantum Markov chains. In particular, the condition $I(A:B|C)_\rho=0$ holds if and only if $\rho_{ABC}$ is a quantum Markov chain in the order $A-C-B$ \cite{hayden04CMI}, in the sense that the full state can be recovered from its marginal $\rho_{AC}$ by acting only on the conditioning system $C$. Equivalently, there exists a quantum channel $\mathcal{R}_{C\to BC}$ such that
\begin{align}
    \mathcal{R}_{C\to BC}(\rho_{AC})=\rho_{ABC}.
\end{align}
The conditional mutual information also provides a quantitative measure of approximate conditional independence. In particular, a series of works showed that if the QCMI is small, then $\rho_{ABC}$ can be approximately recovered from its marginal $\rho_{AC}$ by acting only on the conditioning system $C$~\cite{fawzi2015quantum,brandao2015,Wilde2015Recoverability,sutter2016universal,junge2015universal,sutter2017multivariate}. More precisely, there exists a recovery channel $\mathcal{R}_{C\to BC}$, which can be chosen as a rotated Petz recovery map, such that the reconstructed state $\mathcal{R}_{C\to BC}(\rho_{AC})$ has high fidelity with $\rho_{ABC}$~\cite{Wilde2015Recoverability,junge2015universal}.

Although the QCMI is well understood, its R\'enyi generalizations remain subtle in the quantum setting~\cite{berta2015renyi}. A satisfactory R\'enyi QCMI should reduce to a classical R\'enyi conditional mutual information in the commuting case, obey data processing under local quantum channels on \(A\) and \(B\), and be additive under tensor products, reflecting the extensivity of correlations for independent systems. The main obstruction is non-commutativity, which prevents a direct extension of the classical definition. Consequently,
prior to our work, no R\'enyi-type QCMI was known to satisfy all of these desirable properties~\cite{berta2015renyi}.

\subsection{Summary of contributions}

In this work, we address this problem by introducing several log-Euclidean Rényi generalizations of the QCMI that satisfy the properties mentioned above. We begin with one particular variant, defined for a tripartite state $\rho_{ABC}$ and $\alpha \in (0,1)\cup (1,\infty)$ by
\begin{align}
I_\alpha^{\flat}(A:B|C)_\rho
\coloneqq \frac{1}{\alpha - 1}
\log \Tr\!\left[
\exp\!\big(
\alpha \log \rho_{ABC}
+ (1-\alpha)\big(
\log \rho_{AC}
+ \log \rho_{BC}
- \log \rho_C
\big)
\big)
\right].
\end{align}
For states that are not full rank, this expression is understood through the limiting definitions in Section~\ref{sec: notation}. 
We show that it satisfies the data-processing inequality for all $\alpha \in (0,1)$, as well as additivity and monotonicity in the parameter $\alpha$ for all $\alpha > 0$. Moreover, it converges to the standard conditional mutual information in the limit $\alpha \to 1$. 

We further show that this quantity provides a faithful
characterization of exact conditional independence, vanishing precisely when the
underlying tripartite state is a quantum Markov chain. Moreover, it provides
a quantitative measure of approximate conditional independence through the
following bounds. For $\alpha\in (0,1)$, we obtain 
\begin{equation}
I^\flat_\alpha(A:B|C)_\rho
\geq
\max\!\left\{
\int_{\mathbb{R}}dt\, \beta_0(t)\,
D_{\alpha,1-\alpha}\!\left(
\rho_{ABC}\middle\|\mathcal{R}^{[t]}_{C\to BC}(\rho_{AC})
\right),
D_\alpha^{\mathbb{M}}\!\left(
\rho_{ABC}\middle\|
\int_{\mathbb{R}}dt\, \beta_0(t)\,
\mathcal{R}^{[t]}_{C\to BC}(\rho_{AC})
\right)
\right\},
\end{equation}
while, for $\alpha>1$ and full-rank $\rho_{ABC}$, we obtain
\begin{equation}
I^\flat_\alpha(A:B|C)_\rho
\leq
\int_{\mathbb{R}}dt\, \beta_0(t)\,
D_{\alpha,\alpha-1}\!\left(
    \rho_{ABC}
    \middle\|
    \mathcal{R}^{[t]}_{C\to BC}(\rho_{AC})
\right)\, .
\end{equation}
The bounds use the rotated Petz recovery map $\mathcal R^{[t]}_{C\to BC}$~\cite{Wilde2015Recoverability} and the probability density defined in~\eqref{eq: rotated recovery} and~\eqref{eq:beta}, respectively.
Moreover, $D^\mathbb{M}_{\alpha}$ denotes the measured R\'enyi relative entropy while $D_{\alpha,1-\alpha}$ and $D_{\alpha,\alpha-1}$  are the minimal and
maximal R\'enyi relative entropies for $\alpha\in[0,1/2]$ and $\alpha\geq2$,
respectively~\cite{hayashi2016correlation, mosonyi2024geometric,rubboli2026maximal}; see
Section~\ref{sec: notation} for further details.

In the limit $\alpha\to1$, the lower bound involving the measured relative
entropy recovers the result of~\cite{sutter2017multivariate}, whereas the
remaining bounds yield the following new  estimates for the QCMI 
\begin{equation}
\begin{aligned}
\int_{\mathbb{R}}dt\, \beta_0(t)\,
D^\star\!\left(
    \rho_{ABC}
    \middle\|
    \mathcal{R}^{[t]}_{C\to BC}(\rho_{AC})
\right)\,
\leq I(A:B|C)_\rho
\leq
\int_{\mathbb{R}}dt\, \beta_0(t)\,
D^{\text{\normalfont\leftmoon}}\!\!\left(
    \rho_{ABC}
    \middle\|
    \mathcal{R}^{[t]}_{C\to BC}(\rho_{AC})
\right)\, .
\end{aligned}
\end{equation}
Here, $D^\star$ and $D^{\text{\normalfont\leftmoon}}\!$ denote quantum
extensions of the classical Kullback--Leibler (KL) divergence, with
$D^\star$ known as Keyl's rate function~\cite{keyl2006estimation,audenaert13_alphaz} (see Section~\ref{sec: notation} for more details). In particular, these bounds are tight classically, and the lower bound always improves upon the fidelity-based bound
of~\cite{junge2015universal}. In certain cases, it is also tighter than the
bound based on measured relative entropy, although neither of these two bounds
dominates the other in general. Our upper bound resembles the one derived using continuity bounds in~\cite{berta2015renyi,fawzi2015quantum}, which holds for arbitrary recovery channels, but contains no dimension-dependent factors.

For the QCMI, the standard formula admits equivalent variational reformulations involving optimizations over suitable sets of states. Although these formulations coincide with the usual definition of conditional mutual information~\cite{berta2015renyi}, they generally cease to be equivalent in the R\'enyi case. Motivated by this distinction, we introduce a second quantity, defined by
\begin{align}
I_\alpha^{\flat,*}(A:B|C)_\rho
\coloneqq \inf_{\sigma_{BC},\,\eta_{AC}} \sup_{\omega_C}
\frac{1}{\alpha - 1}
\log \Tr\!\left[
\exp\!\big(
\alpha \log \rho_{ABC}
+ (1-\alpha)\big(
\log \eta_{AC}
+ \log \sigma_{BC}
- \log \omega_C
\big)
\big)
\right]\,,
\end{align}
where all optimizations are over states on the corresponding systems.
For $\alpha\in(0,2)$, we show that this quantity is not larger than $I_\alpha^{\flat}$. It obeys data processing for $\alpha\in (0,1)$ and is additive for $\alpha\in[1/2,2]$. It also converges to the standard conditional mutual information in the limit $\alpha \to 1$. In addition, for $\alpha\in(0,1)$, we prove the following bounds
\begin{align}
    \inf_{\sigma_{AB} \in \mathrm{Sep}(A:B)}D^{\mathrm{LOCC}_1(A\to B)}_{\alpha}\left(\rho_{AB}\middle\|\sigma_{AB}\right)\leq  I^{\flat,*}_{\alpha}(A:B|C)_\rho \leq \inf_{\omega_{ABC} \in \mathcal{M}_{A:B|C} }D^\flat_{\alpha}(\rho_{ABC}\| \omega_{ABC}),
\end{align}
where $\mathcal{M}_{A:B|C}$ is the set of Markov chains with order $A-C-B$, the symbol \(D^\flat_{\alpha}\) denotes the log-Euclidean R\'enyi relative
entropy, as defined in Section~\ref{sec: notation}, and
\(D^{\mathrm{LOCC}_1(A\to B)}_{\alpha}\) denotes the R\'enyi relative
entropy measured with respect to one-way LOCC measurements from \(A\) to \(B\).
The latter consists of a local measurement on \(A\), followed by a
local measurement on \(B\) conditioned on the classical outcome of the measurement 
on \(A\).
The corresponding $\alpha=1$ lower bounds were established in~\cite{brandao2011faithful,li2018squashed,li2014relative}, while the upper bound generalizes the result of~\cite{ibinson2008robustness}.

\medskip
Other contributions of our paper include a family of chain rules for the
log-Euclidean R\'enyi conditional mutual informations. More precisely, for
\((\alpha-1)(\beta-1)>0\) and
\(\frac{\beta}{\beta-1}
=
\frac{\alpha}{\alpha-1}
+
\frac{1}{\gamma-1}\), we prove that, whenever \(\gamma \in (0,1)\),
\begin{align}
    I_\alpha^{\flat}(A:B|C)_\rho
    &\leq
    H_\gamma^{\flat,\downarrow}(B|C)_\rho
    -
    H_\beta^{\flat,\downarrow}(B|AC)_\rho,
    \\
    I_\alpha^{\flat}(A:B|C)_\rho
    &\leq
    I_\beta^{\flat,\uparrow \uparrow}(AC:B)_\rho
    -
    I_\gamma^{\flat,\uparrow \uparrow}(C:B)_\rho,
    \\
    I_\alpha^{\flat}(A:B|CD)_\rho
    & \leq
    I_\beta^{\flat}(A:BC|D)_\rho
    -
    I_\gamma^{\flat}(A:C|D)_\rho .
\end{align}
For \(\gamma>1\), the same relations hold with all inequalities reversed.
Here, \(H_\theta^{\flat}\) and \(I_\theta^{\flat}\) denote the
log-Euclidean R\'enyi conditional entropy and mutual information,
respectively; see Section~\ref{sec: notation}. We also establish analogous
chain rules for the corresponding optimized quantities, under appropriate
parameter relations. To the best of our knowledge, these relations appear to be novel even in the classical setting. Moreover, they are tight as the parameters approach one, in which case they recover the known chain rules for the QCMI.

\section{Notation}
\label{sec: notation}

\subsection{Entropies and information}

Throughout our paper, we assume that all Hilbert spaces are finite dimensional, and all logarithms are natural logarithms. For a state $\rho$ and a positive semidefinite operator $\sigma$, we write $\rho\ll\sigma$ when $\supp(\rho)\subseteq\supp(\sigma)$, and $\rho\perp\sigma$ when their supports are orthogonal.
The Umegaki relative entropy is defined for a state $\rho$ and a positive semidefinite operator $\sigma$ as~\cite{Umegaki62}
\begin{align}
    D(\rho\|\sigma) \coloneqq 
    \begin{cases}
\Tr[\rho(\log\rho-\log\sigma)], & \rho\ll\sigma,\\
+\infty, & \text{otherwise},
\end{cases}
\end{align}
The von Neumann conditional entropy is defined as
\begin{align}
    H(A|B)_\rho \coloneqq-D(\rho_{AB}\|I_A\otimes \rho_B) = \sup_{\sigma_B} -D(\rho_{AB}\|I_A\otimes \sigma_B) \,,
\end{align}
where the supremum is over every state $\sigma_B$.
The mutual information is defined as
\begin{align}
    I(A:B)_\rho \coloneqq D(\rho_{AB}\|\rho_A\otimes \rho_B) =\inf_{\sigma_A,\sigma_B}D(\rho_{AB}\|\sigma_A\otimes \sigma_B) \,,
\end{align}
where the infimum is over all states $\sigma_A$ and $\sigma_B$. For $\alpha\in(0,1)\cup(1,\infty)$, the Petz R\'enyi relative entropy is defined as~\cite{petz1986quasi}
\begin{equation}
\widebar{D}_\alpha(\rho\|\sigma)\coloneqq
\begin{cases}
\frac{1}{\alpha-1}\log\Tr[\rho^\alpha\sigma^{1-\alpha}],&\text{if }(\alpha<1\text{ and }\rho\not\perp\sigma)\text{ or }\rho\ll\sigma,\\[1ex]
+\infty,&\text{otherwise}.
\end{cases}
\label{eq:petz-renyi}
\end{equation}
It satisfies the data-processing inequality for $\alpha\in(0,1)\cup(1,2]$. The sandwiched R\'enyi relative entropy is defined for the same parameter domain $\alpha\in(0,1)\cup(1,\infty)$ as~\cite{muller2013quantum,wilde2014strong}
\begin{equation}
\widetilde D_\alpha(\rho\|\sigma)\coloneqq
\begin{cases}
\frac{1}{\alpha-1}\log\Tr\!\left[\left(\sigma^{\frac{1-\alpha}{2\alpha}}\rho\sigma^{\frac{1-\alpha}{2\alpha}}\right)^\alpha\right],&\text{if }(\alpha<1\text{ and }\rho\not\perp\sigma)\text{ or }\rho\ll\sigma,\\[1ex]
+\infty,&\text{otherwise}.
\end{cases}
\label{definition sandwiched}
\end{equation}
All negative powers are evaluated on the corresponding supports.

We use two further families of R\'enyi relative entropies. For $\alpha\in(0,1)$, the reverse sandwiched R\'enyi relative entropy is defined as
\begin{equation}
\begin{aligned}
D_{\alpha,1-\alpha}(\rho\|\sigma)
&\coloneqq\frac{\alpha}{1-\alpha}\widetilde D_{1-\alpha}(\sigma\|\rho)=\begin{cases}
\frac{1}{\alpha-1}\log\Tr\!\left[\left(\rho^{\frac{\alpha}{2(1-\alpha)}}\sigma\rho^{\frac{\alpha}{2(1-\alpha)}}\right)^{1-\alpha}\right],&\text{if }\rho\not\perp\sigma,\\[1ex]
+\infty,&\text{otherwise}.
\end{cases}
\end{aligned}
\label{eq: reverse sandwiched}
\end{equation}
It is the minimal R\'enyi relative entropy for $\alpha\in[0,1/2]$~\cite{hayashi2016correlation,mosonyi2024geometric}, is monotone in $\alpha$ on $(0,1)$~\cite{rubboli2024magic}, and satisfies the data-processing inequality precisely for $\alpha\in[0,1/2]$~\cite{zhang2020wigner}, with $\alpha=0$ defined by a limit.

For $\alpha>1$, we define
\begin{equation}
D_{\alpha,\alpha-1}(\rho\|\sigma)\coloneqq
\begin{cases}
\frac{1}{\alpha-1}\log\Tr\!\left[\left(\rho^{\frac{\alpha}{2(\alpha-1)}}\sigma^{-1}\rho^{\frac{\alpha}{2(\alpha-1)}}\right)^{\alpha-1}\right],&\text{if }\rho\ll\sigma,\\[1ex]
+\infty,&\text{otherwise}.
\end{cases}
\label{eq: moon renyi}
\end{equation}
This family coincides with the maximal R\'enyi relative entropy for $\alpha\geq2$~\cite{rubboli2026maximal}. It is monotone in $\alpha$ on $(1,\infty)$~\cite{rubboli2024magic} and satisfies the data-processing inequality precisely for $\alpha\geq2$~\cite{zhang2020wigner}. At $\alpha=2$, it equals both the geometric R\'enyi relative entropy and the Petz R\'enyi relative entropy: $D_{2,1}=\widebar{D}_2$~\cite{petz1986quasi}.

As $\alpha\downarrow0$, $D_{\alpha,1-\alpha}$ converges to the min-relative entropy, while as $\alpha\to\infty$, $D_{\alpha,\alpha-1}$ converges to the max-relative entropy~\cite{renner2008security,datta2009min}; see also~\cite[Lemma~23]{rubboli2024magic}. These are defined by
\begin{equation}
D_{\min}(\rho\|\sigma)\coloneqq-\log\Tr[P_\rho\sigma],\qquad
D_{\max}(\rho\|\sigma)\coloneqq\inf\{\lambda\in\mathbb R:\rho\leq e^\lambda\sigma\},
\end{equation}
where $P_\rho$ is the projector onto the support of $\rho$.

The limits of $D_{\alpha,1-\alpha}$ and $D_{\alpha,\alpha-1}$ as $\alpha\to1^-$ and $\alpha\to1^+$, respectively, yield two quantum extensions of the classical relative entropy that generally differ for noncommuting states.
 To introduce these quantities, we first express the Umegaki relative entropy using projectors onto the subspaces spanned by the leading eigenvectors of $\rho$. Write $\rho=\sum_{j=1}^d p_j|j\rangle\langle j|$, with $p_1\geq\cdots\geq p_d\geq p_{d+1}=0$, and set $P_k=\sum_{j=1}^k|j\rangle\langle j|$. For $\rho\ll\sigma$, the identity $\rho=\sum_k(p_k-p_{k+1})P_k$ gives
\begin{equation}
D(\rho\|\sigma)=\Tr[\rho\log\rho]-\sum_{k=1}^d(p_k-p_{k+1})\Tr[P_k(\log\sigma)P_k].
\label{eq: umegaki projector form}
\end{equation}
Moving the logarithm outside each projection gives Keyl's rate function~\cite{keyl2006estimation,grootveld2026tomography}:
\begin{equation}
D^\star(\rho\|\sigma)\coloneqq
\begin{cases}
\Tr[\rho\log\rho]-\sum_{k=1}^d(p_k-p_{k+1})\Tr[\log(P_k\sigma P_k)],&\text{if }\supp(\rho)\cap\ker(\sigma)=\{0\},\\
+\infty,&\text{otherwise}.
\end{cases}
\label{eq: keyl rate}
\end{equation}
Replacing $\sigma$ by $\sigma^{-1}$ and changing the minus sign to a plus in the second term gives
\begin{equation}
D^{\text{\normalfont\leftmoon}}\!(\rho\|\sigma)\coloneqq
\begin{cases}
\Tr[\rho\log\rho]+\sum_{k=1}^d(p_k-p_{k+1})\Tr[\log(P_k\sigma^{-1}P_k)],&\text{if }\rho\ll\sigma,\\
+\infty,&\text{otherwise}.
\end{cases}
\label{eq: moon rate}
\end{equation}
 Terms with $p_k=p_{k+1}$ are omitted, and the expressions are therefore independent of the eigenbasis within degenerate eigenspaces.
The inverse $\sigma^{-1}$ in~\eqref{eq: moon rate} is taken on $\supp(\sigma)$.
\par
The quantity $D^\star$ is Keyl's large-deviation rate function for quantum state estimation~\cite{keyl2006estimation,grootveld2026tomography}. Its representation-theoretic framework also appears in quantum measurement and moment-map estimation~\cite{notzel2016measurements,botero2021moment,franks2023minimal}, transformations of pairs of states and hypothesis testing~\cite{lipka2024dichotomies,ji2025beyond,hayashi2026reverse}, and quantum thermodynamics~\cite{watanabe2026reliability}. It has also been used to derive single-letter lower bounds on composite-hypothesis-testing error exponents~\cite{beigi2026additivity}.
The two divergences defined above also arise as limits of the R\'enyi divergences in~\eqref{eq: reverse sandwiched} and~\eqref{eq: moon renyi}, respectively~\cite[Theorems~2 and~3]{audenaert13_alphaz}:
\begin{equation}
\lim_{\alpha\to1^-}D_{\alpha,1-\alpha}(\rho\|\sigma)=D^\star(\rho\|\sigma),\qquad
\lim_{\alpha\to1^+}D_{\alpha,\alpha-1}(\rho\|\sigma)=D^{\text{\normalfont\leftmoon}}\!(\rho\|\sigma).
\label{eq: recovery divergence limits}
\end{equation}
Both limits hold for arbitrary states, with the finite and infinite cases specified in~\eqref{eq: keyl rate} and~\eqref{eq: moon rate}, respectively.

Writing $F(\rho,\sigma)\coloneqq\left\|\sqrt\rho\sqrt\sigma\right\|_1^2$ for the squared fidelity, we have
\begin{equation}
D_{\min}(\rho\|\sigma)\leq-\log F(\rho,\sigma)\leq D^\star(\rho\|\sigma)\leq D(\rho\|\sigma)\leq D^{\text{\normalfont\leftmoon}}\!(\rho\|\sigma)\leq \widebar{D}_2(\rho\|\sigma)\leq D_{\max}(\rho\|\sigma).
\label{eq: keyl moon comparison}
\end{equation}

The two outer chains of inequalities follow from monotonicity in $\alpha$ of $D_{\alpha,1-\alpha}$ and $D_{\alpha,\alpha-1}$~\cite[Lemma~24]{rubboli2024magic} for $\alpha\in(0,1)$ and $\alpha\in (1,\infty)$, respectively. Moreover, $-\log F$ coincides with $D_{\alpha,1-\alpha}$ for $\alpha=1/2$. The inequalities $D^\star(\rho\|\sigma) \leq D(\rho\|\sigma) \leq D^{\text{\normalfont\leftmoon}}\!(\rho\|\sigma)$ follow, for example, by applying the Araki--Lieb--Thirring inequality and taking the limits $\alpha\to1^-$ and $\alpha\to1^+$, respectively; see~\cite[Proposition 6]{lin2015investigating}. Since $D^\star$ is additive and agrees with classical relative entropy, data processing would imply $D^\star\geq D$. The strict inequality $D^\star<D$ for some states therefore rules out data processing.

For full-rank positive definite operators $\rho$ and $\sigma$, and
$\alpha\in(0,1)\cup(1,\infty)$, the log-Euclidean relative entropy is
defined as~\cite{hiai1993golden,audenaert13_alphaz,mosonyi2017strong}
\begin{align}
D_\alpha^\flat(\rho\|\sigma)\coloneqq
\frac{1}{\alpha-1}\log\Tr[\exp\!\left(
  \alpha \log\rho +(1-\alpha)\log\sigma
\right)].
\label{eq:log-Euclidean}
\end{align}
In the case where $\rho$ and $\sigma$ are not full-rank, it is defined as $\lim_{\delta\downarrow0} D_\alpha^\flat(\rho+\delta I\|\sigma+\delta I)$.
 We also set $Q_\alpha^\flat(\rho\|\sigma)=\exp((\alpha-1)D_\alpha^\flat(\rho\|\sigma))$. The log-Euclidean Rényi divergence admits the following variational
characterization~\cite[Theorem~III.6]{mosonyi2017strong}:
\begin{align}
\label{eq: variational form LE}
D_\alpha^\flat(\rho\|\sigma)
=\frac{1}{1-\alpha}\min_{\tau}
\bigl\{\alpha D(\tau\|\rho)+(1-\alpha)D(\tau\|\sigma)\bigr\}\,.
\end{align}
For $\alpha<1$, the minimum is over states supported on the intersection
of the supports of $\rho$ and $\sigma$, and its value is $+\infty$ if
there is no such state. For $\alpha>1$, this formula applies when
$\rho\ll\sigma$, with $\tau$ supported on $\rho$. Thus, each relative
entropy in the objective is finite. This variational form will be
useful in our derivations.
The log-Euclidean conditional entropies are defined as
\begin{align}
    H^{\flat,\downarrow}_\alpha(A|B)_\rho \coloneqq -D_{\alpha}^\flat(\rho_{AB}\|I_A\otimes \rho_B) \, , \quad H^{\flat,\uparrow}_\alpha(A|B)_\rho  \coloneqq \sup_{\sigma_B}-D_{\alpha}^\flat(\rho_{AB}\|I_A\otimes \sigma_B) \,,
\end{align}
where the supremum is over every full-rank state $\sigma_B$.
The log-Euclidean mutual informations are defined as
\begin{align}
    I^{\flat,\uparrow\uparrow}_\alpha(A:B)_\rho \coloneqq D_{\alpha}^\flat(\rho_{AB}\|\rho_A\otimes \rho_B) \, , \quad I^{\flat,\downarrow \downarrow}_\alpha(A:B)_\rho \coloneqq \inf_{\sigma_A,\sigma_B}D_{\alpha}^\flat(\rho_{AB}\|\sigma_A\otimes \sigma_B) \,.
\end{align}
Here, the infimum is taken over all full-rank states $\sigma_A$ and $\sigma_B$.

The measured Rényi relative entropy is defined for $\alpha\in(0,1)\cup(1,\infty)$ as~\cite[Eqs.~(3.116)--(3.117)]{Fuchs1996}
\begin{align}
    D^{\mathbb{M}}_\alpha(\rho\|\sigma) \coloneqq \max_{\mathcal{M}} D_\alpha (\mathcal{M}(\rho) \| \mathcal{M}(\sigma) ),
\end{align}
where the maximization is taken over every measurement $\mathcal{M}$, and $D_\alpha(p\|q)\coloneqq \frac{1}{\alpha-1}\log{\sum_x p(x)^\alpha q(x)^{1-\alpha}}$ is the classical R\' enyi relative entropy  (see also \cite{berta2017variational}). 
We also define
\begin{equation}
Q_\alpha^{\mathbb{M}}(\rho \| \sigma) \coloneqq \exp((\alpha-1)D^{\mathbb{M}}_\alpha(\rho\|\sigma)).     
\end{equation}
The latter quantity admits the following variational form for all $\alpha \in (0, 1)$~\cite[Lemma~3 and Theorem~4]{berta2017variational}:
\begin{align}
\label{eq: variational QalphaM}
    Q_\alpha^{\mathbb{M}}(\rho \| \sigma) = \inf_{\omega > 0} \left\{ \alpha \, \Tr\!\left[\rho \, \omega^{\frac{\alpha-1}{\alpha}}\right] + (1 - \alpha) \, \Tr\!\left[\sigma \, \omega\right] \right\}.
\end{align}

We also define
\begin{align}
    D^{\mathrm{LOCC}_1(A\to B)}_{\alpha}\!\left(\rho_{AB} \middle\| \sigma_{AB}\right)
    \coloneqq \max_{\mathcal{M} \in \mathrm{LOCC}_1(A\to B)}
    D_\alpha\!\left(\mathcal{M}(\rho_{AB}) \middle\| \mathcal{M}(\sigma_{AB})\right).
\end{align}
Here, \(\mathrm{LOCC}_1(A\to B)\) denotes the class of measurements realizable by one-way local operations and classical communication, whereby a measurement is performed on \(A\) and subsequently a measurement on \(B\) conditioned on the outcome obtained on \(A\).

We next recall the recovery maps that will be used in our recovery bounds. For a tripartite state $\rho_{ABC}$ and $t\in\mathbb R$, the rotated Petz recovery map is defined as~\cite{Wilde2015Recoverability}
\begin{align}
 \mathcal R^{[t]}_{C\to BC}(X_C)
 &\coloneqq\rho_{BC}^{\frac{1+it}{2}}
 \left(\rho_C^{-\frac{1+it}{2}}X_C
 \rho_C^{-\frac{1-it}{2}}\otimes I_B\right)
 \rho_{BC}^{\frac{1-it}{2}}.
 \label{eq: rotated recovery}
\end{align}
Inverse powers are taken on the corresponding supports. The map is trace preserving on $\supp(\rho_C)$ and admits a channel extension to the whole input space. Its action on $\rho_{AC}$ is independent of this extension; the identity channel on $A$ is implicit. The recovery bounds will also contain the probability density 
\begin{align}\label{eq:beta}
\beta_0(t)\coloneqq \frac{\pi}{2(\cosh(\pi t)+1)}\,.
\end{align}

\subsection{Log-Euclidean conditional mutual information}

We now introduce two variants of the log-Euclidean conditional mutual information (LE-CMI).
\begin{definition}
Let $\rho_{ABC}$ be a full-rank quantum state, and let $\alpha\in(0,1)\cup(1,\infty)$. We define
\begin{align}
I_\alpha^\flat(A:B|C)_\rho
&\coloneqq \frac{1}{\alpha-1}
\log \Tr\!\left[
 \exp\!\left(
\alpha \log \rho_{ABC}
+(1-\alpha)
(\log \rho_{AC}+\log \rho_{BC}-\log \rho_C)
\right)\right],
\\
I_\alpha^{\flat,*}(A:B|C)_\rho
&\coloneqq
\inf_{\sigma_{BC},\eta_{AC}}\sup_{\omega_C}
\frac{1}{\alpha-1}
\log \Tr\!\left[
\exp\!\left(
\alpha \log \rho_{ABC}
+(1-\alpha)
(\log \eta_{AC}+\log \sigma_{BC}-\log \omega_C)
\right)\right],
\end{align}
where the latter optimizations are over full-rank states.
\end{definition}
For states that are not full rank, we use the following limiting
definitions. For $X\in\{ABC,AC,BC,C\}$, let
$\rho_X(\varepsilon)=(1-\varepsilon)\rho_X+
\varepsilon I_X/d_X$. We define
\begin{align}
 I_\alpha^\flat(A:B|C)_\rho
 &:=\lim_{\varepsilon\downarrow0}
 D_\alpha^\flat\left(
 \rho_{ABC}\middle\|
 \exp\!\left(
 \log\rho_{AC}(\varepsilon)
 +\log\rho_{BC}(\varepsilon)
 -\log\rho_C(\varepsilon)
 \right)\right),
 \label{eq:flat-LE-renyi-CMI-limit}\\
 I_\alpha^{\flat,*}(A:B|C)_\rho
 &:=\lim_{\varepsilon\downarrow0}
 I_\alpha^{\flat,*}(A:B|C)_{\rho(\varepsilon)}.
\end{align}

For compactness, it will be useful to introduce the following auxiliary function. Let $\rho_{ABC}$, $\eta_{AC}$, $\sigma_{BC}$, and $\omega_C$ be full-rank states, and let $\alpha\in(0,1)\cup(1,\infty)$. We define
\begin{align}
\Delta_\alpha^\flat(\rho_{ABC},\eta_{AC},\sigma_{BC},\omega_C)
&\coloneqq\frac{1}{\alpha-1}\log\Tr\!\big[\exp\!\big(
\alpha\log\rho_{ABC}+(1-\alpha)
(\log\eta_{AC}+\log\sigma_{BC}-\log\omega_C)\big)\big].
\label{eq:Delta-alpha-quantity}
\end{align}
In the case where $\rho_{ABC}$ is not full rank, it is defined as $\lim_{\varepsilon\downarrow0}
   \Delta_\alpha^\flat(\rho_{ABC}(\varepsilon),\eta_{AC},\sigma_{BC},\omega_C)$.

The argument of the trace appearing in the definitions of the log-Euclidean relative entropy and the LE-CMI has a form reminiscent of the log-Euclidean Fréchet mean~\cite{arsigny2007geometric}, from which the name has its roots.

\begin{remark}
    The LE-CMI can be written in terms of the log-Euclidean relative entropy explicitly as 
    \begin{align}
I_\alpha^\flat(A:B|C)_\rho & =D_\alpha^\flat(\rho_{ABC}\|\exp(\log \rho_{AC}+\log \rho_{BC}-\log \rho_C))\,, \\
         I_\alpha^{\flat,*}(A:B|C)_\rho & =\inf_{\sigma_{BC},\eta_{AC}}\sup_{\omega_C} D_\alpha^\flat(\rho_{ABC}\|\exp(\log \eta_{AC}+\log \sigma_{BC}-\log \omega_C)) \,,
    \end{align}
    with the first formula resembling that identified in \cite[Section~2]{LiebRuskai1973Proof} for conditional mutual information. For states that are not full rank, these expressions are understood through the limiting definitions above.

    For full-rank states, the matrix-valued quantum conditional mutual information (MCMI), first considered in~\cite{kuwahara2020clustering} and subsequently termed as such in~\cite{capel2025quasi}, is defined by
$\log \rho_{ABC}+\log \rho_C-\log \rho_{AC}-\log \rho_{BC}$.
Its operator norm admits the following representation
\begin{align}
&\left\|
\log \rho_{ABC}
+\log \rho_C
-\log \rho_{AC}
-\log \rho_{BC}
\right\|_\infty
\nonumber\\
&=
\max\big\{
D_\infty^\flat\!\left(
\rho_{ABC}\middle\|
\exp\!\left(
\log \rho_{AC}
+\log \rho_{BC}
-\log \rho_C
\right)
\right), D_\infty^\flat\!\left(
\exp\!\left(
\log \rho_{AC}
+\log \rho_{BC}
-\log \rho_C
\right)
\middle\|
\rho_{ABC}
\right)
\big\}.
\end{align}
In particular, the first term in the maximum coincides with
$\lim_{\alpha\to\infty} I_\alpha^\flat(A:B|C)_\rho$.
\end{remark}

\section{Variational expressions}
In this section, we derive variational expressions for the log-Euclidean conditional entropies, mutual information, and conditional mutual information. These forms will serve as a central tool for the subsequent analysis.

In all variational expressions below, the auxiliary state $\tau$ is
restricted to $\operatorname{supp}\tau\subseteq\operatorname{supp}\rho$.
The objective functions are then finite and continuous on this compact
set. For states that are not full rank, the expressions below are
understood through the limiting definitions above.

\subsection{Variational expressions for conditional entropies and mutual information}

Let us start by deriving variational expressions for the conditional entropies and mutual informations. The result for conditional entropies has been partly derived in~\cite[Lemma 7]{rubboli2026strong}.
\begin{lemma}
\label{lem:LE_cond}
Let $\rho_{AB}$ be a quantum state. Then, for all $\alpha\in (0,1)$, 
    \begin{align}\label{eq:LE_cond}
    H_\alpha^{\flat,\downarrow}(A|B)_\rho & =\sup_{\tau_{AB}}\left\{\frac{\alpha}{\alpha-1}D(\tau_{AB}\|\rho_{AB})-D(\tau_B\|\rho_B)+H(A|B)_\tau\right\},\\
        H_\alpha^{\flat,\uparrow}(A|B)_\rho & =\sup_{\tau_{AB}}\left\{\frac{\alpha}{\alpha-1}D(\tau_{AB}\|\rho_{AB})+H(A|B)_\tau\right\}\,,
    \end{align}
   where the optimization is over every state $\tau_{AB}\ll\rho_{AB}$. Moreover, the same variational formulas remain valid for all \(\alpha>1\), with the supremum replaced by an infimum.
\end{lemma}
 \begin{proof}
 Let us start with the case $\alpha>1$.
    First, observe that $H(A|B)_\tau=D(\tau_B\|\rho_B)-D(\tau_{AB}\|I_A\otimes\rho_B)$.
    Using the variational form for the log-Euclidean relative entropy in~\eqref{eq: variational form LE} we obtain
    \begin{align}
        H_\alpha^{\flat,\downarrow}(A|B)_\rho&=-D_\alpha^\flat(\rho_{AB}\|I_A\otimes\rho_B) \\
        &=-\sup_{\tau_{AB}}\left\{\frac{\alpha}{1-\alpha}D(\tau_{AB}\|\rho_{AB})+D(\tau_{AB}\|I_A\otimes\rho_B)\right\} \\
        &=-\sup_{\tau_{AB}}\left\{\frac{\alpha}{1-\alpha}D(\tau_{AB}\|\rho_{AB})+D(\tau_B\|\rho_B)-H(A|B)_\tau\right\} \\
        &=\inf_{\tau_{AB}}\left\{\frac{\alpha}{\alpha-1}D(\tau_{AB}\|\rho_{AB})-D(\tau_B\|\rho_B)+H(A|B)_\tau\right\},
    \end{align}
       For the optimized conditional entropy, we similarly have
        \begin{align}
            H_\alpha^{\flat,\uparrow}(A|B)_\rho&=-\inf_{\sigma_B}D_\alpha^\flat(\rho_{AB}\|I_A\otimes\sigma_B) \\
            &=-\inf_{\sigma_B}\sup_{\tau_{AB}}\left\{\frac{\alpha}{1-\alpha}D(\tau_{AB}\|\rho_{AB})+D(\tau_{AB}\|I_A\otimes\sigma_B)\right\} \\
            &=-\inf_{\sigma_B}\sup_{\tau_{AB}}\left\{\frac{1}{1-\alpha}\Tr[\tau_{AB}\log{\tau_{AB}}]-\frac{\alpha}{1-\alpha}\Tr[\tau_{AB}\log{\rho_{AB}}]-\Tr[\tau_{AB}\log{(I_A\otimes \sigma_B})]\right\} \\
            &=-\sup_{\tau_{AB}}\left\{\frac{\alpha}{1-\alpha}D(\tau_{AB}\|\rho_{AB})+\inf_{\sigma_B}D(\tau_{AB}\|I_A\otimes\sigma_B)\right\} \\
            &=\inf_{\tau_{AB}}\left\{\frac{\alpha}{\alpha-1}D(\tau_{AB}\|\rho_{AB})-\inf_{\sigma_B}D(\tau_{AB}\|I_A\otimes\sigma_B)\right\} \\
            &=\inf_{\tau_{AB}}\left\{\frac{\alpha}{\alpha-1}D(\tau_{AB}\|\rho_{AB})-D(\tau_{AB}\|I_A\otimes\tau_B)\right\} \\
            &=\inf_{\tau_{AB}}\left\{\frac{\alpha}{\alpha-1}D(\tau_{AB}\|\rho_{AB})+H(A|B)_\tau\right\}.
        \end{align}
       Here, in the fourth equality, we invoked Sion's minimax theorem~\cite{sion1958general} to interchange the infimum and the supremum. To verify its assumptions, note that the set of states satisfying $\tau_{AB}\ll\rho_{AB}$ is compact and convex, while the set of full-rank states $\sigma_B$ is convex. Furthermore, the objective function is concave in~$\tau_{AB}$. Since the logarithm is operator concave, the objective function is also convex in $\sigma_B$. Finally, the objective function is continuous in $\tau_{AB}$ and lower semicontinuous in $\sigma_B$ (as follows, for example, from the lower semicontinuity of the relative entropy in its second argument). Therefore, the conditions of Sion's minimax theorem are satisfied, allowing us to exchange the infimum and the supremum.
In addition, the penultimate line holds by noting that the infimum over full-rank $\sigma_B$ equals the value at $\tau_B$, approached by full-rank approximations when $\tau_B$ is singular.

The proof for the case $\alpha<1$ is analogous and simpler, as no minimax theorem is needed.   
    \end{proof}

Similarly, for the mutual information, the following variational expressions hold:
\begin{lemma}\label{cor:LE_mutual}Let $\rho_{AB}$ be a quantum state. Then, 
    for all $\alpha\in (0,1)$, 
    \begin{align}
        \label{eq:LE_mutual}
       I_\alpha^{\flat, \uparrow \uparrow}(A:B)_\rho & =\inf_{\tau_{AB}}\left\{\frac{\alpha}{1-\alpha}D(\tau_{AB}\|\rho_{AB})+D(\tau_A\|\rho_A)+D(\tau_B\|\rho_B)+I(A:B)_\tau\right\},\\ 
       I_\alpha^{\flat,\downarrow \downarrow}(A:B)_\rho & =\inf_{\tau_{AB}}\left\{\frac{\alpha}{1-\alpha}D(\tau_{AB}\|\rho_{AB})+I(A:B)_\tau\right\},
    \end{align}
where the optimization is over every state $\tau_{AB}\ll\rho_{AB}$. Moreover, the same variational formulas remain valid for \(\alpha>1\), with the infimum replaced by a supremum.
\end{lemma}
    \begin{proof}
    The proof is similar to the proof of Lemma~\ref{lem:LE_cond}. Let us start with the case $\alpha >1$.
    First observe that
    \begin{align}
        \label{eq:mutual_dec}
        D(\tau_{AB}\|\rho_A \otimes\rho_B)=I(A:B)_\tau + D(\tau_A\|\rho_A) + D(\tau_B\|\rho_B).
    \end{align}
    The variational form in~\eqref{eq: variational form LE} gives
    \begin{align}
        I_\alpha^{\flat,\uparrow \uparrow}(A:B)_\rho&=D_\alpha^\flat(\rho_{AB}\|\rho_A\otimes\rho_B) \\
        &=\sup_{\tau_{AB}}\left\{\frac{\alpha}{1-\alpha}D(\tau_{AB}\|\rho_{AB})+D(\tau_{AB}\|\rho_A\otimes\rho_B)\right\} \\
        &=\sup_{\tau_{AB}}\left\{\frac{\alpha}{1-\alpha}D(\tau_{AB}\|\rho_{AB})+D(\tau_A\|\rho_A)+D(\tau_B\|\rho_B)+I(A:B)_\tau\right\},
    \end{align}
    where the second line follows from the definition and the last line from~\eqref{eq:mutual_dec}.

        For the optimized quantity, we obtain
        \begin{align}
            I_\alpha^{\flat,\downarrow\downarrow}&(A:B)_\rho\notag \\
            &=\inf_{\eta_A,\sigma_B}D_\alpha^\flat(\rho_{AB}\|\eta_A\otimes\sigma_B) \\
            &=\inf_{\eta_A,\sigma_B}\sup_{\tau_{AB}}\left\{\frac{\alpha}{1-\alpha}D(\tau_{AB}\|  \rho_{AB})+D(\tau_{AB}\|\eta_A\otimes\sigma_B)\right\} \\
            &=\inf_{\eta_A,\sigma_B}\sup_{\tau_{AB}}\left\{\frac{1}{1-\alpha}\Tr[\tau_{AB}\log{\tau_{AB}}]-\frac{\alpha}{1-\alpha}\Tr[\tau_{AB}\log{\rho_{AB}}]-\Tr[\tau_A\log{\eta_A}]-\Tr[\tau_B\log{\sigma_B}]\right\} \\&=\sup_{\tau_{AB}}\left\{\frac{\alpha}{1-\alpha}D(\tau_{AB} \|\rho_{AB})+\inf_{\eta_A,\sigma_B}D(\tau_{AB}\|\eta_A\otimes\sigma_B)\right\} \\
            &=\sup_{\tau_{AB}}\left\{\frac{\alpha}{1-\alpha}D(\tau_{AB}\| \rho_{AB})+D(\tau_{AB}\|\tau_A\otimes\tau_B)\right\} \\
            &=\sup_{\tau_{AB}}\left\{\frac{\alpha}{1-\alpha}D(\tau_{AB}\|\rho_{AB})+I(A:B)_\tau\right\}.
        \end{align}
        Here, in the fourth equality, we applied Sion's minimax theorem, exactly as in the derivation of the conditional entropy variational formula. Indeed, by expanding the relative entropy terms, one readily verifies that the objective function is concave in $\tau_{AB}$ and jointly convex in $(\eta_A,\sigma_B)$.
        The penultimate line holds by noting that the infimum over full-rank $\eta_A,\sigma_B$ equals the value at the marginals of $\tau_{AB}$, approached by full-rank approximations when necessary.
    \end{proof}

\subsection{Variational expressions for conditional mutual informations}

\label{sec: variational from LE-CMI}

\begin{lemma}
\label{lem: variational form}
Let $\rho_{ABC}$ be a quantum state, and let $\eta_{AC}$, $\sigma_{BC}$, and $\omega_C$ be full-rank states. Then, for all $\alpha \in (0,1)$, 
\begin{equation}
    \begin{aligned}
 &\Delta_\alpha^\flat(\rho_{ABC},\eta_{AC},\sigma_{BC},\omega_C) \\
 & \qquad = \inf_{\tau_{ABC}}\Big\{\frac{\alpha}{1-\alpha} D(\tau_{ABC}\|\rho_{ABC})+ D(\tau_{AC}\|\eta_{AC})+D(\tau_{BC}\|\sigma_{BC})-D(\tau_C\|\omega_C)+I(A:B|C)_{\tau}\Big\}.
    \end{aligned}
    \end{equation}
    For $\alpha>1$, the corresponding expression is obtained by replacing the infimum with a supremum.
    \end{lemma}
    
\begin{proof}
We apply the Gibbs variational principle, recalled in Lemma~\ref{lem: Gibbs}, to obtain
    \begin{align}
&-\log \Tr\!\Big[
\exp\!\big(
\alpha \log \rho_{ABC}
+(1-\alpha)(\log \eta_{AC}+\log \sigma_{BC}-\log \omega_C)
\big)
\Big] \notag \\
&\quad=
\inf_{\tau_{ABC}}
\Big\{
\Tr[\tau_{ABC}\log \tau_{ABC}]
-
\Tr\!\Big[\tau_{ABC}\big(
\alpha \log \rho_{ABC}
+(1-\alpha)(\log \eta_{AC}+\log \sigma_{BC}-\log \omega_C)
\big)\Big]
\Big\} \label{eq: equation for minimax}\\
&\quad=
\inf_{\tau_{ABC}}
\Big\{
\alpha \Tr\!\big[\tau_{ABC}(\log \tau_{ABC}-\log \rho_{ABC})\big]
+(1-\alpha)\Tr[\tau_{ABC}\log \tau_{ABC}] \notag \\
&\qquad\qquad\qquad
-(1-\alpha)\Tr\!\big[\tau_{ABC}(\log \eta_{AC}+\log \sigma_{BC}-\log \omega_C)\big]
\Big\}\\
& \quad=\inf_{\tau_{ABC}}
\Big\{
\alpha D(\tau_{ABC}\|\rho_{ABC}) +(1-\alpha)
\big(
D(\tau_{AC}\|\eta_{AC})
+D(\tau_{BC}\|\sigma_{BC})
-D(\tau_C\|\omega_C)
+I(A\!:\!B|C)_\tau
\big)
\Big\},
\end{align}
where the last equality follows from \cite[Lemma 1]{berta2015renyi}.
Therefore,
\begin{equation}
\begin{aligned}
&(1-\alpha)\Delta_\alpha^\flat(\rho_{ABC},\eta_{AC},\sigma_{BC},\omega_C) \\
&\qquad  =
\inf_{\tau_{ABC}}
\Big\{
\alpha D(\tau_{ABC}\|\rho_{ABC}) +(1-\alpha)
\big(
D(\tau_{AC}\|\eta_{AC})
+D(\tau_{BC}\|\sigma_{BC})
-D(\tau_C\|\omega_C)
+I(A\!:\!B|C)_\tau
\big)
\Big\}.
\end{aligned}
\end{equation}
Dividing by \(1-\alpha>0\), we obtain
\begin{equation}
\begin{aligned}
&\Delta_\alpha^\flat(\rho_{ABC},\eta_{AC},\sigma_{BC},\omega_C) \\
&\quad =
\inf_{\tau_{ABC}}
\Big\{
\frac{\alpha}{1-\alpha}D(\tau_{ABC}\|\rho_{ABC})
+D(\tau_{AC}\|\eta_{AC})
+D(\tau_{BC}\|\sigma_{BC})
-D(\tau_C\|\omega_C)
+I(A\!:\!B|C)_\tau
\Big\},
\end{aligned}
\end{equation}
which proves the claim. 

The case $\alpha>1$ is analogous.
For states that are not full rank, these expressions are understood
through the limiting definitions above;~\eqref{eq: variational form LE}
gives the same calculation with $\tau_{ABC}\ll\rho_{ABC}$.
\end{proof}

As a result, we obtain the following variational form for the LE-CMI.

\begin{corollary}
\label{cor: variational form}
Let $\rho_{ABC}$ be a quantum state. Then, for all $\alpha \in (0,1)$, 
    \begin{multline}
        \label{eq: variational QCMI non-optimized}I_\alpha^\flat(A:B|C)_{\rho} = \\
        \inf_{\tau_{ABC}}\Big\{\frac{\alpha}{1-\alpha} D(\tau_{ABC}\|\rho_{ABC})+ D(\tau_{AC}\|\rho_{AC})+D(\tau_{BC}\|\rho_{BC})-D(\tau_C\|\rho_C)+I(A:B|C)_{\tau}\Big\}.
    \end{multline}
    For $\alpha>1$, the corresponding expression is obtained by replacing the infimum with a supremum.
\end{corollary}
\begin{proof}
The result follows by applying Lemma~\ref{lem: variational form} to
the full-rank approximations of the marginals and taking the limit in~\eqref{eq:flat-LE-renyi-CMI-limit}. The
variational objective converges uniformly over
$\tau_{ABC}\ll\rho_{ABC}$, which also establishes the existence of this limit.
\end{proof}

\begin{corollary}
\label{cor: variational form triply optimized}
Let $\rho_{ABC}$ be a quantum state, and let $\alpha\in(0,1)$. Then, 
    \begin{align}
\label{eq: variational triply optimized}
I_\alpha^{\flat,*}(A:B|C)_\rho=\inf_{\tau_{ABC}}\left\{\frac{\alpha}{1-\alpha}D(\tau_{ABC}\|\rho_{ABC})+I(A:B|C)_\tau\right\}.
\end{align}
For $\alpha\in(1,2]$, the corresponding expression is obtained by replacing the infimum with a supremum.
\end{corollary}

\begin{proof}
We first assume that $\rho_{ABC}$ is full rank. For $0<\alpha<1$,
Lemma~\ref{lem: variational form} expresses $\Delta_\alpha^\flat$ as
an infimum over $\tau_{ABC}$. Its objective is convex and continuous
in $\tau_{ABC}$ and concave and continuous in the full-rank state
$\omega_C$. The state space of $\tau_{ABC}$ is compact and convex.
Sion's minimax theorem therefore allows the infimum over $\tau_{ABC}$
and the supremum over $\omega_C$ to be exchanged. The latter supremum
sets $-D(\tau_C\|\omega_C)$ to zero, approached by full-rank states $\omega_C$ if $\tau_C$ is singular. The remaining infima over
$\tau_{ABC},\eta_{AC},\sigma_{BC}$ commute, and the marginal relative
entropies have infimum zero. This proves the stated formula.

Let now $1<\alpha\le2$. The same lemma and commutation of the two
suprema give
\begin{align}
 I_\alpha^{\flat,*}(A:B|C)_\rho
 &=\inf_{\eta_{AC},\sigma_{BC}}\sup_{\tau_{ABC}}
   \left\{\frac{\alpha}{1-\alpha}D(\tau_{ABC}\|\rho_{ABC})
   +D(\tau_{AC}\|\eta_{AC})+D(\tau_{BC}\|\sigma_{BC})
   +I(A:B|C)_\tau\right\}.
\end{align}
Here, we first optimized over $\omega_C$, again using
$\sup_{\omega_C>0}-D(\tau_C\|\omega_C)=0$. Expanding the entropies gives
\begin{align}
 &\frac{\alpha}{1-\alpha}D(\tau_{ABC}\|\rho_{ABC})
   +D(\tau_{AC}\|\eta_{AC})+D(\tau_{BC}\|\sigma_{BC})
   +I(A:B|C)_\tau\nonumber\\ 
 &=\frac{2-\alpha}{\alpha-1}H(ABC)_\tau+H(AB|C)_\tau
   +\frac{\alpha}{\alpha-1}\operatorname{Tr}[\tau_{ABC}\log\rho_{ABC}]
   -\operatorname{Tr}[\tau_{AC}\log\eta_{AC}]
   -\operatorname{Tr}[\tau_{BC}\log\sigma_{BC}].
\end{align}
Both entropy terms are concave in $\tau_{ABC}$ and their coefficients
are nonnegative. The expression is therefore concave and continuous in
$\tau_{ABC}$, and jointly convex and continuous in the full-rank pair
$(\eta_{AC},\sigma_{BC})$. Sion's theorem now gives
\begin{align}
 I_\alpha^{\flat,*}(A:B|C)_\rho
 &=\sup_{\tau_{ABC}}\inf_{\eta_{AC},\sigma_{BC}}
       \left\{\frac{\alpha}{1-\alpha}D(\tau_{ABC}\|\rho_{ABC})
   +D(\tau_{AC}\|\eta_{AC})+D(\tau_{BC}\|\sigma_{BC})
   +I(A:B|C)_\tau\right\}\nonumber\\
       &=\sup_{\tau_{ABC}}\left\{
       \frac{\alpha}{1-\alpha}D(\tau_{ABC}\|\rho_{ABC})
        +I(A:B|C)_\tau\right\}.
\end{align}
When a marginal of $\tau$ is singular, its relative-entropy infimum
is approached by full-rank approximations of that marginal.

For states that are not full rank, these expressions are understood
through the limiting definitions above. Compactness, boundedness of
the CMI, and lower semicontinuity of relative entropy control limits
of the optimizing states; the reverse bounds follow by evaluating at
an optimizer supported on $\rho$.
\end{proof}

\begin{remark}
\label{rem: full-support CE and MI}
The extrema in the variational representations of the log-Euclidean conditional entropies, mutual informations, and conditional mutual informations are attained. Thus, the corresponding infima and suprema can be replaced by minima and maxima, respectively. Indeed, their objective functions are continuous on the compact set $\{\tau:\tau\ll\rho\}$.
\end{remark}

\section{Basic properties of log-Euclidean conditional mutual information}

We defined the LE-CMI
through the auxiliary function~\(\Delta_\alpha^\flat\). Accordingly, we
first record the basic properties of \(\Delta_\alpha^\flat\), which will
be used throughout the analysis.

\subsection{Properties of the function \texorpdfstring{$\Delta_\alpha^\flat$}{Delta}}

\begin{lemma}
\label{lem: general limit}
Let $\rho_{ABC}$ be a quantum state, and let $\eta_{AC}$, $\sigma_{BC}$, and $\omega_C$ be full-rank states.
Then
\begin{align}
&\lim_{\alpha\to 1}
\Delta_\alpha^\flat(\rho_{ABC},\eta_{AC},\sigma_{BC},\omega_C) =
D\!\left(
\rho_{ABC}
\middle\|
\exp(
\log \eta_{AC}
+\log \sigma_{BC}
-\log \omega_C)
\right).
\end{align}
\end{lemma}

\begin{proof}
For a Hermitian operator $H_{ABC}$, define
\begin{equation}
f(\alpha)
\coloneqq
\log \Tr\!\left[
e^{\alpha \log \rho_{ABC}+(1-\alpha)H_{ABC}}
\right].
\end{equation}
Since $f(1)=\log \Tr \rho_{ABC}=0$, the desired limit equals $f'(1)$ after recalling the definition in~\eqref{eq:Delta-alpha-quantity}. Writing
\begin{equation}
X(\alpha)\coloneqq \alpha \log \rho_{ABC}+(1-\alpha)H_{ABC},
\end{equation}
we have $X'(\alpha)=\log \rho_{ABC}-H_{ABC}$. By the Duhamel formula and cyclicity of the trace,
\begin{equation}
\frac{\d}{\d\alpha}\Tr e^{X(\alpha)}
=
\Tr\!\left[e^{X(\alpha)}X'(\alpha)\right].
\end{equation}
Therefore,
\begin{align}
f'(1)
&=
\frac{\Tr\!\left[
e^{X(1)}(\log \rho_{ABC}-H_{ABC})
\right]}{\Tr e^{X(1)}}  =
\Tr \left[\rho_{ABC}(\log \rho_{ABC}-H_{ABC})\right],
\end{align}
because $X(1)=\log \rho_{ABC}$ and $\Tr\rho_{ABC}=1$. Substituting $H_{ABC} \coloneqq \log \eta_{AC}+\log \sigma_{BC}-\log \omega_C$ gives the claim. For states that are not full rank, the limiting definition gives the same calculation on the support of $\rho_{ABC}$.
\end{proof}

\begin{lemma}
\label{lem: monotonicity Delta}
Let $\rho_{ABC}$ be a quantum state, and let $\eta_{AC}$, $\sigma_{BC}$, and $\omega_C$ be full-rank states.
Then, the function $\alpha \mapsto \Delta_\alpha^\flat(\rho_{ABC},\eta_{AC},\sigma_{BC},\omega_C)$ is monotonically increasing on the interval $\alpha > 0$, with its value at $\alpha=1$ defined by Lemma~\ref{lem: general limit}.
\end{lemma}

\begin{proof}
    From the proof of  Lemma~\ref{lem: variational form}, we obtain the following for all $\alpha\in(0,1)$:
    \begin{align}
\Delta_\alpha^\flat(\rho_{ABC},\eta_{AC},\sigma_{BC},\omega_C)
        &=\inf_{\tau_{ABC}\ll\rho_{ABC}}
\Big\{
\frac{\alpha}{1-\alpha} \Tr\!\big[\tau_{ABC}(\log \tau_{ABC}-\log \rho_{ABC})\big]
+\Tr[\tau_{ABC}\log \tau_{ABC}] \notag \\
&\qquad\qquad\qquad
-\Tr\!\big[\tau_{ABC}(\log \eta_{AC}+\log \sigma_{BC}-\log \omega_C)\big]
\Big\}.
\end{align}
The monotonicity follows from the fact that the function $\alpha \mapsto \frac{\alpha}{1-\alpha}$ is monotonically increasing, together with the non-negativity of the relative entropy, $\Tr\!\big[\tau_{ABC}(\log \tau_{ABC}-\log \rho_{ABC})\big] = D(\tau_{ABC}\|\rho_{ABC}) \geq 0$. The monotonicity in the interval $\alpha > 1$ is established analogously, with the variational characterization involving a supremum rather than an infimum. Finally, the ordering $\Delta_\alpha \leq \Delta_{\alpha'}$ for $\alpha < 1 < \alpha'$ follows from the fact that the limits from below and above at $\alpha = 1$ coincide, as shown in Lemma~\ref{lem: general limit}.
\end{proof}

\subsection{Relationship between different conditional mutual informations}

\begin{corollary}\label{cor:vN limit}
    Let $\rho_{ABC}$ be a quantum state. Then,
\begin{equation}
    \lim_{\alpha \to 1} I^\flat_\alpha(A:B|C)_{\rho}  = I(A:B|C)_\rho\, .
\end{equation}  
\end{corollary}

\begin{proof}
Apply Lemma~\ref{lem: general limit} with
$\eta_{AC}=\rho_{AC}$, $\sigma_{BC}=\rho_{BC}$, and $\omega_C=\rho_C$.
For states that are not full rank, these expressions are understood
through the limiting definitions above, and the same differentiation
applies on the support of $\rho_{ABC}$.
\end{proof}

In Appendix~\ref{app: Taylor expansion}, we also derive an expansion of $I^\flat_\alpha(A:B|C)$ around $\alpha=1$.

\begin{lemma}
\label{lem: limit Istar}
Let $\rho_{ABC}$ be a quantum state. Then,
\begin{equation}
    \lim_{\alpha \to 1} I^{\flat,*}_\alpha(A:B|C)_{\rho} = I(A:B|C)_\rho\, .
\end{equation}
\end{lemma}

\begin{proof}
We first assume that $\rho_{ABC}$ is full rank and prove the limit
$\alpha\to1^-$. The inequality
\begin{align}
\lim_{\alpha \to 1^-} I^{\flat,*}_\alpha(A:B|C)_{\rho}
\leq I(A:B|C)_{\rho}
\end{align}
follows from Corollary~\ref{cor: variational form triply optimized}
by evaluating the optimization at $\tau_{ABC}=\rho_{ABC}$.
For the opposite inequality, consider that
\begin{align}
    \lim_{\alpha \to 1^-}I^{\flat,*}_\alpha(A:B|C)_{\rho}
    &= \sup_{\alpha \in (1/2,1)}I^{\flat,*}_\alpha(A:B|C)_{\rho}\\
    & = \sup_{\alpha \in(1/2,1)}\inf_{\sigma_{BC},\,\eta_{AC}}
        \sup_{\omega_C}\Delta_\alpha^\flat(\rho_{ABC},\eta_{AC},\sigma_{BC},\omega_C)\\
    & \geq \sup_{\alpha \in(1/2,1)}\inf_{\sigma_{BC},\,\eta_{AC}}
        \Delta_\alpha^\flat(\rho_{ABC},\eta_{AC},\sigma_{BC},\rho_C)\\
    & = \sup_{\alpha \in(1/2,1)}\min_{\sigma_{BC},\,\eta_{AC}}
        \Delta_\alpha^\flat(\rho_{ABC},\eta_{AC},\sigma_{BC},\rho_C)\\
    & = \min_{\sigma_{BC},\,\eta_{AC}}\sup_{\alpha \in(1/2,1)}
        \Delta_\alpha^\flat(\rho_{ABC},\eta_{AC},\sigma_{BC},\rho_C)
        \label{eq: minimax for limit}\\
    & = \min_{\sigma_{BC},\,\eta_{AC}}\lim_{\alpha \to 1^-}
        \Delta_\alpha^\flat(\rho_{ABC},\eta_{AC},\sigma_{BC},\rho_C)\\
    & = \min_{\sigma_{BC}>0,\,\eta_{AC}>0}D\!\left(
        \rho_{ABC}\middle\|\exp(
        \log\eta_{AC}+\log\sigma_{BC}-\log\rho_C)\right)\\
    & = D\!\left(\rho_{ABC}\middle\|\exp(
        \log\rho_{AC}+\log\rho_{BC}-\log\rho_C)\right)\\
    & = I(A:B|C)_{\rho}\, .
\end{align}
Here, monotonicity follows from Lemma~\ref{lem: monotonicity Delta}.
For fixed $\alpha$, the minimum is attained on the compact state
spaces by Lemma~\ref{lem: lower-semicont Delta}, and its optimizers
are full rank by Lemma~\ref{lem: full-rank}.
Equation~\eqref{eq: minimax for limit} follows from the Mosonyi--Hiai
minimax theorem~\cite[Corollary~A.2]{mosonyi2011quantum}, since the
objective is lower semicontinuous and monotone in $\alpha$
(applying the bounded increasing transform $\arctan$ to include
infinite boundary values).
At a singular pair of auxiliary states, the limiting trace is strictly
less than one, so the pointwise limit of $\Delta_\alpha^\flat$ is
$+\infty$. We may therefore restrict the subsequent minimum to full-rank
states, apply Lemma~\ref{lem: general limit}, and use that the marginals
of $\rho$ attain the minimum~\cite[Proposition~2]{berta2015renyi}.

Let us now consider the limit $\alpha\to1^+$. Evaluating the variational
formula at $\tau_{ABC}=\rho_{ABC}$ gives
\begin{align}
\lim_{\alpha \to 1^+} I^{\flat,*}_\alpha(A:B|C)_{\rho}
\geq I(A:B|C)_{\rho}.
\end{align}
For the reverse inequality, we have
\begin{align}
    \lim_{\alpha \to 1^+}I^{\flat,*}_\alpha(A:B|C)_{\rho}
    &= \inf_{\alpha \in (1,2)}I^{\flat,*}_\alpha(A:B|C)_{\rho}\\
    & = \inf_{\alpha\in(1,2)}\inf_{\sigma_{BC},\,\eta_{AC}}
        \sup_{\omega_C}\Delta_\alpha^\flat(\rho_{ABC},\eta_{AC},\sigma_{BC},\omega_C)\\
    & \leq \inf_{\alpha\in(1,2)}\sup_{\omega_C}
        \Delta_\alpha^\flat(\rho_{ABC},\rho_{AC},\rho_{BC},\omega_C)\\
    & = \max_{\omega_C}\inf_{\alpha\in(1,2)}
        \Delta_\alpha^\flat(\rho_{ABC},\rho_{AC},\rho_{BC},\omega_C)
        \label{eq: minimax for limit 2}\\
    & = \max_{\omega_C}\lim_{\alpha\to1^+}
        \Delta_\alpha^\flat(\rho_{ABC},\rho_{AC},\rho_{BC},\omega_C)\\
    & = \max_{\omega_C>0}D\!\left(
        \rho_{ABC}\middle\|\exp(
        \log\rho_{AC}+\log\rho_{BC}-\log\omega_C)\right)\\
    & = D\!\left(\rho_{ABC}\middle\|\exp(
        \log\rho_{AC}+\log\rho_{BC}-\log\rho_C)\right)\\
    & = I(A:B|C)_{\rho}\, .
\end{align}
For fixed $1<\alpha<2$, the maximum is attained on the compact state
space by Lemma~\ref{lem: lower-semicont Delta}, and its optimizers
are full rank by Lemma~\ref{lem: full-rank Q}.
Equation~\eqref{eq: minimax for limit 2} follows from
the same minimax theorem, applied to $-\arctan\Delta_\alpha^\flat$.
For singular $\omega_C$, the limiting trace is strictly less than one
and the pointwise limit is $-\infty$. Hence the final maximization
may be restricted to full-rank $\omega_C$, and
Lemma~\ref{lem: general limit} applies; its maximum is attained at
$\rho_C$~\cite[Proposition~2]{berta2015renyi}.

For states that are not full rank, these expressions are understood
through the limiting definitions above. Indeed, on the compact domain
$\tau\ll\rho$, Corollary~\ref{cor: variational form triply optimized}
gives $D(\tau_\alpha\|\rho)\leq
2\log\min\{d_A,d_B\}\,|1-\alpha|/\alpha$ for its optimizing states;
compactness and continuity of ordinary conditional mutual information
then give the same two limits.
\end{proof}

\begin{lemma}
\label{lem: ordering LE-CMIs}
Let $\rho_{ABC}$ be a quantum state. Then, for all $\alpha\in(0,1)\cup(1,2]$,
\begin{equation}
    I^\flat_\alpha(A:B|C)_{\rho} \geq I^{\flat,*}_{\alpha}(A:B|C)_\rho\, .
\end{equation}  
\end{lemma}

\begin{proof}
    The claim follows directly from the variational characterizations in terms of relative entropy obtained in Corollaries~\ref{cor: variational form} and \ref{cor: variational form triply optimized}. Indeed, \( I^\flat_\alpha(A:B|C)_\rho \) involves additional non-negative contributions, and the inequality $D(\tau_{AC}\|\rho_{AC}) \ge D(\tau_C\|\rho_C)$
holds by the data-processing inequality. 
\end{proof}

\subsection{Data-processing inequality and monotonicity in \texorpdfstring{$\alpha$}{alpha}}

\begin{theorem}
Let $\alpha \in (0,1)$, let $\rho_{ABC}$ be a quantum state, let $\mathcal{N}_{A \rightarrow A'} \otimes \mathcal{N}_{B \rightarrow B'}  $ be a quantum channel, and let $\sigma_{A'B'C} = \mathcal{N}_{A \rightarrow A'}\otimes \mathcal{N}_{B \rightarrow B'}(\rho_{ABC})$. Then,
\begin{equation}
    I^\flat_\alpha(A:B|C)_{\rho} \geq I^\flat_\alpha(A':B'|C)_\sigma, \quad \text{and} \quad  I^{\flat,*}_\alpha(A:B|C)_{\rho} \geq I^{\flat,*}_\alpha(A':B'|C)_\sigma.
\end{equation}  
\end{theorem}

\begin{proof}
    We use the variational form in Corollary~\ref{cor: variational form}. We can upper bound the total infimum with the infimum over states $\mathcal{N}_{A \rightarrow A'} \otimes \mathcal{N}_{B \rightarrow B'}(\tau_{ABC})$. We then obtain
    \begin{equation}
\begin{aligned}
&(1-\alpha)I_\alpha^\flat(A':B'|C)_{\mathcal{N}_{A \rightarrow A'} \otimes \mathcal{N}_{B \rightarrow B'}(\rho_{ABC})} \\
    &\qquad \leq \alpha D(\mathcal{N}_{A \rightarrow A'} \otimes \mathcal{N}_{B \rightarrow B'}(\tau_{ABC})\|\mathcal{N}_{A \rightarrow A'} \otimes \mathcal{N}_{B \rightarrow B'}(\rho_{ABC})) \\
    & \qquad \quad  + (1-\alpha) \Big(D(\mathcal{N}_{A \rightarrow A'}(\tau_{AC})\|\mathcal{N}_{A \rightarrow A'}(\rho_{AC}))+D(\mathcal{N}_{B \rightarrow B'}(\tau_{BC})\|\mathcal{N}_{B \rightarrow B'}(\rho_{BC}))\\
    & \qquad \qquad \qquad \qquad \qquad \qquad \qquad \qquad -D(\tau_C\|\rho_C)+I(A':B'|C)_{\mathcal{N}_{A \rightarrow A'} \otimes \mathcal{N}_{B \rightarrow B'}(\tau_{ABC})}\Big).
\end{aligned}
\end{equation}
    The data-processing inequality (DPI) of the Umegaki relative entropy and the conditional mutual information concludes the proof.
For states that are not full rank, the limiting conventions apply,
and $\tau\ll\rho$ implies $\mathcal N(\tau)\ll\mathcal N(\rho)$. The proof for $I^{\flat,*}_\alpha(A:B|C)$ is similar, using Corollary~\ref{cor: variational form triply optimized} instead.
\end{proof}

An immediate corollary of Lemma~\ref{lem: monotonicity Delta} is the monotonicity in $\alpha$ of the Rényi conditional mutual information.

\begin{corollary}\label{cor:monotonicity}
Let $0<\alpha<\alpha'$, and let $\rho_{ABC}$ be a quantum state. Then
\begin{equation}
 I_\alpha^\flat(A:B|C)_\rho\le I_{\alpha'}^\flat(A:B|C)_\rho.
\end{equation}
If in addition $\alpha'\le2$, then
\begin{equation}
 I_\alpha^{\flat,*}(A:B|C)_\rho
 \le I_{\alpha'}^{\flat,*}(A:B|C)_\rho.
\end{equation}
\end{corollary}
\begin{proof}
The monotonicity of $I_\alpha^\flat(A:B|C)_\rho$ follows directly from
Lemma~\ref{lem: monotonicity Delta}, together with the limiting definition of
$I_\alpha^\flat$.

For $I_\alpha^{\flat,*}(A:B|C)_\rho$, we use
Corollary~\ref{cor: variational form triply optimized}. On each of the
intervals $(0,1)$ and $(1,2]$, the coefficient $\alpha/(1-\alpha)$ is
increasing, and $D(\tau_{ABC}\Vert\rho_{ABC})\geq 0$. Hence, the corresponding
variational infimum or supremum is monotonically increasing in $\alpha$.
Evaluating the variational objective at $\tau=\rho$ gives
\[
I_\alpha^{\flat,*}(A:B|C)_\rho
\leq I(A:B|C)_\rho \quad (\alpha<1),
\qquad
I_\alpha^{\flat,*}(A:B|C)_\rho
\geq I(A:B|C)_\rho \quad (1<\alpha\leq2).
\]
Together with $I_1^{\flat,*}(A:B|C)_\rho=I(A:B|C)_\rho$, this also proves
the comparisons across $\alpha=1$.
\end{proof}

\begin{lemma}
\label{lem: non-negativity}
Let $\rho_{ABC}$ be a quantum state. Then $I_\alpha^\flat(A:B|C)_\rho\ge0$. In addition, if $\alpha \leq 2$, then $I_\alpha^{\flat,*}(A:B|C)_\rho \ge0$.
\end{lemma}

\begin{proof}
    Let us consider first the case $\alpha<1$. The non-negativity of $I^{\flat,*}_{\alpha}(A:B|C)_\rho$ is clear from the variational form, as the objective function is the sum of two non-negative terms and hence it is always non-negative. The case $1<\alpha\le2$ follows from the fact that $I^{\flat,*}_{\alpha}(A:B|C)_\rho\geq I(A:B|C)_\rho\geq 0$. For $I^\flat_\alpha(A:B|C)_{\rho}$, the claim for $\alpha<1$ follows then from Lemma~\ref{lem: ordering LE-CMIs}, and for $\alpha>1$ holds since from Corollaries~~\ref{cor:vN limit} and~\ref{cor:monotonicity}, we get that $I^{\flat}_{\alpha}(A:B|C)_\rho\geq I(A:B|C)_\rho\geq 0$. 
\end{proof}

\begin{lemma}
\label{lem: additivity Ialpha}
    Let $\rho_{ABC}$ and $\sigma_{A' B' C'}$ be two quantum states. Then,
    \begin{equation}
        I_\alpha^{\flat}(AA' :BB'|CC')_{\rho\otimes \sigma} = I_\alpha^{\flat}(A:B|C)_\rho+ I_\alpha^{\flat}(A' :B'|C')_\sigma.
    \end{equation}
\end{lemma}

\begin{proof}
The identity follows from additivity of the log-Euclidean relative
entropy and the tensor factorization
$\exp(X\otimes I+I\otimes Y)=\exp(X)\otimes\exp(Y)$.
For states that are not full rank, the result follows by a limit argument. 
\end{proof}

\section{Additivity of \texorpdfstring{$I^{\flat,*}_\alpha$}{the optimized LE-CMI}}
\label{sec: additivity}
In this section, we prove the additivity of \(I^{\flat,*}_\alpha\).
Let us first outline the main ideas of the proof in the regime
\(\alpha\in(1/2,1)\). While subadditivity follows directly from the
variational characterization, superadditivity is more delicate. The first
step is to restrict the supremum over all states to product states of the
form \(\omega_C\otimes\omega_{C'}\). The second step is to show that the
resulting minimization problem is additive. 
To prove the required additivity of the minimization problem, we first
establish joint convexity in the variables \(\sigma_{BC}\) and
\(\eta_{AC}\). By convexity, it then suffices to verify that the directional
derivatives are separately non-negative in all admissible directions. This argument is inspired by
the proof of~\cite[Proposition~23]{cheng2025tight}.
We then obtain a fixed-point equation, from which one can readily check that the
tensor product of the optimizers for the marginal problems is itself an
optimizer of the problem. The fixed-point characterization of optimizers of relative entropies first appeared in~\cite{hayashi2016correlation}, was subsequently considered in~\cite{cheng2025tight}, and has more recently been studied in~\cite{rubboli2024quantum,burri2026doubly}. In~\cite{brahmachari2025fixed,AfhamKuengFerrie2022QuantumMeanStates},
it was used to derive an iterative algorithm that converges to the optimizer through repeated
application of the corresponding fixed-point map, starting from an arbitrary initial point.

Finally, we remark that our proof does not establish additivity in the range $\alpha\in(0,1/2)$. The main obstruction is the lack of joint convexity of the objective function in $(\eta_{AC},\sigma_{BC})$. For the regime $\alpha>1$, we can prove additivity for $\alpha\in[1,2]$. For $\alpha>2$, the required concavity in $\omega_C$ is no longer available. Additivity for $\alpha\in(0,1/2)$ remains open.

\subsection{Joint convexity and properties of the optimizers}
The first lemma establishes the joint convexity of the relevant function in the variables \(\sigma_{BC}\) and \(\eta_{AC}\). The second part of the statement was also obtained in~\cite{lieb1973convex,epstein1973remarks} for $\beta=1$ (see also~\cite{tropp2012joint}).

\begin{lemma}
\label{lem: Joint concavity}
Let $H$ be a Hermitian operator.
    The function
    \begin{equation}
        (A,B)\to \log{\Tr\exp((H+(1-\alpha)\log A +(1-\alpha)\log{B} ))}
    \end{equation}
    is concave for $\alpha \in [1/2,1]$ on the set of positive definite matrices. Moreover, the function
    \begin{equation}
        C\to \log{\Tr\exp(H +\beta\log{C} )}
    \end{equation}
    is concave for $\beta \in [0,1]$ on the set of positive definite matrices.
\end{lemma}

\begin{proof}
    By using the Gibbs variational principle (recalled as Lemma~\ref{lem: Gibbs}), we can write 
    \begin{align}
    &\log{\Tr\exp((H+(1-\alpha)\log A +(1-\alpha)\log{B} )) }\notag \\*
    &\qquad\qquad\qquad = \sup_\tau \{\Tr[\tau (H+(1-\alpha)\log A +(1-\alpha)\log{B} )]-\Tr[\tau \log{\tau}]\}\\
    &\qquad\qquad\qquad = \sup_\tau \{ \Tr[\tau H]+(2(1-\alpha)-1)\Tr[\tau \log{\tau}] - (1-\alpha)D(\tau \|A)- (1-\alpha)D(\tau \|B) \} .
    \end{align}
    Hence, in the case $(2(1-\alpha)-1)\leq 0$ and $1-\alpha \geq 0$, we can use that the above objective function is jointly concave in $\tau$, $A$, and $B$ by concavity of the von Neumann entropy and joint convexity of the relative entropy. We then use the fact that the supremum of a jointly concave function over a convex set preserves concavity.

To prove the second statement, we similarly note that
\begin{align}
    &\log{\Tr\exp (H +\beta\log{C} )}  = \sup_\tau \{\Tr[\tau H]+(\beta-1)\Tr[\tau \log{\tau}]-\beta D(\tau\|C)\}\, .
    \end{align}
Since the objective function on the right-hand side is jointly concave in $\tau$ and $C$, the fact that the supremum over jointly concave functions is concave completes the argument.
\end{proof}

Boundary reference states are understood by replacing each varying
argument $X$ with $X+\delta I$ and taking the limit $\delta\downarrow0$.

\begin{lemma}
\label{lem: lower-semicont Delta}
For fixed full-rank $\rho_{ABC}$ and $\omega_C$, and $0<\alpha<1$,
the boundary extension of
$(\eta_{AC},\sigma_{BC})\mapsto\Delta_\alpha^\flat(\rho,\eta,\sigma,\omega)$
is lower semicontinuous. For fixed full-rank $\rho_{ABC}$,
$\eta_{AC}$, and $\sigma_{BC}$, and $\alpha>1$, the boundary extension
of $\omega_C\mapsto\Delta_\alpha^\flat(\rho,\eta,\sigma,\omega)$
is upper semicontinuous.
\end{lemma}
\begin{proof}
In both cases, the coefficients of the logarithms of the varying
arguments in the trace exponential are positive. Adding $\delta I$
to those arguments gives continuous trace functions, decreasing to
the stated extension as $\delta\downarrow0$, by operator monotonicity
of the logarithm and monotonicity of the trace exponential.
Their infimum is upper semicontinuous. Applying the logarithm
preserves upper semicontinuity, with $\log0=-\infty$.
Multiplication by $1/(\alpha-1)$ gives the two claims.
\end{proof}

The next lemma derives the optimality conditions for the relevant function being optimized. In the case where the function is convex, these conditions follow from the non-negativity of its directional derivatives. We then use them to show that when the input state is full rank, the optimizer must also be full rank. Hence, the optimizer lies in the interior of the feasible set, and the optimality conditions can therefore be rewritten as a fixed-point equation.

For a differentiable function \(f\) at a full-rank state, we write
the directional derivative of \(f\) at \(\tau\) in the direction of \(\sigma\) as
\begin{align}
    \left.\frac{\mathrm d}{\mathrm d x}
    f\big((1-x)\tau+x\sigma\big)\right|_{x=0}
    =
    \big\langle \nabla f(\tau),\sigma-\tau \big\rangle ,
\end{align}
where \(\nabla f(\tau)\) denotes the gradient of \(f\) at \(\tau\).
All trace-exponential functions used below are differentiable on their
positive definite domains. For a jointly differentiable function of
two variables, the directional derivative in a joint direction is
the sum of its two partial directional derivatives. Consequently,
for a jointly convex function it suffices to verify the nonnegativity
of the two partial directional derivatives; for a jointly concave
function, the corresponding derivatives must be nonpositive.
Throughout this section, $\pi_X=I_X/d_X$ denotes the maximally mixed state.

For $L_{AB}>0$ and $0<\beta<1$, we define the function
\begin{align}
    Q_{\beta}^\flat( L_{AB},\sigma_{B}) \coloneqq  \Tr[\exp(\log L_{AB}+\beta\log\sigma_{B})]
\end{align}
and its gradient (up to the constant $\beta$)
\begin{align}
    \Xi_\beta(L_{AB},\tau_{B}) & \coloneqq \beta^{-1}\nabla_{\tau_B} Q_{\beta}^\flat( L_{AB},\tau_{B})\\
    & =  \int_{0}^{\infty} (\tau_{B}+t)^{-1}\Tr_A[\exp(\log L_{AB}+\beta\log\tau_{B})](\tau_{B}+t)^{-1} \d t \, .
\end{align}

We first prove that for full-rank matrices, any optimizer must be full rank for $\beta\in (0,1)$.
\begin{lemma}
\label{lem: full-rank Q}
Let $L_{AB}$ be a positive definite matrix, and let $\beta \in (0,1)$. If $\tau_{B} \in \argmax_{\sigma_{B}} Q_{\beta}^\flat(L_{AB},\sigma_{B}) $, then $\tau_B$ is full rank.
\end{lemma}
\begin{proof}
The boundary extension of $Q_\beta^\flat$ is upper semicontinuous,
by the same limiting argument as in
Lemma~\ref{lem: lower-semicont Delta}, and therefore attains its
maximum on the compact state space. This maximum is positive,
since the value at $\pi_B$ is positive.

    Suppose that $\tau_{B}$ is not full rank. We then consider a mixture of the form  $(1-x)\tau_{B}+x\pi_{B}$ and show that as $x$ goes to zero, the derivative goes to $\infty$. Since the function $\sigma_B \to \log Q_{\beta}^\flat(L_{AB},\sigma_{B})$ is concave by Lemma~\ref{lem: Joint concavity}, this contradicts the fact that $\tau_B$ is an optimizer.

    The derivative with respect to $x$ in the direction of the function
    \begin{equation}
    x\mapsto \log Q_{\beta}^\flat(L_{AB},(1-x)\tau_{B}+x\pi_{B})     
    \end{equation}
    is equal to  (up to a positive factor)
    \begin{align}
&\Tr\Bigg[
\Tr_A\Big[
\exp\big(\log L_{AB}
+ \beta\log((1-x)\tau_{B} + x\pi_{B})
\big)  \frac{\partial}{\partial x} \big(\beta\log((1-x)\tau_{B} + x\pi_{B})
\big)
\Big]
\Bigg]\,.
\label{eq:deriv-exp-1}
\end{align}
    We then use the known formula for the derivative of the logarithm
    \begin{align}
        \frac{\partial}{\partial x} \log{B(x)} = \int_0^\infty  (B(x)+t)^{-1} \left(\frac{\partial}{\partial x} B(x)\right)  (B(x)+t)^{-1}\d t
    \end{align}
    and that $\frac{\partial}{\partial x} ((1-x)\tau_{B} + x\pi_{B})=\pi_B-\tau_B$
to obtain that, up to a positive constant $\beta$, \eqref{eq:deriv-exp-1} is equal to
\begin{align}
&\Tr\Bigg[
(\pi_{B}-\tau_{B})
\int_{0}^{\infty}
\bigl((1-x)\tau_{B} + x\pi_{B} + t\bigr)^{-1}
\notag\\
&\qquad\qquad \times
\Tr_A\Big[
\exp\big(\log L_{AB}
+ \beta\log((1-x)\tau_{B} + x\pi_{B})
\big)
\Big]
\bigl((1-x)\tau_{B} + x\pi_{B} + t\bigr)^{-1}
\, dt
\Bigg]\notag\\
&=\Tr\Bigg[\pi_{B}
\int_{0}^{\infty}
\bigl((1-x)\tau_{B} + x\pi_{B} + t\bigr)^{-1}
\notag\\
&\qquad\qquad \times
\Tr_A\Big[
\exp\big(\log L_{AB}
+ \beta\log((1-x)\tau_{B} + x\pi_{B})
\big)
\Big]
\bigl((1-x)\tau_{B} + x\pi_{B} + t\bigr)^{-1}
\, dt
\Bigg]\notag\\
&\quad -\Tr\Bigg[\tau_{B}
\int_{0}^{\infty}
\bigl((1-x)\tau_{B} + x\pi_{B} + t\bigr)^{-1}
\notag\\
&\qquad\qquad \times
\Tr_A\Big[
\exp\big(\log L_{AB}
+ \beta\log((1-x)\tau_{B} + x\pi_{B})
\big)
\Big]
\bigl((1-x)\tau_{B} + x\pi_{B} + t\bigr)^{-1}
\, dt
\Bigg]\\
&=\Tr\Bigg[\pi_{B}
\int_{0}^{\infty}
\bigl((1-x)\tau_{B} + x\pi_{B} + t\bigr)^{-1}
\notag\\
&\qquad\qquad \times
\Tr_A\Big[
\exp\big(\log L_{AB}
+ \beta\log((1-x)\tau_{B} + x\pi_{B})
\big)
\Big]
\bigl((1-x)\tau_{B} + x\pi_{B} + t\bigr)^{-1}
\, dt
\Bigg]\notag\\
&\quad -\Tr\Bigg[\tau_{B}
\bigl((1-x)\tau_{B} + x\pi_{B}\bigr)^{-1}
\Tr_A\Big[
\exp\big(\log L_{AB}
+ \beta\log((1-x)\tau_{B} + x\pi_{B})
\big)
\Big]\Bigg]\,.
\end{align}
In the last equality, we used that, for the term involving \( \tau_{B} \), the integral can be evaluated by using cyclicity of the trace. For $0<x<1/2$, this nonnegative term is at most
$(1-x)^{-1}Q_\beta^\flat(L_{AB},(1-x)\tau_B+x\pi_B)$, because
$\tau_B((1-x)\tau_B+x\pi_B)^{-1}\le(1-x)^{-1}I_B$.
It is therefore uniformly bounded: the trace exponential is at most
$\Tr L_{AB}$, by monotonicity of the trace exponential and
$\log((1-x)\tau_B+x\pi_B)\le0$.

For the first term in the sum, we then show that if $\tau_B$ is not full rank, the derivative goes to $\infty$, hence contradicting the fact that concavity implies that the derivative in any direction at the optimizers must be non-positive. We use that $(1-x)\tau_B + x \pi_B \leq P(\tau_B)+x P(\tau_B)^c $ where $P(\tau_B)$ is the projector onto the support of $\tau_B$ and $P(\tau_B)^c\coloneqq I_B-P(\tau_B)$. We then use the fact that $L_{AB}$ is full rank, and hence $\log L_{AB} \geq \log(\lambda_{\min} (L_{AB})) I_{AB}$, where $\log(\lambda_{\min} (L_{AB})) \neq  -\infty$ denotes the smallest eigenvalue of $\log L_{AB}$.
 The terms inside the exponential are lower bounded by 
\begin{align}
    \log L_{AB}
+ \beta\log((1-x)\tau_{B} + x\pi_{B}) \geq (\log(\lambda_{\min}(L_{AB}))+\beta\log (x/d_B))I_{AB} \,.
\end{align}
We then use the fact that the exponential operator is monotone for commuting operators, solve the integral, and use the operator antimonotonicity of the inverse to obtain the lower bound
\begin{align}
&\Tr\Bigg[\pi_{B}
\int_{0}^{\infty}
\bigl((1-x)\tau_{B} + x\pi_{B} + t\bigr)^{-1}
\notag\\
&\qquad\qquad \times
\Tr_A\Big[
\exp\big(\log L_{AB}
+ \beta\log((1-x)\tau_{B} + x\pi_{B})
\big)
\Big]
\bigl((1-x)\tau_{B} + x\pi_{B} + t\bigr)^{-1}
\, dt
\Bigg]\notag\\
&\quad \geq \Tr[\pi_{B}
\bigl((1-x)\tau_{B} + x\pi_{B}\bigr)^{-1} \lambda_{\min}(L_{AB})x^{\beta}d_A I_B d_B^{-\beta}]\\
&\quad \geq \Tr[\pi_{B}
(P(\tau_B)+x^{-1} P(\tau_B)^c) \lambda_{\min}(L_{AB})x^{\beta}d_A I_Bd_B^{-\beta}]\\
&\quad \geq \lambda_{\min}(L_{AB})d_B^{-\beta}\Tr[\pi_{B} d_A P(\tau_{B})^c]x^{\beta-1},
\end{align}
which diverges to $\infty$ as $x$ goes to zero for $\beta \in (0,1)$.
\end{proof}

\begin{lemma}
\label{lem; derivative Q}
Let $L_{AB}$ be a positive definite matrix and $\beta \in (0,1)$.
    Then a full-rank state $\tau_{B}$ belongs to $\argmax_{\sigma_{B}} Q_{\beta}^\flat(L_{AB},\sigma_{B})$ if and only if, for all states $\sigma_{B}$, the following equality holds:
    \begin{align}
        \Tr[\sigma_{B}\Xi_\beta(L_{AB},\tau_{B})] = Q_{\beta}^\flat(L_{AB},\tau_{B})\, .
    \end{align}
\end{lemma}

\begin{proof}
   The logarithm of the objective function is concave in $\sigma_B$ by Lemma~\ref{lem: Joint concavity}. Hence, \(\tau_B\) is a maximizer if and only if the directional derivative at \(\tau_B\) in every feasible direction toward a state \(\sigma_B\) is non-positive. Moreover, by Lemma~\ref{lem: full-rank Q}, any optimizer is full rank and therefore lies in the interior of the state space. Consequently, the directional derivative must in fact vanish in every admissible direction.
   We then use the known formula for the derivative of the logarithm
    \begin{align}
        \frac{\partial}{\partial x} \log{B(x)} = \int_0^\infty  (B(x)+t)^{-1} \left(\frac{\partial}{\partial x} B(x)\right)  (B(x)+t)^{-1}\d t
    \end{align}
    to obtain that the derivative with respect to $x$ of the function $x\mapsto Q_{\beta}^\flat(L_{AB},(1-x)\tau_{B}+x\sigma_{B})$ at $x=0$ is equal to (up to a constant factor of $\beta$)
    \begin{align}
        \Tr[\sigma_{B} \int_{0}^{\infty} (\tau_{B}+t)^{-1}\Tr_A[\exp(\log L_{AB}+\beta\log\tau_{B})](\tau_{B}+t)^{-1} \d t] - Q_{\beta}^\flat(L_{AB},\tau_{B})\, .
    \end{align}
   Setting the derivative equal to zero yields the result.
For later use, this condition is equivalent to the fixed-point equation
\begin{equation}
 \Tr_A\exp(\log L_{AB}+\beta\log\tau_B)
   =Q_\beta^\flat(L_{AB},\tau_B)\tau_B.
 \label{eq:single-fixed-point}
\end{equation}
Indeed, the equality for every state $\sigma_B$ says that
$\Xi_\beta=Q_\beta^\flat I_B$. In an eigenbasis of $\tau_B$,
the derivative of the logarithm multiplies the $(i,j)$ matrix entry
by $(\log\lambda_i-\log\lambda_j)/(\lambda_i-\lambda_j)$, with
value $1/\lambda_i$ when $\lambda_i=\lambda_j$. These factors are
strictly positive. Inverting this entrywise map sends $I_B$ to
$\tau_B$, which proves~\eqref{eq:single-fixed-point}. Conversely,
that fixed-point equation makes every displayed directional
derivative zero; concavity of $\log Q_\beta^\flat$ then gives global
optimality.
\end{proof}
Next, we derive a similar statement to that of Lemma~\ref{lem: full-rank Q} for $\Delta_\alpha^\flat(\rho_{ABC},\eta_{AC},\sigma_{BC},\omega_C)$ in the range $\alpha \in (1/2,1)$.

\begin{lemma}
\label{lem: full-rank}
    Let $\alpha \in (1/2,1)$, and let $\rho_{ABC}$ and $\omega_C$ be full-rank states. Then all optimizers of 
    \begin{align}
        \min_{\eta_{AC},\sigma_{BC}} \Delta_\alpha^\flat(\rho_{ABC},\eta_{AC},\sigma_{BC},\omega_C)\,,
    \end{align}
    where the optimization is over all states, are full rank.
\end{lemma}

\begin{proof}
The proof is similar to the one of Lemma~\ref{lem: full-rank Q}.
    We assume that at least one of the operators $\eta_{AC}$ and $\sigma_{BC}$ is not full rank. We then consider the convex combinations $(1-x)\eta_{AC}+x\pi_{AC}$ and $(1-x)\sigma_{BC}+x\pi_{BC}$ and show that as $x$ goes to zero, the derivative goes to $-\infty$. Since the function is convex, this contradicts the fact that the boundary states are optimizers.  
    
    We equivalently maximize the positive trace exponential inside the logarithm and show that its derivative is positive for sufficiently small $x$. Its logarithm is concave by Lemma~\ref{lem: Joint concavity}.
    The derivative with respect to $x$ of the trace exponential
\begin{align}
x\mapsto\exp\!\left((\alpha-1)
\Delta_\alpha^\flat(\rho_{ABC},(1-x)\eta_{AC}+x\pi_{AC},
                         (1-x)\sigma_{BC}+x\pi_{BC},\omega_C)\right)
\end{align}
is equal, up to the positive constant $1-\alpha$, to 
    \begin{align}
&\Tr\Bigg[
(\pi_{AC}-\eta_{AC})
\int_{0}^{\infty}
\bigl((1-x)\eta_{AC} + x\pi_{AC} + t\bigr)^{-1}
\notag \\
&\qquad\qquad \times
\Tr_B\Big[
\exp\Big(
\alpha \log \rho_{ABC}
+ (1-\alpha)\log\bigl((1-x)\eta_{AC} + x\pi_{AC}\bigr)
\notag \\
&\qquad\qquad\qquad\qquad\quad
+ (1-\alpha)\log\bigl((1-x)\sigma_{BC} + x\pi_{BC}\bigr)
+ (\alpha-1)\log \omega_C
\Big)
\Big]
\notag \\
&\qquad\qquad \times
\bigl((1-x)\eta_{AC} + x\pi_{AC} + t\bigr)^{-1}
\, dt
\Bigg]
+ \notag \\
&\Tr\Bigg[
(\pi_{BC} - \sigma_{BC})
\int_{0}^{\infty}
\bigl((1-x)\sigma_{BC} + x\pi_{BC} + t\bigr)^{-1}
\notag \\
&\qquad\qquad \times
\Tr_A\Big[
\exp\Big(
\alpha \log \rho_{ABC}
+ (1-\alpha)\log\bigl((1-x)\eta_{AC} + x\pi_{AC}\bigr)
\notag \\
&\qquad\qquad\qquad\qquad\quad
+ (1-\alpha)\log\bigl((1-x)\sigma_{BC} + x\pi_{BC}\bigr)
+ (\alpha-1)\log \omega_C
\Big)
\Big]
\notag \\
&\qquad\qquad \times
\bigl((1-x)\sigma_{BC} + x\pi_{BC} + t\bigr)^{-1}
\, dt
\Bigg].
\end{align}
As in the proof of Lemma~\ref {lem: full-rank Q}, we exploit the linearity of the trace to separate the two terms into four terms involving $\pi_{AC}$, $\eta_{AC}$, $\pi_{BC}$, and $\sigma_{BC}$. We then observe that, for the terms involving \( \eta_{AC} \) and \( \sigma_{BC} \), the integral can be evaluated by using the cyclicity of the trace. The two nonnegative terms subtracted in this decomposition are
uniformly bounded for $0<x<1/2$: each is at most $(1-x)^{-1}$ times
the trace exponential. That trace is uniformly bounded because
$\alpha\log\rho_{ABC}+(\alpha-1)\log\omega_C$ is fixed and both
remaining logarithms are nonpositive. For the remaining two positive terms, we then show that if either one of the states is not full rank, the derivative goes to $\infty$, hence contradicting the fact that concavity implies that the derivative in any direction at the optimizers must be non-positive. We use that $(1-x)\gamma + x \pi \leq P(\gamma)+x P(\gamma)^c $, where $P(\gamma)$ is the projector onto the support of $\gamma$. We then use the fact that the logarithm is operator monotone, and that, by the full-rank assumption, $\alpha\log \rho_{ABC}
+ (\alpha-1)\log\omega_C \geq -k I_{ABC}$ for a positive constant $k$. Let us assume, for example, that \( \eta_{AC} \) is not full rank. The terms inside the exponential are lower bounded by
\begin{align}
&\alpha \log \rho_{ABC}
+ (1-\alpha)\log\bigl((1-x)\eta_{AC} + x\pi_{AC}\bigr)
+ (1-\alpha)\log\bigl((1-x)\sigma_{BC} + x\pi_{BC}\bigr)
+ (\alpha-1)\log \omega_C \notag\\
& \qquad\qquad  \geq \left(-k+2(1-\alpha)\log x-(1-\alpha)\log{(d_Ad_Bd_C^2)}\right)I_{ABC}.
\end{align}
The fact that the exponential operator is monotone for commuting operators implies that the first term is lower bounded by
\begin{align}
    &\Tr\Bigg[
\pi_{AC}
\int_{0}^{\infty}
\bigl((1-x)\eta_{AC} + x\pi_{AC} + t\bigr)^{-1}
\notag \\
&\qquad\qquad \times
\Tr_B\Big[
\exp\Big(
\alpha \log \rho_{ABC}
+ (1-\alpha)\log\bigl((1-x)\eta_{AC} + x\pi_{AC}\bigr)
\notag \\
&\qquad\qquad\qquad\qquad\quad
+ (1-\alpha)\log\bigl((1-x)\sigma_{BC} + x\pi_{BC}\bigr)
+ (\alpha-1)\log \omega_C
\Big)
\Big]
\notag \\
&\qquad\qquad \times
\bigl((1-x)\eta_{AC} + x\pi_{AC} + t\bigr)^{-1}
\, dt
\Bigg]\\
& \qquad \geq \Tr\Big[
\pi_{AC}
\bigl((1-x)\eta_{AC} + x\pi_{AC}\bigr)^{-1}  e^{-k}(d_Ad_Bd_C^2)^{\alpha-1} d_B I_{AC}x^{2(1-\alpha)}\Big]\\
&\qquad \geq \Tr\Big[
\pi_{AC} e^{-k}(d_Ad_Bd_C^2)^{\alpha-1} d_B I_{AC}P(\eta_{AC})^c\Big]x^{2(1-\alpha)-1},
\end{align}
which diverges to $\infty$ as $x$ goes to zero for $\alpha \in (1/2,1)$. Thus the trace exponential is strictly increasing away from the
purported boundary maximizer for all sufficiently small $x$.
As $x\downarrow0$, the trace exponential converges to its boundary value,
so this strict increase contradicts optimality.
\end{proof}

\begin{remark}
\label{rem:additivity-limits}
For fixed $\rho$, the auxiliary-state domain in
Corollary~\ref{cor: variational form triply optimized} is compact and
$D(\tau\|\rho)$ is continuous and bounded on it. Hence
$\alpha\mapsto I_\alpha^{\flat,*}(A:B|C)_\rho$ is continuous on
$(0,1)$ and $(1,2]$: the dependence on $\alpha$ is only through the
continuous coefficient $\alpha/(1-\alpha)$, uniformly over that
compact set. Continuity at one follows from
Lemma~\ref{lem: limit Istar}. This justifies taking limits at
$\alpha=1/2,1,2$ after proving additivity in the open intervals.

For states that are not full rank, the expressions are understood
through the limiting definitions above. The compactness argument in
Corollary~\ref{cor: variational form triply optimized} also applies to
products of depolarized inputs, since the relative entropy converges
on the limiting support. Thus additivity extends to arbitrary states.
\end{remark}

\subsection{Additivity of \texorpdfstring{$I_\alpha^{\flat,*}$ for $\alpha \in [1/2,1]$}{the optimized LE-CMI for alpha in [1/2,1]}}

As discussed above, to establish additivity, we lower bound the triply optimized log-Euclidean expression by restricting to product states of the form \( \omega_C \otimes \omega'_{C'} \), which can be chosen to be full rank. We then show that the resulting expression is additive for any such full-rank product state. This follows by exploiting that the infima are attained at full-rank states, hence in the interior of the feasible set, allowing us to apply first-order optimality conditions which yield a fixed-point equation. From this characterization, additivity readily follows. Superadditivity is then established by taking the supremum over all product states \( \omega_C \otimes \omega'_{C'} \) and exchanging it with the infima via a minimax argument, justified by the joint convexity established above.
\begin{theorem}
\label{thm: optimized additivity}
    Let $\rho_{ABC}$ and $\rho' _{A' B' C' }$ be two quantum states, and let $\alpha \in [1/2,1]$. Then, we have that
    \begin{align}
        I_\alpha^{\flat,*}(AA' :BB'|CC')_{\rho\otimes \rho' } = I_\alpha^{\flat,*}(A:B|C)_\rho+ I_\alpha^{\flat,*}(A' :B'|C')_{\rho'}\, .
    \end{align}
\end{theorem}
\begin{proof}
We first take $\alpha\in(1/2,1)$ and assume that $\rho_{ABC}$ and $\rho'_{A'B'C'}$ are full rank. We first note that in the variational characterization in~\eqref{eq: variational triply optimized}, we can restrict the infimum over all states to tensor-product states. This immediately yields subadditivity
\begin{align}
    I_\alpha^{\flat,*}(AA' :BB'|CC')_{\rho\otimes \rho'} \leq I_\alpha^{\flat,*}(A:B|C)_\rho+ I_\alpha^{\flat,*}(A' :B'|C')_{\rho'}\, .
\end{align}

We are then left to prove superadditivity.
To prove it, we can first lower bound the expression with the tensor product of two full-rank states $\omega_C \otimes \omega' _{C'}$ 
\begin{align}
    &I_\alpha^{\flat,*}(AA' :BB'|CC')_{\rho\otimes \rho'} \geq \inf_{\sigma_{BB' CC'},\eta_{AA' CC'}} \Delta_\alpha^\flat(\rho_{ABC}\otimes \rho'_{A' B' C'},\eta_{AA'CC'},\sigma_{BB'CC'},\omega_C \otimes \omega' _{C'})\, .
\end{align}
Next, we show that the above infimum is additive by proving that the tensor product of the marginal problems is a solution. 
In particular, since the function is jointly convex in the range $\alpha\in [1/2,1)$ by Lemma~\ref{lem: Joint concavity}, this amounts to proving that the derivatives are separately non-negative for the tensor products of the optimizers (the objective is differentiable at these full-rank states, so its differential is the sum of its two partial differentials).
We show the case of the derivative for the optimization problem in $\eta_{AA'CC'}$ as the one for $\sigma_{BB' CC'}$ is similar.

Let $(\eta_{AC}^\star,\sigma_{BC}^\star)$ and $(\eta'^\star _{A' C' },\sigma'^\star _{B' C' })$ be any solution of the marginal minimization problems
\begin{align}
    \min_{\eta_{AC},\sigma_{BC}} \Delta_\alpha^\flat(\rho_{ABC},\eta_{AC},\sigma_{BC},\omega_C) \,,\quad  \min_{\eta'_{A'C'},\sigma'_{B'C'}} \Delta_\alpha^\flat(\rho'_{A'B'C'},\eta'_{A'C'},\sigma'_{B'C'},\omega'_{C'})\,,
\end{align}
respectively. By Lemma~\ref{lem: lower-semicont Delta}, the objective function is lower semicontinuous, and any lower semicontinuous function optimized over a compact set attains its minimum on that set. Moreover, by Lemma~\ref{lem: full-rank}, all optimizers are full rank.
Since the objective functions are convex in each variable by
Lemma~\ref{lem: Joint concavity}, and since the optimizers are full rank
and therefore lie in the interior of the feasible set, the derivative in
every direction must vanish. By Lemma~\ref{lem; derivative Q}, this yields the corresponding fixed-point equations; see also~\cite[Theorem~5]{rubboli2024quantum}. For example, $\eta^\star_{AC}$ must satisfy
\begin{align}
    \eta_{AC}^\star
    \propto
    \Tr_B \exp\!\left(
        \alpha \log\rho_{ABC}
        +(1-\alpha)\log \eta^\star_{AC}
        +(1-\alpha)\log\sigma^\star_{BC}
        +(\alpha-1)\log\omega_C
    \right),
    \label{eq:fixed-point-eq-eta}
\end{align}
where, for a state $\tau$ and a positive semidefinite operator $A$, the notation
$\tau \propto A$ means that $\tau=A/\Tr[A]$.

By differentiability, the partial derivatives at
\((\eta^\star_{AC}\otimes \eta'^\star_{A'C'},\sigma_{BC}^\star\otimes \sigma'^\star_{B'C'})\)
may be considered separately in the directions
\(\eta_{AA'CC'}\) and \(\sigma_{BB'CC'}\).
It therefore suffices to verify the equality from
Lemma~\ref{lem; derivative Q} on the enlarged system \(AA'CC'\). This equality makes the corresponding directional derivative zero.
 We then have that the condition of Lemma~\ref{lem; derivative Q} is satisfied since
\begin{align}
&\Tr\Bigg[
\eta_{AA'CC'}
\int_{0}^{\infty}
(\eta_{AC}^\star \otimes \eta'^\star_{A'C'} + t)^{-1}
\notag \\
&\qquad \times
\Tr_{BB' }\Bigg[
\exp\Big(
\alpha \log(\rho_{ABC} \otimes \rho'_{A'B'C'})
+ (1-\alpha)\log(\sigma^\star_{BC} \otimes \sigma'^\star_{B'C'})
\notag \\
&\qquad\qquad\quad
+ (\alpha-1)\log(\omega_C \otimes \omega'_{C'})
+ (1-\alpha)\log(\eta^\star_{AC} \otimes \eta'^\star_{A'C'})
\Big)
\Bigg]
\notag \\
&\qquad \times
(\eta_{AC}^\star \otimes \eta'^\star_{A'C'} + t)^{-1}
\, dt
\Bigg]
\notag \\[0.5em]
&=
\Tr\Bigg[
\eta_{AA'CC'}
\int_{0}^{\infty}
(\eta^\star_{AC} \otimes \eta'^\star_{A'C'} + t)^{-1}
(\eta_{AC}^\star \otimes \eta'^\star_{A'C'})
\notag \\
&\qquad \times
(\eta^\star_{AC} \otimes \eta'^\star_{A'C'} + t)^{-1}
\, dt
\Bigg]
\notag \\
&\qquad \times
Q_{\alpha}^\flat(\rho_{ABC},\eta^\star_{AC},\sigma^\star_{BC},\omega_C)\,
Q_{\alpha}^\flat(\rho'_{A'B'C'},\eta'^\star_{A'C'},\sigma'^\star_{B'C'},\omega'_{C'})
\\[0.5em]
&=
Q_{\alpha}^\flat(\rho_{ABC},\eta^\star_{AC},\sigma^\star_{BC},\omega_C)\,
Q_{\alpha}^\flat(\rho'_{A'B'C'},\eta'^\star_{A'C'},\sigma'^\star_{B'C'},\omega'_{C'}) \, ,
\end{align}
where the equality follows from \eqref{eq:fixed-point-eq-eta} and a similar fixed-point equation for $\eta'^\star_{A'C'}$ and we defined
\begin{align}
    Q_{\alpha}^\flat(\rho_{ABC},\eta_{AC},\sigma_{BC},\omega_C)\coloneqq \exp\big((\alpha-1) \,\Delta_\alpha^\flat(\rho_{ABC},\eta_{AC},\sigma_{BC},\omega_C)\big)\,.
\end{align}
Hence both partial directional derivatives of the jointly convex function $\Delta_\alpha^\flat$ vanish. The same calculation applies to the $BB'CC'$ variable. Joint differentiability and convexity then prove global optimality and hence additivity for this joint minimization problem.
Overall, we obtain
\begin{align}
    &I_\alpha^{\flat,*}(AA' :BB'|CC')_{\rho\otimes \rho'} \notag \\
    & \qquad \geq \sup_{\omega_C}\inf_{\sigma_{B C},\eta_{A C}} \Delta_\alpha^\flat(\rho_{ABC},\eta_{AC},\sigma_{BC},\omega_C )+\sup_{\omega'_{C'}}\inf_{\sigma'_{B' C'},\eta'_{A' C'}} \Delta_\alpha^\flat( \rho'_{A' B' C'},\eta'_{A'C'},\sigma'_{B'C'}, \omega' _{C'})\\
    & \qquad =I_\alpha^{\flat,*}(A:B|C)_\rho+ I_\alpha^{\flat,*}(A' :B'|C')_{\rho'}\, .
\end{align}
It remains to justify the minimax exchange used in the last
equality. Write $c_\alpha=\alpha/(1-\alpha)\ge1$. By
Lemma~\ref{lem: variational form}, commuting the three infima gives
\begin{align}
 &\sup_{\omega_C}\inf_{\eta_{AC},\sigma_{BC}}
       \Delta_\alpha^\flat(\rho,\eta,\sigma,\omega)=\sup_{\omega_C}\inf_{\tau_{ABC}}
   \{c_\alpha D(\tau_{ABC}\|\rho_{ABC})
       +I(A:B|C)_\tau-D(\tau_C\|\omega_C)\}.
\end{align}
The objective can be written as
\begin{align}
 &-(c_\alpha-1)H(ABC)_\tau-H(B|AC)_\tau-H(A|BC)_\tau-c_\alpha\Tr[\tau_{ABC}\log\rho_{ABC}]
                  +\Tr[\tau_C\log\omega_C].
\end{align}
It is convex and continuous in $\tau_{ABC}$, since ordinary entropy
and conditional entropy are concave, and it is concave and continuous
in the full-rank state $\omega_C$. The $\tau$ state space is compact
and convex. Sion's minimax theorem therefore exchanges these last
two optimizations. Since
$\sup_{\omega_C>0}-D(\tau_C\|\omega_C)=0$, the result is
$\inf_\tau\{c_\alpha D(\tau\|\rho)+I(A:B|C)_\tau\}
=I_\alpha^{\flat,*}(A:B|C)_\rho$ by
Corollary~\ref{cor: variational form triply optimized}. The infima
over full-rank $\eta$ and $\sigma$ used here have value zero for the
corresponding marginal relative entropies, approached by full-rank
approximations if their marginal optimizers are singular.

The endpoint $\alpha=1/2$ follows by taking
$\alpha\downarrow1/2$, using the continuity established in
Remark~\ref{rem:additivity-limits}. The endpoint $\alpha=1$ follows
from Lemma~\ref{lem: limit Istar}. 
\end{proof}

\subsection{Additivity of \texorpdfstring{$I_\alpha^{\flat,*}$ for $\alpha\in [1,2]$}{the optimized LE-CMI for alpha in [1,2]}}
Next, we show additivity for $\alpha \in [1,2]$.
\begin{theorem}
    Let $\rho_{ABC}$ and $\rho'_{A'B'C'}$ be two quantum states and $\alpha \in [1,2]$. Then, we have that
    \begin{align}
        I_\alpha^{\flat,*}(AA' :BB'|CC')_{\rho\otimes \rho'} = I_\alpha^{\flat,*}(A:B|C)_\rho+ I_\alpha^{\flat,*}(A' :B'|C')_{\rho'}\, .
    \end{align}
\end{theorem}

\begin{proof}
We first take $\alpha\in(1,2)$ and assume that $\rho_{ABC}$ and $\rho'_{A'B'C'}$ are full rank. We first note that in the variational characterization in~\eqref{eq: variational triply optimized}, we can restrict the supremum over all states to tensor-product states. This immediately yields superadditivity
\begin{align}
    I_\alpha^{\flat,*}(AA' :BB'|CC')_{\rho\otimes \rho'} \geq I_\alpha^{\flat,*}(A:B|C)_\rho+ I_\alpha^{\flat,*}(A' :B'|C')_{\rho'}\, .
\end{align}

We are then left to prove subadditivity.
To prove it, we can first upper bound the expression with the tensor product of full-rank states $\eta_{AC} \otimes \eta' _{A'C'}$ and $\sigma_{BC} \otimes \sigma' _{B'C'}$ 
\begin{align}
    &I_\alpha^{\flat,*}(AA' :BB'|CC')_{\rho\otimes \rho'} \leq \sup_{\omega_{CC'}} \Delta_\alpha^\flat(\rho_{ABC}\otimes \rho'_{A' B' C'},\eta_{AC} \otimes \eta' _{A'C'},\sigma_{BC} \otimes \sigma' _{B'C'},\omega_{CC'})\, .
\end{align}
Next, we show that the above supremum is additive by proving that the tensor product of the marginal problems $ \omega^\star_{C} \otimes \omega'^\star_{C'}$ is a solution. In particular, we prove that the tensor products of the fixed-point solutions satisfy the required condition for optimality.
Let $\omega^\star_{C}$ and $\omega'^\star_{C'}$ be arbitrary solutions of the marginal maximization problems
\begin{align}
    \max_{\omega_C} \Delta_\alpha^\flat(\rho_{ABC},\eta_{AC},\sigma_{BC},\omega_{C}) \,,\quad  \max_{\omega'_{C'}} \Delta_\alpha^\flat(\rho'_{A'B'C'},\eta'_{A'C'},\sigma'_{B'C'},\omega'_{C'}) \,,
\end{align}
respectively. By Lemma~\ref{lem: lower-semicont Delta}, the objective function is upper semicontinuous, and any upper semicontinuous function maximized over a compact set attains its maximum on that set. Moreover, by Lemma~\ref{lem: full-rank Q}, all optimizers are full rank. Since the objective function is concave in each variable by Lemma~\ref{lem: Joint concavity}, and since the optimizers are full rank and therefore lie in the interior of the feasible set, the derivative in every direction must vanish.  By Lemma~\ref{lem; derivative Q}, this yields the corresponding fixed-point equations; see also~\cite[Theorem~5]{rubboli2024quantum}. For example, $\omega^\star_{C}$ must satisfy
\begin{align}
    \omega^\star_{C} \propto \Tr_{AB}\exp\big(\alpha \log\rho_{ABC}+(1-\alpha)\log \eta_{AC} +(1-\alpha)\log{\sigma_{BC}} +(\alpha-1)\log\omega^\star_C \big)\,,
\end{align}
where, for a state $\tau$ and a positive semidefinite operator $A$, the notation
$\tau \propto A$ means that $\tau=A/\Tr[A]$.

 We then have that the tensor product of the marginal problems $\omega^\star_{C}\otimes \omega'^\star_{C'}$ is an optimizer since the condition of Lemma~\ref{lem; derivative Q} is satisfied. Indeed, we have 
\begin{align}
&\Tr\Bigg[
\omega_{CC'}
\int_{0}^{\infty}
(\omega^\star_{C} \otimes \omega'^\star_{C'} + t)^{-1}
\notag \\
&\qquad \times
\Tr_{AA'BB' }\Bigg[
\exp\Big(
\alpha \log(\rho_{ABC} \otimes \rho'_{A'B'C'})
+ (1-\alpha)\log(\sigma_{BC} \otimes \sigma'_{B'C'})
\notag \\
&\qquad\qquad\quad
+ (\alpha-1)\log(\omega^\star_C \otimes \omega'^\star_{C'})
+ (1-\alpha)\log(\eta_{AC} \otimes \eta'_{A'C'})
\Big)
\Bigg]
\notag \\
&\qquad \times
(\omega^\star_{C} \otimes \omega'^\star_{C'} + t)^{-1}
\, dt
\Bigg]
\notag \\[0.5em]
&=
\Tr\Bigg[
\omega_{CC'}
\int_{0}^{\infty}
(\omega^\star_{C} \otimes \omega'^\star_{C'} + t)^{-1}
(\omega^\star_{C} \otimes \omega'^\star_{C'})
(\omega^\star_{C} \otimes \omega'^\star_{C'} + t)^{-1}
\, dt
\Bigg]
\notag \\
&\qquad \times
Q_{\alpha}^\flat(\rho_{ABC},\eta_{AC},\sigma_{BC},\omega^\star_C)\,
Q_{\alpha}^\flat(\rho'_{A'B'C'},\eta'_{A'C'},\sigma'_{B'C'},\omega'^\star_{C'})
\\[0.5em]
&=
Q_{\alpha}^\flat(\rho_{ABC},\eta_{AC},\sigma_{BC},\omega^\star_C)\,
Q_{\alpha}^\flat(\rho'_{A'B'C'},\eta'_{A'C'},\sigma'_{B'C'},\omega'^\star_{C'}) \, ,
\end{align}
where we defined
\begin{align}
    Q_{\alpha}^\flat(\rho_{ABC},\eta_{AC},\sigma_{BC},\omega_C)= \exp\big((\alpha-1)\,\Delta_\alpha^\flat(\rho_{ABC},\eta_{AC},\sigma_{BC},\omega_C)\big)\,.
\end{align}
Hence, we obtain the condition of the derivative being non-positive for concave functions. 
Overall, we obtain
\begin{align}
    &I_\alpha^{\flat,*}(AA' :BB'|CC')_{\rho\otimes \rho'} \notag \\
    & \qquad \leq \inf_{\eta_{AC},\sigma_{BC}}\sup_{\omega_C}\Delta_\alpha^\flat(\rho_{ABC},\eta_{AC},\sigma_{BC},\omega_C )+\inf_{\eta'_{A' C' },\sigma'_{B' C' }}\sup_{\omega'_{C'}}\Delta_\alpha^\flat( \rho'_{A' B' C'},\eta'_{A'C'},\sigma'_{B'C'}, \omega' _{C'})\\
    & \qquad =I_\alpha^{\flat,*}(A:B|C)_\rho+ I_\alpha^{\flat,*}(A' :B'|C')_{\rho'}\, .
\end{align}

At $\alpha=2$, full rank of the optimizing $\omega_C$ need
not hold, so we do not apply the fixed-point argument directly.
Instead, take $\alpha\uparrow2$ in the already established equality.
Remark~\ref{rem:additivity-limits} justifies the limit on each of
its three terms. The case $\alpha=1$ follows from
Lemma~\ref{lem: limit Istar}.
\end{proof}

\section{Conditional independence and recovery bounds}

\subsection{Conditional independence}
Quantum Markov chains in the order \(A-C-B\) are tripartite states that
can be reconstructed from their \(AC\)-marginal by acting only on the
conditioning system \(C\). Equivalently, \(\rho_{ABC}\) is a quantum
Markov chain in this order if there exists a quantum channel
\(\mathcal{R}_{C\to BC}\) such that
\begin{align}
    \mathcal{R}_{C\to BC}(\rho_{AC})
    =
    \rho_{ABC}.
\end{align}
It is known that this condition is equivalent to the vanishing of the
conditional mutual information \cite{hayden04CMI}, namely
\begin{align}
    I(A:B|C)_\rho = 0 .
\end{align}
Quantum
Markov chains admit a structural characterisation in terms of
a direct-sum decomposition~\cite{hayden04CMI}
\begin{align}\label{eq:direct-sum}
    \rho_{ABC}=\bigoplus_k p_k\, \rho^k_{AC_k^L}\otimes\rho^k_{C_k^RB}\,,
\end{align}
with respect to some induced decomposition $C=\bigoplus_kC_k^L\otimes C_k^R$ and some probability distribution $(p_k)_k$.
The next lemma shows that the LE-CMI faithfully captures exact conditional independence. In particular, it vanishes precisely when the underlying state is a quantum Markov chain.

\begin{lemma}
\label{lem: Markov chain optimized}
Let $\rho_{ABC}$ be a quantum state, and let $\alpha \in (0,1)$.
Then
\begin{align}
I_{\alpha}^{\flat,*}(A:B|C)_\rho = 0
\end{align}
if and only if $\rho_{ABC}$ is a quantum Markov chain.
\end{lemma}

\begin{proof}
It is sufficient to prove that $I_{\alpha}^{\flat,*}(A:B|C)_\rho = 0 $ if and only if $
I(A:B|C)_\rho = 0$. Indeed, the condition $I(A:B|C)_\rho=0$
is equivalent to $\rho_{ABC}$ being a quantum Markov chain.

We first prove the implication ``$\Longleftarrow$''. Suppose that $I(A:B|C)_\rho=0$.
Evaluating the variational expression at $\tau_{ABC}=\rho_{ABC}$ yields
\begin{align}
I_{\alpha}^{\flat,*}(A:B|C)_\rho
\le
\frac{\alpha}{1-\alpha}D(\rho_{ABC}\|\rho_{ABC})
+
I(A:B|C)_\rho
=0.
\end{align}
The non-negativity of $I_{\alpha}^{\flat,*}(A:B|C)_\rho$ obtained in Lemma~\ref{lem: non-negativity} implies that $I_{\alpha}^{\flat,*}(A:B|C)_\rho=0$.

We now prove the implication ``$\Longrightarrow$''. Assume that $I_{\alpha}^{\flat,*}(A:B|C)_\rho=0$.
Choose a minimizing state $\tau_{ABC}$ in the variational expression; its existence follows from Remark~\ref{rem: full-support CE and MI}. Since both terms are nonnegative and $\frac{\alpha}{1-\alpha}>0$, it follows that $D(\tau_{ABC}\|\rho_{ABC})= 0$ and $I(A:B|C)_{\tau} = 0$.
Since the relative entropy is zero, its faithfulness implies that $\tau_{ABC}=\rho_{ABC}$ and hence $I(A:B|C)_{\rho}=0$. This proves the claim.
\end{proof}

\begin{lemma}
\label{lem: Markov chain}
Let $\rho_{ABC}$ be a quantum state, and let $\alpha \in (0,1)$.
Then
\begin{align}
I_{\alpha}^{\flat}(A:B|C)_\rho = 0
\end{align}
if and only if $\rho_{ABC}$ is a quantum Markov chain.
\end{lemma}
\begin{proof}
    The proof is similar to that of Lemma~\ref{lem: Markov chain optimized} with the additional observations that $D(\tau_{AC}\|\rho_{AC})-D(\tau_C\|\rho_C)\geq0$ by data processing and $D(\tau_{BC}\|\rho_{BC})\geq0$. The minimum is again attained, so vanishing forces the minimizing state to equal $\rho_{ABC}$.
\end{proof}

\subsection{Variational expression for convex trace functions}
In this section, we derive several variational expressions for trace functions involving exponentials of operators. These formulas will be particularly useful later for deriving recovery bounds for the LE-CMI, but they may also be of independent interest. The proofs below follow directly from convexity of the trace exponential and its logarithm. These formulas are also consistent with the \(z\to\infty\) limits of the variational formulas for convex trace functions derived in~\cite{frank2013monotonicity}; see also~\cite[Theorem~3.3]{zhang2020wigner}.

We begin by deriving an additive form, expressed in terms of a sum of trace functions.
\begin{lemma}
    Let $\{X_i\}_{i=1}^n$ be a set of Hermitian operators, and let $\{\alpha_i\}_{i=1}^n$ be non-negative coefficients such that $\sum_{i=1}^n \alpha_i = 1$. Then,
    \begin{align}
\Tr\!\left[\exp\!\left(\sum_{i=1}^n \alpha_i X_i\right)\right]
=
\inf_{\substack{Z_1,\dots,Z_n, \\ \sum_i \alpha_i Z_i = 0}}
\sum_{i=1}^n \alpha_i \Tr[\exp(X_i+Z_i)] .
\end{align}
    Here the optimization is taken over Hermitian operators labeled by $Z_i$.
\end{lemma}

\begin{proof}
We first show the inequality ``$\leq$''. Let $\{Z_i\}_{i=1}^n$ be a set of Hermitian operators such that $\sum_i \alpha_i Z_i = 0$. Since the trace functional $M \mapsto \Tr(f(M))$ inherits convexity from the function $f$ (see, e.g.,~\cite{carlen2010trace}), we have
\begin{align}
\Tr\!\left[\exp\!\left(\sum_{i=1}^n \alpha_i (X_i + Z_i)\right)\right]
\leq
\sum_{i=1}^n \alpha_i \Tr[\exp(X_i+Z_i)].
\end{align}
Using $\sum_i \alpha_i Z_i = 0$, this becomes
\begin{align}
\Tr\!\left[\exp\!\left(\sum_{i=1}^n \alpha_i X_i\right)\right]
\leq
\sum_{i=1}^n \alpha_i \Tr[\exp(X_i+Z_i)].
\end{align}
Taking the infimum over all admissible $\{Z_i\}$ yields
\begin{align}
\Tr\!\left[\exp\!\left(\sum_{i=1}^n \alpha_i X_i\right)\right]
\leq
\inf_{\substack{Z_1,\dots,Z_n \\ \sum_i \alpha_i Z_i = 0}}
\sum_{i=1}^n \alpha_i \Tr[\exp(X_i+Z_i)].
\end{align}

For the reverse inequality ``$\geq$'', consider the choice
\begin{align}
Z_i =  - X_i+ \sum_{j=1}^n \alpha_j X_j,
\end{align}
which is Hermitian and satisfies $\sum_i \alpha_i Z_i = 0$. Then,
\begin{align}
X_i + Z_i = \sum_{j=1}^n \alpha_j X_j,
\end{align}
so that
\begin{align}
\sum_{i=1}^n \alpha_i \Tr[\exp(X_i+Z_i)]
=
\sum_{i=1}^n \alpha_i \Tr\!\left[\exp\!\left(\sum_{j=1}^n \alpha_j X_j\right)\right]
=
\Tr\!\left[\exp\!\left(\sum_{j=1}^n \alpha_j X_j\right)\right].
\end{align}
This shows that the infimum is upper bounded by the left-hand side, yielding the reverse inequality.

Combining both inequalities proves the claim.
\end{proof}

\begin{corollary}
Let $X_1$ and $X_2$ be Hermitian operators, and let $\alpha \in (0,1]$. Then,
\begin{align}
\label{eq:var_exp_trace}
\Tr\!\left[\exp\!\left(\alpha X_1 + (1-\alpha) X_2\right)\right]
=
\inf_{\omega >0}
\alpha \Tr\!\left[\exp(X_1+\frac{\alpha-1}{\alpha}\log \omega)\right] +(1-\alpha) \Tr\!\left[\exp(X_2+\log \omega)\right].
\end{align}
\end{corollary}

\begin{proof}
Setting $Z_i=\log P_i$, the optimization over Hermitian operators $Z_i$ can be equivalently reformulated as an optimization over positive definite operators $P_i$. The constraint $\sum_i \alpha_i Z_i=0$ becomes $\sum_i \alpha_i \log P_i=0$. The claim then follows directly from the additive form; in the case of two operators, this constraint simplifies to $\alpha \log P_1 + (1-\alpha)\log P_2=0$, which is equivalent to $P_1 = P_2^{\frac{\alpha-1}{\alpha}}$.
\end{proof}

Next, we derive a multiplicative variational form.

\begin{lemma}[Multiplicative variational form]
Let $\{X_i\}_{i=1}^n$ be a set of Hermitian operators, and let $\{\alpha_i\}_{i=1}^n$ be non-negative coefficients such that $\sum_{i=1}^n \alpha_i=1$. Then,
\begin{align}
\Tr\!\left[\exp\!\left(\sum_{i=1}^n \alpha_i X_i\right)\right]
=
\inf_{\substack{Z_1,\dots,Z_n, \\ \sum_i \alpha_i Z_i=0}}
\prod_{i=1}^n
\Tr[\exp(X_i+Z_i)]^{\alpha_i},
\end{align}
where the optimization is over Hermitian operators labeled by $Z_i$.
\end{lemma}

\begin{proof}
Let $Z_1,\dots,Z_n$ be Hermitian operators such that $\sum_i \alpha_i Z_i=0$. By the convexity of the function $X \mapsto \log \Tr[e^X]$, we have
\begin{align}
\log \Tr\!\left[\exp\!\left(\sum_{i=1}^n \alpha_i (X_i+Z_i)\right)\right]
\leq
\sum_{i=1}^n \alpha_i \log \Tr[\exp(X_i+Z_i)].
\end{align}
Using $\sum_i \alpha_i Z_i=0$ and exponentiating both sides gives
\begin{align}
\Tr\!\left[\exp\!\left(\sum_{i=1}^n \alpha_i X_i\right)\right]
\leq
\prod_{i=1}^n
\Tr[\exp(X_i+Z_i)]^{\alpha_i}.
\end{align}
Taking the infimum over all admissible $Z_i$ yields one inequality.

For the reverse inequality, choose
\begin{align}
Z_i= -X_i + \sum_{j=1}^n \alpha_j X_j.
\end{align}
Then $\sum_i \alpha_i Z_i=0$ and
\begin{align}
X_i+Z_i=\sum_{j=1}^n \alpha_j X_j
\end{align}
for every $i$. Hence,
\begin{align}
\prod_{i=1}^n
\Tr[\exp(X_i+Z_i)]^{\alpha_i}
=
\prod_{i=1}^n
\Tr\!\left[\exp\!\left(\sum_{j=1}^n \alpha_j X_j\right)\right]^{\alpha_i}
=
\Tr\!\left[\exp\!\left(\sum_{j=1}^n \alpha_j X_j\right)\right].
\end{align}
Thus the infimum is also upper bounded by the left-hand side, and the equality follows.
\end{proof}

\subsection{Recovery bounds}
\label{sec: recovery bounds}
In this section, we derive several recovery bounds. We first establish upper
and lower bounds on the log-Euclidean CMI in terms of relative entropies
 between the original state and the rotated Petz state recovered from the marginal
$\rho_{AC}$. These bounds involve the quantum extensions of relative entropy, denoted by
$D^\star$ and $D^{\text{\normalfont\leftmoon}}$ and introduced in
Section~\ref{sec: notation}. Both bounds are tight in the classical case, and
the lower bound is always stronger than the fidelity-based bound
of~\cite{junge2015universal}. We then establish a lower bound on the log-Euclidean CMI in terms of the
measured relative entropy between the original state and the recovered state.
This bound generalizes the recovery result of~\cite{sutter2017multivariate},
which is recovered in the limit $\alpha\to1$. It further shows that a small
log-Euclidean CMI guarantees the existence of a recovery map that approximately
reconstructs $\rho_{ABC}$ from $\rho_{AC}$, and hence implies approximate
conditional independence.

\begin{theorem}
\label{thm: integral recovery bounds}
Let $\rho_{ABC}$ be a quantum state. For every $\alpha\in(0,1)$,
\begin{align}
I_\alpha^\flat(A:B|C)_\rho
\geq\int_{\mathbb R}dt\, \beta_0(t)\, D_{\alpha,1-\alpha}(\rho_{ABC}\|\mathcal R^{[t]}_{C\to BC}(\rho_{AC})),
\label{eq: integral lower recovery}
\end{align}
Moreover, if $\rho_{ABC}$ is full-rank, for every $\alpha>1$
\begin{align}
I_\alpha^\flat(A:B|C)_\rho
\leq\int_{\mathbb R}dt\, \beta_0(t)\, D_{\alpha,\alpha-1}(\rho_{ABC}\|\mathcal R^{[t]}_{C\to BC}(\rho_{AC})).
\label{eq: integral upper recovery}
\end{align}
\end{theorem}

\begin{proof}
First suppose that $\rho_{ABC}>0$. For Hermitian $H_1,\ldots,H_n$ and $p>0$, the multivariate trace inequality of~\cite[Corollary~18]{hiai2017logmajorization} gives
\begin{align}
\log\Tr\!\left[\exp\!\left(p\sum_{j=1}^nH_j\right)\right]
\leq\int_{\mathbb R}dt\, \beta_0(t)
 \log\left\|\prod_{j=1}^n e^{(1+it)H_j}\right\|_p^p.
\label{eq: recovery multivariate}
\end{align}
Here $\left\|X\right\|_p^p=\operatorname{Tr}[(X^\dagger X)^{p/2}]$. The inequality holds for all $p>0$, including $0<p<1$, by applying~\cite[Corollary~18]{hiai2017logmajorization} with the trace norm and $f(x)=x^p$, for which $x\mapsto\log f(e^x)$ is convex.

For $0<\alpha<1$, choose
\begin{align}
p=2(1-\alpha),\quad
H_1=\frac{\alpha}{2(1-\alpha)}\log\rho_{ABC},\quad
H_2=\frac12\log\rho_{BC},\quad
H_3=-\frac12\log\rho_C,\quad
H_4=\frac12\log\rho_{AC}.
\end{align}
Unitary invariance of the Schatten norm gives
\begin{align}\label{eq: unitary inv}
\left\|\prod_{j=1}^4e^{(1+it)H_j}\right\|_{2(1-\alpha)}^{2(1-\alpha)}
=\Tr\!\left[\left(\rho_{ABC}^{\frac{\alpha}{2(1-\alpha)}}\mathcal R^{[t]}_{C\to BC}(\rho_{AC})
\rho_{ABC}^{\frac{\alpha}{2(1-\alpha)}}\right)^{1-\alpha}\right].
\end{align}
Dividing~\eqref{eq: recovery multivariate} by $\alpha-1$ proves~\eqref{eq: integral lower recovery}.

For $\alpha>1$, we use a similar argument. Choose
\begin{align}
p=2(\alpha-1),\quad H_1=\frac{\alpha}{2(\alpha-1)}\log\rho_{ABC},\quad
H_2=-\frac12\log\rho_{BC},\quad
H_3=\frac12\log\rho_C,\quad
H_4=-\frac12\log\rho_{AC}.
\end{align}
Applying~\eqref{eq: recovery multivariate}, using unitary invariance of the Schatten norm similar to~\eqref{eq: unitary inv}, and dividing by $\alpha-1$ yields 
\begin{align}
    I_\alpha^\flat(A:B|C)_\rho
\leq\int_{\mathbb R}dt\, \beta_0(t)\, D_{\alpha,\alpha-1}(\rho_{ABC}\|\mathcal R^{[-t]}_{C\to BC}(\rho_{AC})).
\end{align}
The desired upper bound in~\eqref{eq: integral upper recovery} follows by the change of variables, $t\mapsto -t$, in the above integral and noting that $\beta_0$ is even.

For states that are not full rank, the lower bound is understood through~\eqref{eq:flat-LE-renyi-CMI-limit}. Apply~\eqref{eq: recovery multivariate} with $\rho(\delta)$ in $H_1$ and the marginals of $\rho(\varepsilon)$ in the remaining terms, and take $\delta\downarrow0$ before $\varepsilon\downarrow0$. The recovered states converge to $\mathcal R^{[t]}_{C\to BC}(\rho_{AC})$ for each $t$. Lower semicontinuity and nonnegativity of $D_{\alpha,1-\alpha}$ then give~\eqref{eq: integral lower recovery} by Fatou's lemma.
\end{proof}
\begin{corollary}
\label{cor: integral recovery endpoints}
Let $\rho_{ABC}$ be a quantum state. Then, 
\begin{align}
I(A:B|C)_\rho
\geq\int_{\mathbb R}dt\, \beta_0(t)\, D^\star(\rho_{ABC}\|\mathcal R^{[t]}_{C\to BC}(\rho_{AC})),
\label{eq: keyl recovery bound}
\end{align}
Moreover, if $\rho_{ABC}>0$, then
\begin{align}
I(A:B|C)_\rho
\leq\int_{\mathbb R}dt\, \beta_0(t)\, D^{\text{\normalfont\leftmoon}}\!(\rho_{ABC}\|\mathcal R^{[t]}_{C\to BC}(\rho_{AC})).
\label{eq: moon recovery bound}
\end{align}
\end{corollary}

\begin{proof}
The lower bound follows by taking $\alpha\to1^-$ in~\eqref{eq: integral lower recovery}, using~\eqref{eq: recovery divergence limits} and Fatou's lemma. For the upper bound, the full-rank assumption gives $m>0$ such that $\mathcal R^{[t]}_{C\to BC}(\rho_{AC})\geq mI$ for every $t$. Monotonicity in $\alpha$~\cite[Lemma~24]{rubboli2024magic} gives, for $1<\alpha\leq2$,
\begin{align}
0\leq D_{\alpha,\alpha-1}(\rho_{ABC}\|\mathcal R^{[t]}_{C\to BC}(\rho_{AC}))
&\leq D_{2,1}(\rho_{ABC}\|\mathcal R^{[t]}_{C\to BC}(\rho_{AC}))\notag\\
&=\log\Tr\!\left[\rho_{ABC}^2\bigl(\mathcal R^{[t]}_{C\to BC}(\rho_{AC})\bigr)^{-1}\right]\notag\\
&\leq-\log m.
\end{align}
Dominated convergence therefore permits taking $\alpha\to1^+$ in~\eqref{eq: integral upper recovery}. In both cases, $\lim_{\alpha\to1}I_\alpha^\flat=I$.
\end{proof}
Both bounds in Theorem~\ref{thm: integral recovery bounds} and Corollary~\ref{cor: integral recovery endpoints} are tight for classical states.

\begin{remark}
Note that we stated the upper bounds in Theorem~\ref{thm: integral recovery bounds} and Corollary~\ref{cor: integral recovery endpoints} only for full-rank states. These bounds can nevertheless be extended to non-full-rank states by choosing a full-rank approximation path and taking the corresponding limit. However, we restrict the results to full-rank states since these limits may depend on the chosen path.
\end{remark}

\begin{remark}
The smallest value of $\alpha$ for which $D_{\alpha,\alpha-1}$ satisfies data processing is $\alpha=2$~\cite{zhang2020wigner}. At this value,
$D_{2,1}(\rho\|\sigma)=\widebar{D}_2(\rho\|\sigma)$, with $\widebar{D}_2$ being the Petz R\'enyi relative entropy of order two. Hence, for $\rho_{ABC}>0$ and $1<\alpha\leq2$,
\begin{align}
I(A:B|C)_\rho
\leq I_\alpha^\flat(A:B|C)_\rho\leq I_2^\flat(A:B|C)_\rho
\leq\int_{\mathbb R}dt\, \beta_0(t) \, \widebar{D}_2\!\left(\rho_{ABC}\middle\|\mathcal R^{[t]}_{C\to BC}(\rho_{AC})\right).
\label{eq: petz two recovery upper}
\end{align}
Another upper bound can be obtained via max-relative entropy. Indeed, $\lim_{\alpha\to\infty}D_{\alpha,\alpha-1}=D_{\max}$~\cite[Lemma~23]{rubboli2024magic}. Therefore, monotone convergence in~\eqref{eq: integral upper recovery} yields, for $\rho_{ABC}>0$,
\begin{align}
 I(A:B|C)_\rho
 \leq I_\infty^\flat(A:B|C)_\rho
 \leq\int_{\mathbb R}dt\, \beta_0(t)\, D_{\max}(\rho_{ABC}\|\mathcal R^{[t]}_{C\to BC}(\rho_{AC})),
 \label{eq: max recovery upper}
\end{align}
where $I_\infty^\flat:=\lim_{\alpha\to\infty}I_\alpha^\flat$. For a full-rank state, this limit is the largest eigenvalue of
$\log\rho_{ABC}-\log\rho_{AC}-\log\rho_{BC}+\log\rho_C$.
\end{remark}
We next establish a bound in terms of a single averaged recovery channel.

\begin{theorem}
\label{thm: lower bound recovery}
Let $\rho_{ABC}$ be a quantum state, and let $\alpha\in(0,1)$. Then
\begin{align}
I_\alpha^\flat(A:B|C)_\rho\geq D_\alpha^{\mathbb M}\!\left(\rho_{ABC}\middle\|\int_{\mathbb R}dt\, \beta_0(t)\, \mathcal{R}^{[t]}_{C\to BC}(\rho_{AC})\right)\,.\label{eq: measured recovery lower}
\end{align}
\end{theorem}
\begin{proof}
First suppose $\rho_{ABC}>0$. The variational formula~\eqref{eq:var_exp_trace}, followed by Golden--Thompson and the multivariate trace inequality~\cite[Eq.~(54)]{sutter2017multivariate}, gives
\begin{align}
&\Tr \exp(\alpha\log\rho_{ABC}+(1-\alpha)(\log\rho_{AC}+\log\rho_{BC}-\log\rho_C))\notag\\
&=\inf_{\omega>0}\left\{\alpha\Tr \exp\Big(\log\rho_{ABC}+\frac{\alpha-1}{\alpha}\log\omega\Big)
 +(1-\alpha)\Tr \exp(\log\omega+\log\rho_{AC}+\log\rho_{BC}-\log\rho_C)\right\}\notag\\
&\leq\inf_{\omega>0}\Bigg\{\alpha\Tr[\rho_{ABC}\omega^{\frac{\alpha-1}{\alpha}}]+(1-\alpha)\Tr\!\left[
 \left(\int_{\mathbb R}dt\, \beta_0(t)\mathcal R^{[t]}_{C\to BC}(\rho_{AC})\right)\omega\right]\Bigg\}\notag\\
&=Q_\alpha^{\mathbb M}\!\left(\rho_{ABC}\middle\|
 \int_{\mathbb R}dt\, \beta_0(t)\mathcal R^{[t]}_{C\to BC}(\rho_{AC})\right).
\end{align}
The last equality follows from~\eqref{eq: variational QalphaM}. Taking the logarithm and dividing by $\alpha-1<0$ proves~\eqref{eq: measured recovery lower}. For the states that are not full rank, the same ordered limits as in Theorem~\ref{thm: integral recovery bounds}, together with lower semicontinuity of the measured divergence, gives the desired result.
\end{proof}
\begin{remark}[Comparison of recovery bounds]
For $0<\alpha\leq1/2$, data processing of $D_{\alpha,1-\alpha}$
and convexity of the measured divergence in its second argument give
\begin{align}
\int_{\mathbb R}dt\, \beta_0(t)
D_{\alpha,1-\alpha}\!\left(
\rho_{ABC}\middle\|\mathcal R^{[t]}_{C\to BC}(\rho_{AC})
\right)&\geq\int_{\mathbb R}dt\, \beta_0(t)
D_\alpha^{\mathbb M}\!\left(
\rho_{ABC}\middle\|\mathcal R^{[t]}_{C\to BC}(\rho_{AC})
\right)\notag\\
&\geq D_\alpha^{\mathbb M}\!\left(
\rho_{ABC}\middle\|
\int_{\mathbb R}dt\, \beta_0(t)
\mathcal R^{[t]}_{C\to BC}(\rho_{AC})
\right).
\label{eq: integral comparison small alpha}
\end{align}
Thus~\eqref{eq: integral lower recovery} is at least as strong as
\eqref{eq: measured recovery lower} in this range.

To compare with fidelity-based recovery bounds, recall from
Section~\ref{sec: notation} that
$D^\star(\rho\|\sigma)\geq-\log F(\rho,\sigma)$. Hence,
\begin{align}
I(A:B|C)_\rho
&\geq\int_{\mathbb R}dt\, \beta_0(t)
D^\star\!\left(
\rho_{ABC}\middle\|\mathcal R^{[t]}_{C\to BC}(\rho_{AC})
\right)\notag\\
&\geq-\int_{\mathbb R}dt\, \beta_0(t)
\log F\!\left(
\rho_{ABC},\mathcal R^{[t]}_{C\to BC}(\rho_{AC})
\right)\notag\\
&\geq-\log F\!\left(
\rho_{ABC},
\int_{\mathbb R}dt\, \beta_0(t)
\mathcal R^{[t]}_{C\to BC}(\rho_{AC})
\right).
\label{eq: keyl fidelity recovery comparison}
\end{align}
The last inequality follows from concavity of the root fidelity
and Jensen's inequality. Thus, the $D^\star$ recovery bound is
at least as strong as the fidelity-based bound for the same
averaged recovery map.

For $1/2\leq\alpha<1$, the analogous comparison follows from
$D_{\alpha,1-\alpha}\geq D_{1/2,1/2}=-\log F$.
In particular, at $\alpha=1/2$,
\eqref{eq: integral lower recovery} gives the fidelity bound with
$I_{1/2}^\flat(A:B|C)_\rho\leq I(A:B|C)_\rho$
on the left-hand side.

We emphasize that neither the $D^\star$ bound
nor the measured-relative-entropy bound dominates the other in
general (see Appendix~\ref{app:numerical-recovery} for examples).
\end{remark}
In the next corollary, we use the result in Theorem~\ref{thm: lower bound recovery} to derive a recovery bound in terms of sandwiched relative entropies. The proof follows the strategy of~\cite[Theorem~4.1]{berta2021composite}.
\begin{corollary}
\label{cor: lower bound recovery sandwiched renyi}
Let $\rho_{ABC}$ be a quantum state, and let $\alpha \in [1/2,1]$.
Then,
\begin{align}
I^{\flat}_{\alpha}(A:B|C)_\rho
\ge\limsup_{n\to\infty}{\frac{1}{n}
\widetilde{D}_{\alpha}\!\left(\rho_{ABC}^{\otimes n}\middle\|\int_{-\infty}^{\infty}dt\,
\beta_0(t)\left(\mathcal{R}^{[t]}_{C\to BC}(\rho_{AC})\right)^{\otimes n}\right)}\,.
\end{align}
At $\alpha=1$, the divergences and the LE-CMI are understood by their limits.
\end{corollary}

\begin{proof}
    First, observe that the following relations hold
    \begin{align}
        I^{\flat}_{\alpha}(A:B|C)_\rho&=\frac{1}{n}I^{\flat}_{\alpha}(A^n:B^n|C^n)_{\rho^{\otimes n}} \\
        &\geq \frac{1}{n}D^{\mathbb M}_{\alpha}\!\left(\rho_{ABC}^{\otimes n}\middle\|\int_{-\infty}^{\infty}dt\,
\beta_0(t)\left(\mathcal{R}^{[t]}_{C \to BC}(\rho_{AC})\right)^{\otimes n}\right) \\
        &\geq \frac{1}{n}\widetilde{D}_{\alpha}\!\left(\rho_{ABC}^{\otimes n}\middle\|\int_{-\infty}^{\infty}dt\,
\beta_0(t)\left(\mathcal{R}^{[t]}_{C \to BC}(\rho_{AC})\right)^{\otimes n}\right)\notag \\
&\qquad  -\frac{2}{n}\log\left | \operatorname{spec}\!\left(\int_{-\infty}^{\infty}dt\,
\beta_0(t)\left(\mathcal{R}^{[t]}_{C \to BC}(\rho_{AC})\right)^{\otimes n}\right) \right| \\   
        &\geq \frac{1}{n}\widetilde{D}_{\alpha}\!\left(\rho_{ABC}^{\otimes n}\middle\|\int_{-\infty}^{\infty}dt\,
\beta_0(t)\left(\mathcal{R}^{[t]}_{C \to BC}(\rho_{AC})\right)^{\otimes n}\right)-\frac{2}{n}\log\operatorname{poly}(n),
    \end{align}
    where the first line holds by additivity of the LE-CMI, the second line follows from Theorem~\ref {thm: lower bound recovery}, the third line is an application of~\cite[Eqs.~(4.58)--(4.59)]{tomamichel2015quantum}, and the last line holds by noting that $\int_{-\infty}^{\infty}dt\,
\beta_0(t)\left(\mathcal{R}^{[t]}_{C \to BC}(\rho_{AC})\right)^{\otimes n}$ is a permutation-invariant state. For fixed $d=\dim(ABC)$, its number of distinct eigenvalues is at most $(n+1)^{d^2}$, so the polynomial correction divided by $n$ vanishes. The proof then holds by taking the limit superior $n\to\infty$.
\end{proof}

\subsection{Bounds in terms of the relative entropy distance to Markov states}

In this section, we derive a bound in terms of the relative entropy distance to the set of Markov states.

An important question is whether the conditional mutual information can be interpreted as a distance to the set of Markov states. Classically, the relative entropy distance to the set of Markov states coincides exactly with the conditional mutual information. In the quantum case, however, the relative entropy distance provides only an upper bound on the conditional mutual information~\cite{ibinson2008robustness}.

The next upper bound generalizes the result of~\cite{ibinson2008robustness}, which is recovered in the limit $\alpha \to 1$.
\begin{proposition}
\label{prop: upper bound recovery}
Let $\rho_{ABC}$ be a quantum state, and let $\alpha \in (0,1)$. 
Then, 
\begin{align}
I_{\alpha}^{\flat,*}(A:B|C)_\rho 
 \leq \inf_{\omega_{ABC} \in \mathcal{M}_{A:B|C} }D^\flat_{\alpha}(\rho_{ABC}\| \omega_{ABC}),
\end{align}
where $\mathcal{M}_{A:B|C}$ is the set of quantum Markov states on $ABC$ in the order $A-C-B$. Moreover, in the classical case, where all states in the optimizations are diagonal in the same fixed product basis, the inequality is saturated.
\end{proposition}
\begin{proof}
    We recall that
    \begin{align}
        I_{\alpha}^{\flat,*}(A:B|C)_\rho =\inf_{\tau_{ABC}}
\left\{
\frac{\alpha}{1-\alpha}D(\tau_{ABC}\|\rho_{ABC})
+
I(A:B|C)_\tau
\right\}.
    \end{align}
Next, we use the bound of the CMI in terms of the minimum relative entropy to the set of Markov states derived in~\cite[Theorem~4]{ibinson2008robustness}
\begin{align}
I(A:B|C)_\tau \leq \inf_{\omega_{ABC} \in \mathcal{M}_{A:B|C} }D(\tau_{ABC}\| \omega_{ABC}) \,.
\end{align}
Hence, we can rewrite
\begin{align}
    I_{\alpha}^{\flat,*}(A:B|C)_\rho \leq \inf_{\omega_{ABC} \in \mathcal{M}_{A:B|C} } \inf_{\tau_{ABC}}
\left\{
\frac{\alpha}{1-\alpha}D(\tau_{ABC}\|\rho_{ABC})
+
D(\tau_{ABC}\| \omega_{ABC})
\right\}.
\end{align}
Then, we use the variational form of the log-Euclidean relative entropy to obtain
\begin{align}
    \inf_{\tau_{ABC}}
\left\{
\frac{\alpha}{1-\alpha}D( \tau_{ABC} \|\rho_{ABC})
+
D(\tau_{ABC}\| \omega_{ABC})
\right\} = D^\flat_\alpha(\rho_{ABC} \| \omega_{ABC}).
\end{align}

Finally, in the classical case, the conditional mutual information of every state $\tau_{ABC}$ coincides with its relative-entropy distance to the set of classical Markov states~\cite{ibinson2008robustness}:
\begin{align}
    I(A:B|C)_\tau
    =
    \inf_{\omega_{ABC} \in \mathcal{M}_{A:B|C}}
    D\!\left(\tau_{ABC}\middle\| \omega_{ABC}\right).
\end{align}
In the above proof, all inequalities are equalities, and therefore, for this specific case,
\begin{align}
    I_{\alpha}^{\flat,*}(A:B|C)_\rho 
 = \inf_{\omega_{ABC} \in \mathcal{M}_{A:B|C} }D^\flat_{\alpha}(\rho_{ABC}\| \omega_{ABC}),
\end{align}
thus concluding the proof.
\end{proof}

Next, we derive a lower bound in terms of the relative entropy distance, measured with respect to one-way LOCC from \(A\) to \(B\), to the set of separable states. We note that Markov states become separable after tracing out the conditioning system. In contrast, the analogous lower bound by the relative entropy distance to exact quantum Markov states is false in general, even for the QCMI; analogous obstructions also occur for measured relative entropy~\cite{ibinson2008robustness,berta2024entanglement}.
At $\alpha=1$, the corresponding lower bound is the established QCMI bound~\cite{brandao2011faithful,li2018squashed,li2014relative,berta2024entanglement}.
\begin{proposition}
\label{prop: lower bound separable}
Let $\rho_{ABC}$ be a quantum state, and let $\alpha \in (0,1)$.
Then,
\begin{align}
I^{\flat,*}_{\alpha}(A:B|C)_\rho
\ge
\inf_{\sigma_{AB} \in \mathrm{Sep}(A:B)}D^{\mathrm{LOCC}_1(A\to B)}_{\alpha}\left(\rho_{AB}\middle\|\sigma_{AB}\right),
\end{align}
where $\mathrm{Sep}(A:B)$ is the set of separable states on $AB$.
\end{proposition}

\begin{proof}
      We recall that
    \begin{align}
        I_{\alpha}^{\flat,*}(A:B|C)_\rho =\inf_{\tau_{ABC}}
\left\{
\frac{\alpha}{1-\alpha}D(\tau_{ABC}\|\rho_{ABC})
+
I(A:B|C)_\tau
\right\}.
    \end{align}
    For every state $\tau_{ABC}$, we have the bound~\cite[Equation~(12)]{berta2024entanglement}
    \begin{align}
        I(A:B|C)_\tau \geq \inf_{\sigma_{AB} \in \mathrm{Sep}(A:B)}D^{\mathrm{LOCC}_1(A\to B)}\left(\tau_{AB}\middle\|\sigma_{AB}\right).
    \end{align}
    We can exchange the infima to obtain
    \begin{align}
        I_{\alpha}^{\flat,*}(A:B|C)_\rho \geq\inf_{\sigma_{AB} \in \mathrm{Sep}(A:B)}\inf_{\tau_{ABC}}
\left\{
\frac{\alpha}{1-\alpha}D(\tau_{ABC}\|\rho_{ABC})
+
D^{\mathrm{LOCC}_1(A\to B)}\left(\tau_{AB}\middle\|\sigma_{AB}\right)
\right\}.
    \end{align}
    We can use that $ \inf \sup \geq \sup \inf$ and the DPI under partial trace of the Umegaki relative entropy and under measurements to obtain 
\begin{multline}
        I_{\alpha}^{\flat,*}(A:B|C)_\rho \geq\\
        \inf_{\sigma_{AB} \in \mathrm{Sep}(A:B)}\sup_{\mathcal{M} \in \mathrm{LOCC}_1(A\to B)}\inf_{\tau_{ABC}}
\left\{
\frac{\alpha}{1-\alpha}D(\mathcal{M}(\tau_{AB})\|\mathcal{M}(\rho_{AB}))
+
D\left(\mathcal{M}(\tau_{AB})\middle\|\mathcal{M}(\sigma_{AB})\right)
\right\}.
    \end{multline}
We can then lower bound the expression by allowing $\mathcal M(\tau_{AB})$ to range over all probability distributions on the measurement outcome alphabet, and use the variational form of classical Rényi relative entropy to obtain that
\begin{align}
        I_{\alpha}^{\flat,*}(A:B|C)_\rho \geq\inf_{\sigma_{AB} \in \mathrm{Sep}(A:B)} \sup_{\mathcal{M} \in \mathrm{LOCC}_1(A\to B)}D_\alpha(\mathcal{M}(\rho_{AB})\|\mathcal{M}(\sigma_{AB})),
    \end{align}
    thus concluding the proof.
\end{proof}

\begin{remark}
    The same lower bound, expressed in terms of the measured relative entropy to the set of separable states, also applies to the non-optimized quantity $I_{\alpha}^{\flat}$, since by Lemma~\ref{lem: ordering LE-CMIs} this quantity is no smaller than the triply optimized one $I_{\alpha}^{\flat,*}$.
\end{remark}

\section{Chain rules}\label{sec: chain rules}
For tripartite quantum states, the QCMI satisfies the identities
\begin{align}
     I(A:B|C)_\rho
    = H(B|C)_\rho-H(B|AC)_\rho = I(AC:B)_{\rho} - I(C:B)_\rho .
\end{align}
Moreover, in the case of four parties, it holds that
\begin{align}
    I
    (A:B|CD)_\rho = I(A:BC|D)_\rho-I(A:C|D)_\rho \,.
\end{align}
These identities can be interpreted as chain rules, since they relate different
entropic quantities evaluated on the marginals of the same multipartite state. 

In the R\'enyi setting, entropic quantities in quantum information theory often satisfy analogues of the standard chain rules. Unlike in the latter setting, however, these R\'enyi chain rules typically take the form of inequalities, and often require the R\'enyi parameters to be adjusted according to specific parameter relations. This is the case, for instance, for chain rules involving R\'enyi conditional entropies and mutual informations
\cite{dupuis2015chain,rubboli2024quantum,mckinlay2020decomposition}.
In what follows, we establish a family of chain rules for the LE-CMI. To the best of our knowledge, these chain rules are novel even in the classical setting.

\subsection{Chain rules involving conditional entropies}

\begin{lemma}
Let \(\rho_{ABC}\) be a tripartite quantum state, and let
\(\alpha,\beta,\gamma \in (0,1)\cup (1,\infty)\). Suppose that
\((\alpha-1)(\beta-1)>0\).
If \(\gamma\in (0,1)\), it holds that
\begin{alignat}{2}
I_\alpha^{\flat,*}(A:B|C)_\rho
&\leq H_\gamma^{\flat,\uparrow}(B|C)_\rho
      -H_\beta^{\flat,\uparrow}(B|AC)_\rho
&\qquad&\text{when}\quad
\frac{\beta}{\beta-1}
=\frac{\alpha}{\alpha-1}+\frac{\gamma}{\gamma-1},
\;\;\text{and}\;\; \alpha\leq2,
\\
I_\alpha^{\flat}(A:B|C)_\rho
&\leq H_\gamma^{\flat,\downarrow}(B|C)_\rho
      -H_\beta^{\flat,\downarrow}(B|AC)_\rho
&\qquad&\text{when}\quad
\frac{\beta}{\beta-1}
=\frac{\alpha}{\alpha-1}+\frac{1}{\gamma-1}.
\end{alignat}
If instead \(\gamma>1\), then the inequalities are reversed.
\end{lemma}

\begin{proof}
We only prove the case where $\alpha,\beta,\gamma<1$, and note that other cases follow a similar proof strategy. The optimizers of the variational expression are always achieved (see Remark~\ref{rem: full-support CE and MI}).
    Let $\tau_{ABC}^\star$ be such that
    \begin{align}
        H_\beta^{\flat,\uparrow}(B|AC)_\rho=\frac{\beta}{\beta-1}D(\tau^\star_{ABC}\|\rho_{ABC})+H(B|AC)_{\tau^\star}.
    \end{align}
    Then, we have
    \begin{align}\label{eq:QCMI_cond_unopt_p1}
        I_\alpha^{\flat,*}(A:B|C)_\rho&\leq\frac{\alpha}{1-\alpha}D(\tau^\star_{ABC}\| \rho_{ABC})+I(A:B|C)_{\tau^\star} \\
        &=\left(\frac{\gamma}{\gamma-1}-\frac{\beta}{\beta-1}\right)D(\tau^\star_{ABC}\| \rho_{ABC})+H(B|C)_{\tau^\star}-H(B|AC)_{\tau^\star} \\
        &=\left(\frac{\gamma}{\gamma-1}D(\tau^\star_{ABC}\| \rho_{ABC})+H(B|C)_{\tau^\star}\right)-\left(\frac{\beta}{\beta-1}D(\tau^\star_{ABC}\| \rho_{ABC})+H(B|AC)_{\tau^\star}\right) \\
    &\leq\left(\frac{\gamma}{\gamma-1}D(\tau^\star_{BC}\| \rho_{BC})+H(B|C)_{\tau^\star}\right)-H_\beta^{\flat,\uparrow}(B|AC)_\rho \\
        &\leq\sup_{\tau_{BC}}\left\{ \frac{\gamma}{\gamma-1}D(\tau_{BC}\| \rho_{BC})+H(B|C)_{\tau}\right\}-H_\beta^{\flat,\uparrow}(B|AC)_\rho \\
        &=H_\gamma^{\flat,\uparrow}(B|C)_\rho -H_\beta^{\flat,\uparrow}(B|AC)_\rho.
    \end{align}
    The first line follows from the variational representation in
Corollary~\ref{cor: variational form triply optimized}. The second line follows
from the chain rule for the QCMI,
\(I(A:B|C)_\rho = H(B|C)_\rho - H(B|AC)_\rho\), together with the prescribed
relation between the parameters. The fourth line follows from
Eq.~\eqref{eq:LE_cond} and the data-processing inequality applied to
\(D(\tau^\star_{ABC}\| \rho_{ABC})\).

    The proof of the second chain rule is similar. Consider $\bar{\tau}_{ABC}$ to be such that
    \begin{align}
        H_\beta^{\flat,\downarrow}(B|AC)_\rho=\frac{\beta}{\beta-1}D(\bar{\tau}_{ABC}\|\rho_{ABC})-D(\bar{\tau}_{AC}\|\rho_{AC})+H(B|AC)_{\bar{\tau}}.
    \end{align}
    Then, 
    \begin{align}
        I_\alpha^{\flat}(A:B|C)_\rho&\leq\frac{\alpha}{1-\alpha}D(\bar{\tau}_{ABC}\|\rho_{ABC})+D(\bar{\tau}_{AC}\|\rho_{AC})+D(\bar{\tau}_{BC}\|\rho_{BC})-D(\bar{\tau}_C\|\rho_C)+I(A:B|C)_{\bar{\tau}} \\
        &=\left(\frac{1}{\gamma-1}D(\bar{\tau}_{ABC}\|\rho_{ABC})+D(\bar{\tau}_{BC}\|\rho_{BC})-D(\bar{\tau}_C\|\rho_C)+H(B|C)_{\bar{\tau}}\right) \notag \\
        &\qquad\qquad-\left(\frac{\beta}{\beta-1}D(\bar{\tau}_{ABC}\|\rho_{ABC})-D(\bar{\tau}_{AC}\|\rho_{AC})+H(B|AC)_{\bar{\tau}}\right) \\
        &\leq\left(\frac{\gamma}{\gamma-1}D(\bar{\tau}_{BC}\|\rho_{BC})-D(\bar{\tau}_C\|\rho_C)+H(B|C)_{\bar{\tau}}\right)-H_\beta^{\flat,\downarrow}(B|AC)_\rho \\
        &\leq\sup_{\tau_{BC}}\left\{\frac{\gamma}{\gamma-1}D(\tau_{BC}\|\rho_{BC})-D(\tau_C\|\rho_C)+H(B|C)_{\tau}\right\}-H_\beta^{\flat,\downarrow}(B|AC)_\rho \\
        &=H_\gamma^{\flat,\downarrow}(B|C)_\rho-H_\beta^{\flat,\downarrow}(B|AC)_\rho\,.
    \end{align}
   The second line is obtained by substituting the relation among the parameters
and using the chain rule for the QCMI. The second inequality then follows from
an application of data processing, together with the variational
representation of the conditional entropy.
\end{proof}

\subsection{Chain rules involving mutual informations}

\begin{lemma}
Let \(\rho_{ABC}\) be a tripartite quantum state, and let
\(\alpha,\beta,\gamma \in (0,1)\cup (1,\infty)\). Suppose that
\((\alpha-1)(\beta-1)>0\).
If \(\gamma\in (0,1)\), it holds that
\begin{alignat}{2}
I_\alpha^{\flat,*}(A:B|C)_\rho
&\leq I_\beta^{\flat,\downarrow\downarrow}(AC:B)_\rho
      -I_\gamma^{\flat,\downarrow\downarrow}(C:B)_\rho
&\qquad&\text{when}\quad
\frac{\beta}{\beta-1}
=\frac{\alpha}{\alpha-1}+\frac{\gamma}{\gamma-1},
\;\;\text{and}\;\; \alpha\leq2,
\\
I_\alpha^{\flat}(A:B|C)_\rho
&\leq I_\beta^{\flat,\uparrow\uparrow}(AC:B)_\rho
      -I_\gamma^{\flat,\uparrow\uparrow}(C:B)_\rho
&\qquad&\text{when}\quad
\frac{\beta}{\beta-1}
=\frac{\alpha}{\alpha-1}+\frac{1}{\gamma-1}.
\end{alignat}
If instead \(\gamma>1\), then the inequalities are reversed.
\end{lemma}

\begin{proof}
We only prove the statement for the case where $\alpha,\beta,\gamma<1$ and note that other cases follow a similar strategy.
    Let $\tau^\star_{ABC}$ be such that 
\begin{align}
        I_\beta^{\flat,\downarrow\downarrow}(AC:B)_\rho =\frac{\beta}{1-\beta}D(\tau^\star_{ABC}\|\rho_{ABC})+I(AC:B)_{\tau^\star}.
    \end{align}
    Then, we have
\begin{align}\label{eq:QCMI_cond_unopt_2 p}
        I_\alpha^{\flat,*}(A:B|C)_\rho&\leq\frac{\alpha}{1-\alpha}D(\tau^\star_{ABC}\| \rho_{ABC})+I(A:B|C)_{\tau^\star} \\
        &=\left(\frac{\beta}{1-\beta}-\frac{\gamma}{1-\gamma}\right)D(\tau^\star_{ABC}\| \rho_{ABC})+I(AC:B)_{\tau^\star}-I(C:B)_{\tau^\star} \\
        &=\left(\frac{\beta}{1-\beta}D(\tau^\star_{ABC}\| \rho_{ABC})+I(AC:B)_{\tau^\star}\right)-\left(\frac{\gamma}{1-\gamma}D(\tau^\star_{ABC}\| \rho_{ABC})+I(C:B)_{\tau^\star}\right) \\
        &\leq I_\beta^{\flat,\downarrow\downarrow}(AC:B)_{\rho}-\left(\frac{\gamma}{1-\gamma}D(\tau^\star_{BC}\| \rho_{BC})+I(C:B)_{\tau^\star}\right) \\
        &\leq I_\beta^{\flat,\downarrow\downarrow}(AC:B)_{\rho}-\inf_{\tau_{BC}}\left\{\frac{\gamma}{1-\gamma}D(\tau_{BC}\| \rho_{BC})+I(C:B)_{\tau}\right\} \\
        &= I_\beta^{\flat,\downarrow\downarrow}(AC:B)_{\rho}  - I_\gamma^{\flat,\downarrow\downarrow}(C:B)_{\rho},
    \end{align}
   The first line is an application of the variational representation. The second
line follows by inserting the prescribed relation among the parameters and
using the chain rule \(I(A:B|C)_\rho = I(AC:B)_\rho - I(C:B)_\rho\). The
fourth line then follows from the DPI of relative entropy under the
partial trace, applied to \(D(\tau^\star_{ABC}\| \rho_{ABC})\).

    The proof of the second chain rule follows from similar steps as above. Consider $\bar{\tau}_{ABC}$ to be such that
    \begin{align}
        I_\beta^{\flat,\uparrow\uparrow}(AC:B)_\rho =\frac{\beta}{1-\beta}D(\bar{\tau}_{ABC}\|\rho_{ABC})+D(\bar{\tau}_{AC}\|\rho_{AC})+D(\bar{\tau}_{B}\|\rho_{B})+I(AC:B)_{\bar{\tau}}.
    \end{align}
    Then, 
    \begin{align}
I_\alpha^\flat(A:B|C)_\rho&\leq\frac{\alpha}{1-\alpha}D(\bar{\tau}_{ABC}\|\rho_{ABC})+D(\bar{\tau}_{AC}\|\rho_{AC})+D(\bar{\tau}_{BC}\|\rho_{BC})-D(\bar{\tau}_C\|\rho_C)+I(A:B|C)_{\bar{\tau}} \\
        &=\left(\frac{\beta}{1-\beta}D(\bar{\tau}_{ABC}\|\rho_{ABC})+D(\bar{\tau}_{AC}\|\rho_{AC})+D(\bar{\tau}_B\|\rho_B)+I(AC:B)_{\bar{\tau}}\right) \notag \\
        &\quad-\left(\frac{1}{1-\gamma}D(\bar{\tau}_{ABC}\|\rho_{ABC})-D(\bar{\tau}_{BC}\|\rho_{BC})+D(\bar{\tau}_{C}\|\rho_{C})+D(\bar{\tau}_{B}\|\rho_{B})+I(C:B)_{\bar{\tau}}\right) \\
        &\leq I_\beta^{\flat,\uparrow\uparrow}(AC:B)_\rho-\left(\frac{\gamma}{1-\gamma}D(\bar{\tau}_{BC}\|\rho_{BC})+D(\bar{\tau}_C\|\rho_C)+D(\bar{\tau}_B\|\rho_B)+I(C:B)_{\bar{\tau}}\right) \\
        &\leq I_\beta^{\flat,\uparrow\uparrow}(AC:B)_\rho-\inf_{\tau_{BC}}\left\{\frac{\gamma}{1-\gamma}D(\tau_{BC}\|\rho_{BC})+D(\tau_C\|\rho_C)+D(\tau_B\|\rho_B)+I(C:B)_{\tau}\right\} \\
        &=I_\beta^{\flat,\uparrow\uparrow}(AC:B)_\rho-I_\gamma^{\flat,\uparrow\uparrow}(C:B)_\rho.
    \end{align}
    The first line follows from the variational expression. The second line is
obtained by using the relation among the parameters together with the chain
rule for the QCMI in terms of mutual informations. The second inequality then
follows from an application of data processing.
\end{proof}

\Needspace{12\baselineskip}
\subsection{Chain rules involving conditional mutual informations}

\begin{lemma}
\label{lem: chain rules Renyi QCMI}
Let \(\rho_{ABCD}\) be a quantum state, and let
\(\alpha,\beta,\gamma \in (0,1)\cup (1,\infty)\). Suppose that
\((\alpha-1)(\beta-1)>0\).
If \(\gamma\in (0,1)\), it holds that
\begin{alignat}{2}
	I_\alpha^{\flat,*}(A:B|CD)_\rho
	&\leq I_\beta^{\flat,*}(A:BC|D)_\rho
	-I_\gamma^{\flat,*}(A:C|D)_\rho
	&\;\;&\text{when}\quad
	\frac{\beta}{\beta-1}
	=\frac{\alpha}{\alpha-1}+\frac{\gamma}{\gamma-1},
	\;\;\; \text{and}\;\; \alpha,\beta,\gamma\leq2,
	\\
	I_\alpha^{\flat}(A:B|CD)_\rho
	&\leq I_\beta^{\flat}(A:BC|D)_\rho
	-I_\gamma^{\flat}(A:C|D)_\rho
	&\;\;\;&\text{when}\quad
	\frac{\beta}{\beta-1}
	=\frac{\alpha}{\alpha-1}+\frac{1}{\gamma-1}.
\end{alignat}
If instead \(\gamma>1\), then the inequalities are reversed.
\end{lemma}
\begin{proof}
We only prove the statement for the case where $\alpha,\beta,\gamma<1$ and note that other cases follow a similar strategy.
    Let $\tau_{ABCD}^\star$ be such that 
\begin{align}\label{eq:tau_opt}
    I_\beta^{\flat,*}(A:BC|D)_\rho=\frac{\beta}{1-\beta}D(\tau^\star_{ABCD}\|\rho_{ABCD})+I(A:BC| D)_{\tau^\star}.
    \end{align}
    Then, with similar reasoning as the previous lemma, we have
\begin{align}\label{eq:QCMI_cond_unopt_full p1}
        I_\alpha^{\flat,*}(A:B|CD)_\rho&\leq\frac{\alpha}{1-\alpha}D(\tau^\star_{ABCD}\|\rho_{ABCD})+I(A:B|CD)_{\tau^\star} \\
        &=\left(\frac{\beta}{1-\beta}-\frac{\gamma}{1-\gamma}\right)D(\tau^\star_{ABCD}\| \rho_{ABCD})+I(A:BC|D)_{\tau^\star}-I(A:C|D)_{\tau^\star} \\
        &=\left(\frac{\beta}{1-\beta}D(\tau^\star_{ABCD}\| \rho_{ABCD})+I(A:BC|D)_{\tau^\star}\right)\notag \\
        & \qquad \qquad-\left(\frac{\gamma}{1-\gamma}D(\tau^\star_{ABCD}\| \rho_{ABCD})+I(A:C|D)_{\tau^\star} \right) \\
        &\leq  I_\beta^{\flat,*}(A:BC|D)_\rho-\left(\frac{\gamma}{1-\gamma}D(\tau^\star_{ACD}\| \rho_{ACD})+I(A:C|D)_{\tau^\star} \right) \\
        &\leq I_\beta^{\flat,*}(A:BC|D)_\rho-\inf_{\tau_{ACD}}\left\{\frac{\gamma}{1-\gamma}D(\tau_{ACD}\| \rho_{ACD})+I(A:C|D)_{\tau} \right\} \\
        &=I_\beta^{\flat,*}(A:BC|D)_\rho-I_\gamma^{\flat,*}(A:C|D)_\rho.
    \end{align}
        The proof of the second relation follows from similar steps as above. Consider $\bar{\tau}_{ABCD}$ to be such that
    \begin{multline}
        I_\beta^\flat(A:BC|D)_\rho=\\
        \frac{\beta}{1-\beta}D(\bar{\tau}_{ABCD}\|\rho_{ABCD})+D(\bar{\tau}_{AD}\|\rho_{AD})+D(\bar{\tau}_{BCD}\|\rho_{BCD})-D(\bar{\tau}_{D}\|\rho_{D})+I(A:BC|D)_{\bar{\tau}}.
    \end{multline}
    Then, with a similar reasoning as the previous lemma, we have
    \begin{align}
        &I_\alpha^\flat(A:B|CD)_\rho \notag \\
        &\leq\frac{\alpha}{1-\alpha}D(\bar{\tau}_{ABCD}\|\rho_{ABCD})+D(\bar{\tau}_{ACD}\|\rho_{ACD})+D(\bar{\tau}_{BCD}\|\rho_{BCD})-D(\bar{\tau}_{CD}\|\rho_{CD})+I(A:B|CD)_{\bar{\tau}} \\
        &=\left(\frac{\beta}{1-\beta}D(\bar{\tau}_{ABCD}\|\rho_{ABCD})+D(\bar{\tau}_{AD}\|\rho_{AD})+D(\bar{\tau}_{BCD}\|\rho_{BCD})-D(\bar{\tau}_D\|\rho_D)+I(A:BC|D)_{\bar{\tau}}\right)\notag  \\
        &\quad-\bigg(\frac{1}{1-\gamma}D(\bar{\tau}_{ABCD}\|\rho_{ABCD})-D(\bar{\tau}_{ACD}\|\rho_{ACD})+D(\bar{\tau}_{AD}\|\rho_{AD})\notag\\
        & \qquad \qquad \qquad \qquad \qquad\qquad \qquad +D(\bar{\tau}_{CD}\|\rho_{CD})-D(\bar{\tau}_{D}\|\rho_{D})+I(A:C|D)_{\bar{\tau}}\bigg) \\
        &\leq  I_\beta^\flat(A:BC|D)_\rho-
        \Big(\frac{\gamma}{1-\gamma}D(\bar{\tau}_{ACD}\|\rho_{ACD})+D(\bar{\tau}_{AD}\|\rho_{AD})\notag \\
        & \qquad\qquad\qquad \qquad\qquad+D(\bar{\tau}_{CD}\|\rho_{CD})-D(\bar{\tau}_{D}\|\rho_{D})+I(A:C|D)_{\bar{\tau}}\Big)\\
        &\leq I_\beta^\flat(A:BC|D)_\rho\notag \\
        &\quad -\inf_{\tau_{ACD}}\left\{\frac{\gamma}{1-\gamma}D(\tau_{ACD}\|\rho_{ACD})+D(\tau_{AD}\|\rho_{AD})+D(\tau_{CD}\|\rho_{CD})-D(\tau_{D}\|\rho_{D})+I(A:C|D)_{\tau}\right\} \\
        &=I_\beta^\flat(A:BC|D)_\rho-I_\gamma^\flat(A:C|D)_\rho.
    \end{align}
\end{proof}

\begin{remark}
   The chain rules established in this section are tight in the limit
\(\alpha \to 1\). Indeed, in this limit they reduce to the usual chain rules
for the QCMI. Let us illustrate this for the first chain rule in
Lemma~\ref{lem: chain rules Renyi QCMI}.

Fix \(0<\varepsilon<1\), and choose
\(\alpha=1-\varepsilon\), \(\gamma=1-\varepsilon\), and
\(\beta=1-\varepsilon/(2-\varepsilon)\). These parameters satisfy the required
relation among the parameters. Taking the limit
\(\varepsilon \to 0\), the corresponding inequality gives
\begin{align}
\label{eq:vN_full_upper}
    I(A:B|CD)_\rho
    \leq
    I(A:BC|D)_\rho
    -
    I(A:C|D)_\rho .
\end{align}
To obtain the reverse inequality, choose instead
\(\alpha=1+\varepsilon\), \(\gamma=1+\varepsilon\), and
\(\beta=1+\varepsilon/(2+\varepsilon)\). Passing again to the limit
\(\varepsilon \to 0\), we obtain
\begin{align}
\label{eq:vN_full_lower}
    I(A:B|CD)_\rho
    \geq
    I(A:BC|D)_\rho
    -
    I(A:C|D)_\rho .
\end{align}
Combining~\eqref{eq:vN_full_upper} and~\eqref{eq:vN_full_lower} yields the
standard QCMI chain rule
\begin{align}
    I(A:B|CD)_\rho
    =
    I(A:BC|D)_\rho
    -
    I(A:C|D)_\rho .
\end{align}
\end{remark}

\begin{remark}
We note that the parameter relations can be written in the unified form
\begin{align}
\frac{1}{1-\alpha}
=
\frac{1}{1-\beta}
+\mu
-
\frac{1}{1-\gamma},
\end{align}
where \(\mu=1\) for the optimized quantities and \(\mu=0\) for the non-optimized quantities. An analogous relation for classical conditional entropies appears in~\cite[Eq.~(313)]{rubboli2024quantum}.
\end{remark}

\section{Duality relations}
\label{sec: Duality}
For a pure state \(\rho_{ABCD}\), the quantum conditional mutual information satisfies the duality relation
\begin{align}
I(A:B | C)_\rho = I(A:B | D)_\rho .
\end{align}
Thus, the QCMI is self-dual: under the duality transformation, it is mapped back to a quantity of the same form. In particular, the dual of the QCMI is again a conditional mutual information. This duality has an operational interpretation in terms of time reversing the state redistribution protocol \cite{DevetakYard2008ExactCost}, and it has been explored in other operational contexts related to information measures inspired by conditional mutual information~\cite{PhysRevA.94.022310}.

A prominent class of quantities for which duality has been widely studied is quantum conditional entropies. When defined by axioms such as monotonicity under an appropriate class of quantum channels and additivity under tensor products, the duals of several known quantum conditional entropies are themselves quantum conditional entropies; see, for instance, \cite{Gour2025QuantumResourceTheories,JiGourWilde2025WorkCosts,rubboli2024quantum}. It is therefore natural to ask whether the quantity obtained by applying this duality transformation to the LE-CMI can itself be interpreted as a conditional mutual information. A basic necessary requirement for such an interpretation is data processing under local channels. In the following, we first derive the dual expression associated with the LE-CMI. We then show, by means of an explicit counterexample, that this dual quantity fails to satisfy the data-processing inequality under local channels. This raises the question of whether closure under duality, commonly required for conditional entropies, is equally natural for conditional mutual informations.

We begin with an auxiliary lemma.
\begin{lemma}
\label{lem:schmidt-compression-transpose}
Let $|\rho\rangle_{RD}$ be a pure state with Schmidt decomposition
\begin{align}
    |\rho\rangle_{RD}
    =
    \sum_{i=1}^{r}
    \sqrt{\lambda_i}\,
    |i\rangle_R |i\rangle_D,
    \qquad
    \lambda_i>0.
    \label{eq:schmidt-RD}
\end{align}
Let $V_{D\to R}$ be the Schmidt isometry from $\supp(\rho_D)$ onto $\supp(\rho_R)$ defined by
$V_{D\to R}|i\rangle_D=|i\rangle_R$ for all $i=1,\dots,r$. Then, for every
Hermitian operator $X_R$,
\begin{align}
    \bigl(V_{D\to R}^{\dagger}X_RV_{D\to R}\bigr)^T
    =
    \rho_D^{-1/2}
    \operatorname{Tr}_R\!\left[
        (X_R\otimes I_D)\rho_{RD}
    \right]
    \rho_D^{-1/2},
    \label{eq:compression-transpose}
\end{align}
where the transpose is taken in the Schmidt basis
$\{|i\rangle_D\}_{i=1}^{r}$.
\end{lemma}

\begin{proof}
The reduced states are $\rho_R
    =
    \sum_i\lambda_i |i\rangle\!\langle i|_R$ and $\rho_D
    =
    \sum_i\lambda_i |i\rangle\!\langle i|_D$.
Moreover,
\begin{align}
    \rho_{RD}
    =
    \sum_{i,j}
    \sqrt{\lambda_i\lambda_j}\,
    |i\rangle\!\langle j|_R
    \otimes
    |i\rangle\!\langle j|_D .
\end{align}
Hence
\begin{align}
    \operatorname{Tr}_R\!\left[
        (X_R\otimes I_D)\rho_{RD}
    \right]
    &=
    \sum_{i,j}
    \sqrt{\lambda_i\lambda_j}\,
    \langle j|X_R|i\rangle_R\,
    |i\rangle\!\langle j|_D
     =
    \rho_D^{1/2}
    \bigl(
        V_{D\to R}^{\dagger}X_RV_{D\to R}
    \bigr)^T
    \rho_D^{1/2}.
    \label{eq:partial-trace-index}
\end{align}
Multiplying on the left and right by $\rho_D^{-1/2}$ proves the claim.
\end{proof}

\subsection{Duality relation for the non-optimized log-Euclidean conditional mutual information}

Next, we prove the duality relation.
We define 
\begin{align}
    J_\alpha(A:B| C)_\rho
    \coloneqq
    \frac{1}{\alpha-1}
    \log
    \operatorname{Tr}\exp\!\left[
        \alpha\log\rho_C
        +
        (1-\alpha)G_C(\rho_{ABC})
    \right],
\end{align}
where
\begin{align}
    G_C(\rho_{ABC})
    \coloneqq
    \rho_C^{-1/2}
    \Bigl(
        \operatorname{Tr}_B(\rho_{BC}\log\rho_{BC})
        +
        \operatorname{Tr}_A(\rho_{AC}\log\rho_{AC})
        -
        \operatorname{Tr}_{AB}(\rho_{ABC}\log\rho_{ABC})
    \Bigr)
    \rho_C^{-1/2}.
\end{align}
Logarithms and inverse powers act on their supports; the trace is
taken on $\supp(\rho_C)$.

\begin{theorem}[Duality relation]
\label{thm:duality}
Let $\rho_{ABCD}$ be a pure state and let $\alpha\in (0,1)\cup(1,\infty)$. Then,
\begin{align}
    I^\flat_\alpha(A:B| C)_\rho
    =
    J_\alpha(A:B| D)_\rho .
\end{align}
\end{theorem}

\begin{proof}
Choose a Schmidt decomposition of $|\rho\rangle_{ABCD}$ across the cut
$ABC:D$:
\begin{align}
    |\rho\rangle_{ABCD}
    =
    \sum_{i=1}^{r}
    \sqrt{\lambda_i}\,
    |i\rangle_{ABC} |i\rangle_D .
\end{align}
Let $V_{D\to ABC}$ be the corresponding Schmidt isometry. The limiting
definitions give the first trace on the Schmidt support, with the
exponent compressed there. By
unitary invariance of this trace,
\begin{align}
    &\operatorname{Tr}\exp\!\left[
        \alpha\log\rho_{ABC}
        +
        (1-\alpha)
        \bigl(
            \log\rho_{AC}\otimes I_B
            +
            \log\rho_{BC}\otimes I_A
            -
            \log\rho_C\otimes I_{AB}
        \bigr)
    \right]
    \notag \\
    &\qquad =
    \operatorname{Tr}\exp\!\left[
        \alpha\log\rho_D
        +
        (1-\alpha)G_D(\rho_{ABD})
    \right],
    \label{eq:duality-proof}
\end{align}
where
\begin{align}
    G_D(\rho_{ABD})
    =
    \left(
        V_{D\to ABC}^{\dagger}
        \bigl(
            \log\rho_{AC}\otimes I_B
            +
            \log\rho_{BC}\otimes I_A
            -
            \log\rho_C\otimes I_{AB}
        \bigr)
        V_{D\to ABC}
    \right)^T .
\end{align}
 Here, the transpose is taken in the Schmidt basis of $D$, in which
$\log\rho_D$ is diagonal, and we use
$\operatorname{Tr}e^A=\operatorname{Tr}e^{A^T}$.

It remains to identify $G_D(\rho_{ABD})$. Applying
Lemma~\ref{lem:schmidt-compression-transpose} term by term, and then
using Lemma~\ref{lem:complement-identities}, gives
\begin{align}
    G_D(\rho_{ABD})
    =
    \rho_D^{-1/2}
    \Bigl(
        \operatorname{Tr}_B(\rho_{BD}\log\rho_{BD})
        +
        \operatorname{Tr}_A(\rho_{AD}\log\rho_{AD})
        -
        \operatorname{Tr}_{AB}(\rho_{ABD}\log\rho_{ABD})
    \Bigr)
    \rho_D^{-1/2}.
\end{align}
Substituting this expression into \eqref{eq:duality-proof} gives exactly
the exponential trace appearing in
$J_\alpha(A:B| D)_\rho$. Taking logarithms and dividing by $\alpha-1$ proves the claim.
\end{proof}

\section{Conclusion and outlook}
In this work, we introduced log-Euclidean R\'enyi conditional mutual informations and established data processing under local quantum channels for $0<\alpha\leq1$, monotonicity in the R\'enyi parameter, and convergence to the standard QCMI. The non-optimized quantity is additive for all $\alpha>0$, while the optimized quantity is additive for $\alpha\in[1/2,2]$. The optimized additivity proof uses joint convexity and fixed-point conditions for the optimizers.

We established bounds in terms of the distances to quantum Markov states, separable states, and states prepared by rotated Petz recovery maps. In particular, the lower recovery bound involving Keyl's rate function strengthens the corresponding log-fidelity bounds and can strictly improve on the measured-relative-entropy bound, although these last two bounds are not universally ordered. We also obtained a complementary upper bound involving $D^{\text{\normalfont\leftmoon}}$ for full-rank states. 

Finally, we proved chain rules with explicit parameter conditions, and a duality relation and showed that the corresponding dual quantity is not itself
a conditional mutual information, since it fails to satisfy data processing under local quantum channels. The possible use of these quantities in strong-converse problems and in R\'enyi analogues of squashed entanglement remains a subject for future work.

As in the case of relative entropies, where several quantum generalizations of the classical Rényi divergence are possible, most notably the Petz and sandwiched Rényi divergences, it is natural to ask whether alternative quantum generalizations of the classical Rényi conditional mutual information can be constructed beyond the one developed here. 
A possible route is suggested by the Golden--Thompson inequality, which relates the Petz and log-Euclidean Rényi relative entropies. It is therefore natural to ask whether a Petz-type Rényi conditional mutual information can similarly be derived from the log-Euclidean quantity introduced in this work by means of a multivariate Golden--Thompson inequality. However, in the multivariate case, as we argue below, some difficulty arises.  For two Hermitian matrices \(H_1\) and \(H_2\), one has
\begin{align}
\Tr \exp(H_1+H_2)
\leq
\Tr \exp(H_1)\exp(H_2).
\end{align}
Both sides are invariant under exchanging \(H_1\) and \(H_2\); hence, no ordering ambiguity arises. In contrast, the multivariate extension~\cite{sutter2017multivariate}
\begin{align}
\log \left\|
\exp\left(\sum_{k=1}^{n} H_k\right)
\right\|_{p}
\leq
\int_{-\infty}^{\infty} \d t\,
\beta_0(t) 
\log \left\|
\prod_{k=1}^{n}
\exp\left((1+it)H_k\right)
\right\|_{p} , 
\end{align}
holding for $p\geq 1$ (and extended to every $p>0$ in~\cite{hiai2017logmajorization}),
depends on the order of the matrices on the right-hand side, while the left-hand side is permutation invariant. Here, \(\lVert\cdot\rVert_p\) denotes the Schatten \(p\)-norm, and \(\beta_0\) is a fixed probability density.
Thus, applying a multivariate Golden--Thompson inequality to the LE-CMI, which satisfies a permutation symmetry, naturally leads to Petz-type expressions that are not permutation invariant at the operator level. This ordering dependence reflects a more fundamental difficulty: different placements of noncommuting operators inside the trace may give inequivalent quantities, and such asymmetry can obstruct DPI under local operations on both systems. These issues have also appeared in the study of Rényi conditional mutual informations~\cite{berta2015renyi} and multivariate fidelities~\cite{nuradha2025multivariate}. Although one could attempt to restore symmetry by minimizing the right hand side over all orderings, DPI arguments typically require symmetries already present at the operator level.

The recovery bounds established here provide concrete expressions arising from these trace inequalities, but do not establish data processing for a corresponding Petz or sandwiched conditional mutual information. Determining whether such constructions can satisfy the same structural properties remains a natural direction for future work.

\section{Acknowledgments}
RR thanks Daniel Stilck França for discussions on conditional independence and Thomas C. Fraser for discussions on Keyl’s rate function. RR acknowledges financial support from the ERC grant GIFNEQ 101163938. AA conducted research at the Institute for Quantum
Computing, at the University of Waterloo, which is supported by Innovation, Science, and Economic Development
Canada. Support was also provided by NSERC under the
Discovery Grants Program, Grant No. 341495. MMW
acknowledges support from the National Science
Foundation under grant nos.~2329662 and 2611810.

\section{Statement on the use of artificial intelligence}

The main results in Sections~\ref{sec: intro}-\ref{sec: Duality}, along with Appendices~\ref{app: lemmas}-\ref{app: Taylor expansion}, were derived by the authors and incorporated into the manuscript by early September 2026.

After this point, ChatGPT was used for proofreading, mathematical and bibliographic checks, and assistance with some derivations in Sections~\ref{sec: recovery bounds},~\ref{sec: Duality}, and Appendices~\ref{app:numerical-recovery}-\ref{app: counter example DPI dual}.

\bibliography{library}

\appendix

\clearpage

\section{Useful lemmas}\label{app: lemmas}

We begin with the Gibbs variational principle. See~\cite{petz1988variational} for the usual form, and see also~\cite[Lemma~1.2]{hiai1993golden} for a finite-dimensional proof.

\begin{lemma}[Gibbs variational principle]
\label{lem: Gibbs}
Let $H$ be a self-adjoint operator. Then,
\begin{align}
    -\log\Tr[\exp(-H)]
    =\min_{\tau}\{\Tr[\tau\log\tau]+\Tr[H\tau]\},
\end{align}
where the minimum is over states $\tau$. 
\end{lemma}

In the following, logarithms and inverse powers of reduced states are taken on
the corresponding supports.

\begin{lemma}
\label{lem:log-transport}
Let $|\rho\rangle_{AB}$ be pure. Then
\begin{align}
    \bigl(\log\rho_A\otimes I_B\bigr)|\rho\rangle_{AB}
    =
    \bigl(I_A\otimes\log\rho_B\bigr)|\rho\rangle_{AB} .
    \label{eq:log-transport}
\end{align}
\end{lemma}

\begin{proof}
Choose a Schmidt decomposition
\begin{align}
    |\rho\rangle_{AB}
    =
    \sum_i
    \sqrt{\lambda_i}\,
    |i\rangle_A|i\rangle_B,
    \qquad
    \lambda_i>0.
\end{align}
Both sides of \eqref{eq:log-transport} are equal to
\begin{align}
    \sum_i
    \sqrt{\lambda_i}\log(\lambda_i)\,
    |i\rangle_A |i\rangle_B ,
\end{align}
thus concluding the proof.
\end{proof}

The next identities are the concrete forms used in the main proof.

\begin{lemma}
\label{lem:complement-identities}
Let $\rho_{ABCD}$ be pure. Then
\begin{align}
    \operatorname{Tr}_{ABC}\!\left[
        \bigl(\log\rho_{AC}\otimes I_{BD}\bigr)
        \rho_{ABCD}
    \right]
    &=
    \operatorname{Tr}_B\!\left(
        \rho_{BD}\log\rho_{BD}
    \right),
    \label{eq:AC-to-BD}
    \\
    \operatorname{Tr}_{ABC}\!\left[
        \bigl(\log\rho_{BC}\otimes I_{AD}\bigr)
        \rho_{ABCD}
    \right]
    &=
    \operatorname{Tr}_A\!\left(
        \rho_{AD}\log\rho_{AD}
    \right),
    \label{eq:BC-to-AD}
    \\
    \operatorname{Tr}_{ABC}\!\left[
        \bigl(\log\rho_C\otimes I_{ABD}\bigr)
        \rho_{ABCD}
    \right]
    &=
    \operatorname{Tr}_{AB}\!\left(
        \rho_{ABD}\log\rho_{ABD}
    \right).
    \label{eq:C-to-ABD}
\end{align}
\end{lemma}

\begin{proof}
We prove \eqref{eq:AC-to-BD}; the other two identities follow by the
same argument with the obvious relabeling. Applying
Lemma~\ref{lem:log-transport} to the bipartition $AC:BD$ gives
\begin{align}
    \bigl(\log\rho_{AC}\otimes I_{BD}\bigr)|\rho\rangle
    =
    \bigl(I_{AC}\otimes\log\rho_{BD}\bigr)|\rho\rangle .
\end{align}
Therefore
\begin{align}
    \operatorname{Tr}_{ABC}\!\left[
        \bigl(\log\rho_{AC}\otimes I_{BD}\bigr)
        \rho_{ABCD}
    \right]
    & =
    \operatorname{Tr}_{ABC}\!\left[
        \bigl(I_{AC}\otimes\log\rho_{BD}\bigr)
        \rho_{ABCD}
    \right]
     \\
    & =
    \operatorname{Tr}_B\!\left[
        \log\rho_{BD}\,
        \operatorname{Tr}_{AC}\rho_{ABCD}
    \right]
     \\
    & =
    \operatorname{Tr}_B\!\left(
        \log\rho_{BD}\,\rho_{BD}
    \right)
     \\
    & =
    \operatorname{Tr}_B\!\left(
        \rho_{BD}\log\rho_{BD}
    \right),
\end{align}
where the last equality uses that $\rho_{BD}$ commutes with
$\log\rho_{BD}$.
\end{proof}

\section{Expansion of \texorpdfstring{$I_\alpha^\flat$}{the LE-CMI} around \texorpdfstring{$\alpha=1$}{alpha = 1}}

\label{app: Taylor expansion}

In this appendix, we derive the expansion of the LE-CMI $ I_\alpha^\flat$ around $\alpha=1$.

We first introduce some notation.
For a quantum state \(\rho\) and operators \(X,Y\), define
\begin{align}
\langle X,Y\rangle_{\rho}^{\KM}
\coloneqq 
\int_0^1 ds\,
\Tr\!\big[
\rho^s X^\dagger \rho^{1-s}Y
\big].
\end{align}
When \(\rho\) is full rank, this defines an inner product, commonly known as the Kubo--Mori, or Bogoliubov--Kubo--Mori, inner product. The associated variance is defined by
\begin{align}
\mathrm{V}_{\rho}^{\KM}(X)
\coloneqq 
\left\langle
X-\langle X\rangle_{\rho},
X-\langle X\rangle_{\rho}
\right\rangle_{\rho}^{\KM}\,,
\end{align}
where \(\langle X\rangle_{\rho}\coloneqq \Tr[\rho X]\) denotes the expectation value of \(X\) in the state \(\rho\). 
For a tripartite state \(\rho_{ABC}\), we define the conditional mutual information operator by
\begin{align}
K_{ABC}
\coloneqq 
\log\rho_{ABC}
+\log\rho_C
-\log\rho_{AC}
-\log\rho_{BC}\,.
\end{align}
Here, $K_{ABC}$ is compressed to $\supp(\rho_{ABC})$, and marginal
logarithms act on their supports. 
We then define the conditional mutual information variance as
\begin{align}
\mathrm{V}_{\KM}(A:B|C)_\rho
\coloneqq 
\mathrm{V}_{\rho_{ABC}}^{\KM}(K_{ABC}).
\end{align}

\begin{proposition}
Let \(\rho_{ABC}\) be a quantum state. Then the following expansion holds around \(\alpha=1\):
\begin{align}
I_{\alpha}^{\flat}(A:B|C)_\rho
=
I(A:B|C)_\rho
+
\frac{\alpha-1}{2}
\mathrm{V}_{\KM}(A:B|C)_\rho
+
O\!\left((\alpha-1)^2\right).
\end{align}
\end{proposition}

\begin{proof}
We give the proof for full-rank states. For states that are not full
rank, the limiting definitions give the same calculation on their supports.

Define $\beta\coloneqq\alpha-1$. Then, $I_{\alpha}^{\flat}(A:B|C)_{\rho}=\frac{1}{\beta}\log Z(\beta)$, 
where $Z(\beta)\coloneqq\Tr\!\left[\exp\!\left(\log\rho_{ABC}+\beta K_{ABC}\right)\right]$.
Let us determine the Taylor expansion of $\beta\mapsto\log Z(\beta)$
at $\beta=0$. Consider that
\begin{align}
\left.\log Z(\beta)\right|_{\beta=0}  =\log\Tr\!\left[\exp\!\left(\log\rho_{ABC}\right)\right]=\log\Tr\!\left[\rho_{ABC}\right] =0.
\end{align}
Now observe that
\begin{align}
Z(\beta)\left(\frac{d}{d\beta}\log Z(\beta)\right) & =\frac{d}{d\beta}Z(\beta)\\
 & =\frac{d}{d\beta}\Tr\!\left[\exp\!\left(\log\rho_{ABC}+\beta K_{ABC}\right)\right]\\
 & =\Tr\!\left[\exp\!\left(\log\rho_{ABC}+\beta K_{ABC}\right)\frac{d}{d\beta}\left(\log\rho_{ABC}+\beta K_{ABC}\right)\right]\\
 & =\Tr\!\left[\exp\!\left(\log\rho_{ABC}+\beta K_{ABC}\right)K_{ABC}\right],
\end{align}
implying that
\begin{align}
\left.\frac{d}{d\beta}\log Z(\beta)\right|_{\beta=0}  =\frac{\Tr\!\left[\exp\!\left(\log\rho_{ABC}\right)K_{ABC}\right]}{\Tr\!\left[\exp\!\left(\log\rho_{ABC}\right)\right]} =\Tr\!\left[\rho_{ABC}K_{ABC}\right] =I(A:B|C)_{\rho}.
\end{align}
Defining
\begin{equation}
\rho_{ABC}(\beta)\coloneqq\frac{\exp\!\left(\log\rho_{ABC}+\beta K_{ABC}\right)}{Z(\beta)},
\end{equation}
now observe that
\begin{align}
 & \frac{d^{2}}{d\beta^{2}}\log Z(\beta)\nonumber \\
 & =\frac{d}{d\beta}\left(\frac{d}{d\beta}\log Z(\beta)\right)\\
 & =\frac{d}{d\beta}\left(Z(\beta)^{-1}\Tr\!\left[\exp\!\left(\log\rho_{ABC}+\beta K_{ABC}\right)K_{ABC}\right]\right)\\
 & =-Z(\beta)^{-2}\frac{d}{d\beta}Z(\beta)\Tr\!\left[\exp\!\left(\log\rho_{ABC}+\beta K_{ABC}\right)K_{ABC}\right]\notag \\
 & \qquad +Z(\beta)^{-1}\frac{d}{d\beta}\Tr\!\left[\exp\!\left(\log\rho_{ABC}+\beta K_{ABC}\right)K_{ABC}\right]\\
 & =-Z(\beta)^{-2}\Tr\!\big[\exp\!\left(\log\rho_{ABC}+\beta K_{ABC}\right)K_{ABC}\big]^{2}\nonumber \\
 & \quad+Z(\beta)^{-1}\Tr\!\bigg[\int_{0}^{1}dse^{s\left(\log\rho_{ABC}+\beta K_{ABC}\right)}\notag \\
 & \qquad\quad\qquad \qquad  \times \left(\frac{d}{d\beta}\left(\log\rho_{ABC}+\beta K_{ABC}\right)\right)e^{\left(1-s\right)\left(\log\rho_{ABC}+\beta K_{ABC}\right)}K_{ABC}\bigg]\\
 & =-\left\langle K_{ABC}\right\rangle _{\rho_{ABC}(\beta)}^{2}+\Tr\!\left[\int_{0}^{1}ds\rho_{ABC}(\beta)^{s}K_{ABC}\rho_{ABC}(\beta)^{1-s}K_{ABC}\right]\\
 & =\int_{0}^{1}ds\,\Tr\!\left[\rho_{ABC}(\beta)^{s}\left(K_{ABC}-\left\langle K_{ABC}\right\rangle _{\rho_{ABC}(\beta)}\right)\rho_{ABC}(\beta)^{1-s}\left(K_{ABC}-\left\langle K_{ABC}\right\rangle _{\rho_{ABC}(\beta)}\right)\right].
\end{align}
Then
\begin{align}
\left.\frac{d^{2}}{d\beta^{2}}\log Z(\beta)\right|_{\beta=0} & =\int_{0}^{1}ds\,\Tr\!\left[\rho_{ABC}^{s}\left(K_{ABC}-\left\langle K_{ABC}\right\rangle _{\rho_{ABC}}\right)\rho_{ABC}^{1-s}\left(K_{ABC}-\left\langle K_{ABC}\right\rangle _{\rho_{ABC}}\right)\right]\\
 & =\mathrm{V}_{\KM}(A:B|C)_\rho.
\end{align}
Thus, we conclude the claimed expansion about $\alpha=1$.
\end{proof}

\section{Comparison of recovery bounds}
\label{app:numerical-recovery}
\label{app:mixed-recovery}

We give explicit examples showing that neither the $D^\star$ recovery
bound nor the measured-relative-entropy recovery bound is stronger in
general. We then compare with the fidelity bounds.

\subsection{The measured-relative-entropy bound can be stronger}

Let $C$ be trivial and consider
\begin{equation}
 |\varphi\rangle_{AB}
 =\sqrt{\frac9{10}}|00\rangle+\sqrt{\frac1{10}}|11\rangle,
 \qquad \rho_{AB}=|\varphi\rangle\langle\varphi|.
\end{equation}
Every rotated Petz recovery map used here prepares $\rho_B$, so the recovered state is
$\rho_A\otimes\rho_B$, independently of $t$. For a pure first argument,
\eqref{eq: keyl rate} gives
\begin{equation}
 D^\star(\rho_{AB}\|\rho_A\otimes\rho_B)
 =-\log\langle\varphi|\rho_A\otimes\rho_B|\varphi\rangle
 =\log\frac{100}{73}.
\end{equation}
Measuring both systems in the computational basis gives
\begin{equation}
 D^{\mathbb M}(\rho_{AB}\|\rho_A\otimes\rho_B)
 \geq -\frac9{10}\log\frac9{10}-\frac1{10}\log\frac1{10}
 >\log\frac{100}{73}.
\end{equation}
Thus, the measured-relative-entropy recovery bound is strictly stronger.
The same measurement also shows that $D^\star$ does not satisfy data
processing.
\subsection{The \texorpdfstring{$D^\star$}{D-star} recovery bound can be stronger}
\label{ex:simpler-keyl-recovery}
Let $A$, $B$, and $C$ be qubits, and consider the fully separable state
\begin{equation}
 \rho_{ABC}
 =\frac12|000\rangle\langle000|
  +\frac12|11+\rangle\langle11+|,
 \qquad |+\rangle=\frac{|0\rangle+|1\rangle}{\sqrt2}.
\end{equation}
Thus, $A$ and $B$ contain the same classical bit, while $C$ is prepared
in $|0\rangle$ or $|+\rangle$. Let $X_A$ be the bit flip on $A$.
A direct calculation of the rotated Petz recovery map gives
\begin{equation}
 \mathcal R^{[t]}_{C\to BC}(\rho_{AC})
 =q(t)\rho_{ABC}+(1-q(t))X_A\rho_{ABC}X_A,
 \qquad
 q(t)=\frac12\left[1+
 \frac{\cos\!\bigl(t\log(1+\sqrt2)\bigr)}{\sqrt2}\right].
\end{equation}
The states $\rho_{ABC}$ and $X_A\rho_{ABC}X_A$ have orthogonal
supports. Hence, the input commutes with every recovered state and
with their average. Both divergences therefore reduce to classical
relative entropy, and
\begin{align}
 &\int_{\mathbb R}\, dt\beta_0(t)
 D^\star\!\left(\rho_{ABC}\middle\|
       \mathcal R^{[t]}_{C\to BC}(\rho_{AC})\right)
 =-\int_{\mathbb R}\, dt\beta_0(t)\log q(t)\notag\\
 &\qquad>-\log\!\left(\int_{\mathbb R}\, dt\beta_0(t)q(t)\right)
 =D^{\mathbb M}\!\left(\rho_{ABC}\middle\|
       \int_{\mathbb R}\, dt\beta_0(t)
       \mathcal R^{[t]}_{C\to BC}(\rho_{AC})\right).
\end{align}
The inequality is strict by Jensen's inequality, since $q(t)$ is
positive and nonconstant, and $\beta_0(t)>0$ for all $t\in\mathbb R$.
\subsection{Strict improvement over the fidelity bounds}

A classical full-rank state already suffices for this comparison. Let
$C$ be trivial and take $\rho_{AB}=\operatorname{diag}(3,1,1,3)/8$ in
the computational basis. Both marginals are maximally mixed, so
every recovered state considered here, as well as their average, equals $I_{AB}/4$.
Since the two states commute,
\begin{align}
 D^\star(\rho_{AB}\|I_{AB}/4)
 &=\frac34\log3-\log2,\\
 -\log F(\rho_{AB},I_{AB}/4)
 &=\log\frac4{2+\sqrt3}
 <\frac34\log3-\log2.
\end{align}
The strict inequality also follows from strict Jensen's inequality
for the corresponding classical distributions. Thus, the $D^\star$
recovery bound can be strictly stronger than both fidelity bounds
in~\eqref{eq: keyl fidelity recovery comparison}, even for full-rank
states.

\section{Counterexample to the DPI for the dual quantity}
\label{app: counter example DPI dual}
We now prove that the quantity $J_\alpha$ generally fails to satisfy the DPI under local channels, so that it is not a conditional mutual information.

\begin{proposition}
    The quantity $J_\alpha$ is not a Rényi conditional mutual information for $\alpha=1/2$. 
\end{proposition}
\begin{proof}
We construct an explicit counterexample. 
Consider the pure state
\begin{align}
|\sigma\rangle_{AEBC}
=
\sum_{a,e,b}
\sqrt{p(a,e,b)}\,
|a,e\rangle_{AE}
|b\rangle_B
|a,e,b\rangle_C .
\end{align}
where $p(a,e,b)$ is defined as
\begin{align}
55\,p(a,e,b)
=
\begin{array}{c|cc}
(a,e)\backslash b & 0 & 1 \\ \hline
(0,0) & 7 & 2 \\
(0,1) & 1 & 3 \\
(1,0) & 1 & 20 \\
(1,1) & 4 & 17
\end{array}.
\end{align}
We show that for this state,
\(J_\alpha\) does not satisfy local data processing under the
channel \(\Tr_E:AE\to A\). Explicitly, it does not satisfy the inequality
\begin{align}
J_\alpha(AE:B|C)_\sigma
\ge
J_\alpha(A:B|C)_{\sigma}.
\end{align}
Indeed, duality identifies these quantities with the classical
quantities $I_{1/2}^{\flat}(AE:B)_p$ and
$I_{1/2}^{\flat}(A:B|E)_p$, respectively. Define
\begin{align}
 q(a,e,b)=p_{AE}(a,e)p_B(b),\quad
 r(a,e,b)=\frac{p_{AE}(a,e)p_{BE}(b,e)}{p_E(e)}.
\end{align}
Consequently, the two values are given exactly by
\begin{align}
 J_{\frac{1}{2}}(AE:B|C)_\sigma=-2\log\sum_{a,e,b}\sqrt{p(a,e,b)q(a,e,b)},\;
 J_{\frac{1}{2}}(A:B|C)_\sigma=-2\log\sum_{a,e,b}\sqrt{p(a,e,b)r(a,e,b)}.
\end{align}
Direct evaluation gives
\begin{align}
J_{\frac{1}{2}}(AE:B|C)_\sigma
\approx 0.08658818
<
0.08772824
\approx
J_{\frac{1}{2}}(A:B|C)_\sigma .
\end{align}
On the other hand, one can verify that the violation disappears at
\(\alpha=1\).
\end{proof}
\end{document}